\documentclass[12pt]{article}

\usepackage{amssymb,amsmath,amsthm,nicefrac,bm}
\usepackage[T1]{fontenc}
\usepackage[ngerman,english]{babel}
\usepackage[utf8]{inputenc}
\usepackage{xcolor}
\usepackage{url}
\usepackage{amsmath}
\usepackage{amsthm}
\usepackage{amssymb}
\usepackage{subcaption}
\usepackage{tikz}
\usetikzlibrary{backgrounds,positioning,calc}

\usepackage{pgfplots}
\pgfplotsset{compat=1.9}
\usepgfplotslibrary{groupplots}
\usepackage{booktabs}
\usepackage[numbers,sort&compress,square]{natbib}
\usepackage{tabularx}
\usepackage{array}
\usepackage{multirow}

\usepackage{dsfont}

\usepackage{diagbox}
\usepackage{cleveref}

\usepackage{etoolbox}
\usepackage{appendix}

\usepackage{mathptmx}
\usepackage{graphicx}
\usepackage{marvosym} 
\usepackage{wasysym}

\usepackage{times}
\usepackage{soul}
\usepackage{graphicx}
\usepackage{mathtools}

\usepackage{algorithm}
\usepackage{algorithmic}
\usepackage{enumitem}
\usepackage{mdframed}

\newtheorem{definition}{Definition}
\newtheorem{example}{Example}
\newtheorem{theorem}{Theorem}
\newtheorem{lemma}{Lemma}

\newtheorem{proposition}{Proposition}
\newtheorem{claim}{Claim}

\newcommand{\advantageGaining}[1]{{#1}\dash{}\textsc{AB}}
\newcommand{\advantageGainingnp}{\textsc{AB}}
\newcommand{\multiBribery}[1]{{#1}\dash{}\textsc{MAB}}
\newcommand{\pluralityAG}[1]{{#1}\dash{}\textsc{AWB}}
\newcommand{\pluralityAGnp}{\textsc{AWB}}

\newcommand{\AGnp}{\advantageGainingnp}
\newcommand{\multiPluralityAG}[1]{{#1}\dash{}\textsc{MAWB}}
\newcommand{\multiBriberyAG}[1]{{#1}\dash{}\textsc{MAWB}}

\newcommand{\quota}{{{\mathrm{lqu}}}}
\newcommand{\remainder}{{{\mathrm{rem}}}}

\newcommand{\sainte}{Sainte Lagu\"e}

\newcommand{\orderedlistingof}[2]{\ensuremath{#1_1, #1_2, \ldots, #1_#2}}
\newcommand{\orderedsetof}[2]{\ensuremath{\{\orderedlistingof{#1}{#2}\}}}

\newcommand{\namedorderedsetof}[3]{\ensuremath{{#1}=\orderedsetof{#2}{#3}}}

\newcommand{\dash}{\hbox{-}}

\newcommand{\cubicVertexCover}{\textsc{Cubic\dash{}Vertex\dash{}Cover}}
\newcommand{\unarybinpacking}{\textsc{Unary\dash{}Bin\dash{}Packing}}
\newcommand{\unarybinpackingstar}{\textsc{Unary\dash{}Bin\dash{}Packing$^{\star}$}}

\newcommand{\np}{\ensuremath{\mathrm{NP}}}

\newcommand{\fpt}{\ensuremath{\mathrm{FPT}}}                                     
                                       
\newcommand{\wone}{\ensuremath{\mathrm{W[1]}}}                                   
\newcommand{\wonehard}{\wone-hard}
\newcommand{\wonehardness}{\wonehard{}ness}

\newcommand{\p}{\ensuremath{\mathrm{P}}}

\newcommand{\nphard}{\np-hard}                                       

\newcommand{\nphardness}{\nphard{}ness}

\newcommand{\calR}{{\mathcal{R}}}

\newcommand{\parties}{\ensuremath{\mathcal{P}}}
\newcommand{\partiesabove}{\ensuremath{\parties_{\threshold}}}
\newcommand{\votes}{\ensuremath{\mathcal{V}}}

\newcommand{\threshold}{\ensuremath{\tau}}
\newcommand{\supportalloc}{\ensuremath{\sigma}}
\newcommand{\seats}{\ensuremath{\kappa}}
\newcommand{\appmethod}{\ensuremath{\mathcal{R}}}
\newcommand{\seatalloc}{\ensuremath{\alpha}}
\newcommand{\desiredseats}{\ensuremath{\ell}}
\newcommand{\budget}{\ensuremath{K}}

\DeclareMathOperator{\fairshare}{fs}

\newcommand{\cost}{{\mathrm{cost}}}

\usepackage{booktabs}
\usepackage{tabularx}
\usepackage{array}
\newcolumntype{Z}[1]{>{\raggedright\let\newline\\\arraybackslash\hspace{0pt}}m{#1}}
\newcolumntype{Y}[1]{>{\centering\let\newline\\\arraybackslash\hspace{0pt}}m{#1}}

\newcolumntype{L}{>{\raggedright\arraybackslash}X}
\newcolumntype{R}{>{\raggedleft\arraybackslash}X}
\newcolumntype{C}{>{\centering\arraybackslash}X}

\usepackage{pgfplots}
\usetikzlibrary{matrix}
\usepgfplotslibrary{groupplots}
\usepgfplotslibrary{fillbetween}
\pgfplotsset{compat=1.16}
\definecolor{col_bribery}{rgb}{0.31, 0.31, 0.33}
\definecolor{col_controladd}{rgb}{0.85, 0.85, 0.86}
\definecolor{col_controldel}{rgb}{0.51, 0.50, 0.52}

\usepackage{algorithm}
\usepackage{algorithmic}
\usepackage{listings}

\newcommand{\EP}[3]{
\begin{center}
{\small 
\begin{tabularx}{1.0\columnwidth}{lX}
\toprule
\multicolumn{2}{c}{\sc{#1}} \\
\midrule
{\bf Given:}& \parbox[t]{0.84\columnwidth}{#2\vspace*{1mm}} \\
{\bf Question:}& \parbox[t]{0.84\columnwidth}{#3\vspace*{.5mm}} \\ 
\bottomrule
\end{tabularx}
}
\end{center}
}

\usepackage[margin=1.2in]{geometry}

\begin{document}

\title{How to Tamper with a Parliament: Strategic Campaigns in Apportionment Elections\thanks{This paper combines and substantially extends two conference contributions, one of which appeared in the proceedings of the \emph{29th International Joint Conference on Artificial Intelligence} (IJCAI'20)~\cite{bre-fal-fur-kac-lac:c:strategic-campaign-management-in-apportionment-elections} and the other in the proceedings of the \emph{49th International Conference on Current Trends in Theory and Practice of Computer Science} (SOFSEM'24)~\cite{lau-rot-see:c:apportionment-with-thresholds}.\\ The code and data used for our experiments can be found at \protect\url{https://github.com/bredereck/Strategic-Campaigns-in-Apportionment-Elections}.}}

\author{
Robert Bredereck$^1$ \and
Piotr Faliszewski$^2$ \and
Micha\l{} Furdyna$^2$ \and
Andrzej Kaczmarczyk$^2$ \and
Joanna Kaczmarek$^4$ \and
Martin Lackner$^3$ \and
Christian Laußmann$^4$ \and
Jörg Rothe$^{4,*}$ \and
Tessa Seeger$^4$
}

\date{
$^1$TU Clausthal, Clausthal-Zellerfeld, Germany\\
$^2$AGH University, Kraków, Poland\\
$^3$TU Wien, Vienna, Austria\\
$^4$Heinrich-Heine-Universität Düsseldorf, MNF, Düsseldorf, Germany\\[1ex]
$^*$Corresponding author: \texttt{rothe@hhu.de}
}

\maketitle

\begin{abstract}
	In parliamentary elections, parties compete for a limited, typically fixed number of seats.
	Most parliaments are assembled using apportionment methods that distribute the seats based on the parties' vote counts.
	Common apportionment methods include divisor sequence methods (like D'Hondt or Sainte-Lagu\"e), the largest-remainder method, and first-past-the-post.
	In many countries, an electoral threshold is implemented to prevent very small parties from entering the parliament.
	Further, several countries have apportionment systems that incorporate multiple districts.
	We study how computationally hard it is to change the election outcome (i.e., to increase or limit the influence of a distinguished party) by convincing a limited number of voters to change their vote.
        We refer to these bribery-style attacks as \emph{strategic campaigns} and study the corresponding problems in terms of their computational (both classical and parameterized) complexity.
    We also run extensive experiments on real-world election data and study the effectiveness of optimal campaigns, in particular as opposed to using heuristic bribing strategies and with respect to the influence of the threshold and the influence of the number of districts.
	For apportionment elections with threshold, finally, we propose---as an alternative to the standard top-choice mode---the second-chance mode where voters of parties below the threshold receive a second chance to vote for another party, and we establish computational complexity results also in this setting.
\end{abstract}


\section{Introduction}

Apportionment elections (also called parliamentary or party-list elections)  \citep{bal-you:b:polsci:fair-representation,puk:b-2nd-edition:proportional-representation} are among the most common ways of electing parliaments, used, e.g., in Austria, New Zealand, Poland, Spain, Turkey, and
many other countries.  
In such elections, voters cast their ballots according to their preferences, i.e., depending on the ballot form, they either simply name the party they support most or cast a ranking of the parties. 
The vote counts are then tallied, and---based on this data---an \emph{apportionment method} is used to assign a specific number of seats to each party.\footnote{Note that there are also procedures that specify which particular party members enter the parliament, but we will disregard this issue.}
After the seats have been allocated, the parliament members
propose and discuss bills and then vote on them on behalf of their voters.
Therefore, it is important that the parliament represents the voters as proportionally as possible.

The easiest way, without a doubt, would be to give a party with $\nicefrac{1}{x}$ of the total number of votes exactly $\nicefrac{1}{x}$ of the total number of seats in parliament.
However, since this is usually not an integer (and seats cannot be shared nor split), apportionment methods have to be more sophisticated.
Various apportionment methods that each aim at representativeness have been proposed.
Most notably, the D'Hondt method (a.k.a.\ the Jefferson method), the Sainte-Lagu\"e method (a.k.a.\ the Webster method), and the largest-remainder method (a.k.a.\ the Hamilton method) have been proposed and widely used (see, e.g., the classical textbook by Balinski and Young~\cite{bal-you:b:polsci:fair-representation}).
We formally describe these methods in Section~\ref{sec:apportionment-procedures} and illustrate them by examples.

In many countries, the basic apportionment procedure is extended by taking a so-called \emph{legal electoral threshold} (or simply, a \emph{threshold}, for short) into account.
A threshold specifies the minimum number of votes (or fraction of all votes) a party must receive before any seat is allocated to it at all.
For instance, in Germany, Poland, and Scotland a party must receive at least $5\%$ of the total vote count to participate in the apportionment process.
Electoral thresholds are important for the government to quickly form and allow for effective decision-making by minimizing the effects of fragmentation of the parliament, i.e., by reducing the number of parties in it (see the work of Pellicer and Wegner~\cite{pel-weg:j:effects-of-legal-thresholds} for a study of how mechanical and psychological effects reduce fragmentation).
Undoubtedly, with fewer parties in the parliament compromises can be made more efficiently.
However, a disadvantage of the threshold is that voters supporting a party that did not make it above the threshold are not represented in the parliament at all because their votes are simply ignored.
For example, more than 19\% of the votes in the French election of the European Parliament in 2019 were lost due to a threshold of 5\%, i.e., \emph{``of five votes, just four become effective, and one is discarded as ineffective''} \cite[p.~30]{oel-puk:study:european-elections-of-may-2019}.

Further, some countries, such as Poland, are partitioned into districts and hold separate apportionment elections in each; the results of these elections are then summed up. Other countries, such as Austria, have a single, nation-wide district.\footnote{Although Austria has multiple electoral districts, for the final seat apportionment the country is treated as a single district.}

In this work, we study so-called \emph{strategic campaigns}.
In such scenarios, an external agent intends to change the election outcome in their favor by bribing voters within a certain budget to change their vote.
That is, an external agent seeks to change a limited number of votes in order to either ensure a party they support receives at least a desired number of seats in the parliament (constructive case), or to limit the influence of a party they despise by ensuring it receives no more than a desired number of seats (destructive case).
Strategic campaigns attract increasing attention in political elections, e.g., because today's possibilities to process enormous amounts of data from social networks, search engines, etc.\ make it possible to predict the voting behavior of individuals and to target them highly accurately with individualized (political) advertising.

Traditionally, such problems---where voter preferences are modified by an external agent at some cost and within a given budget---fall into the category of \emph{bribery} problems~\citep{fal-hem-hem:j:bribery,fal-rot:b:handbook-comsoc-control-and-bribery,bau-rot:b-2nd-edition:economics-and-computation-preference-aggregation-by-voting}, but they also have other, more benign and perhaps more useful, interpretations.
We mention the following two (already taken, e.g., by Elkind and Faliszewski~\cite{elkind2010shiftbribery}, Faliszewski et al.~\cite{fal-sko-tal:c:bribery-success}, Xia~\cite{xia:c:margin-of-victory}, Bredereck et al.~\cite{bre-fal-kac-nie-sko-tal:j:mw-robustness}, and Baumeister et al.~\cite{bau-fal-lan-rot:c:campaigns-for-lazy-voters}): 
First, consider a political campaign preceding a given election.
Based on poll data, leaders of each party may wish to know which voters they should try to convince to vote for their party, i.e.,
they may wish to find the most effective way of obtaining additional seats.	
Second, after an election, one may wish to check how close the parties actually are to obtaining additional seats.
If this number is small enough, it may be reasonable to request a post-election audit.	

While it is clear that strategic campaigns can work in practice and such attempts are already being used in the real world, concrete knowledge about how much the election outcome can be changed and how easy it is to conduct a suitable campaign is rather limited.
We contribute to this understanding by studying the computational complexity of such strategic campaigns, i.e., of deciding whether a successful campaign exists, and of finding optimal campaigns.

\paragraph{Related Work}
This paper merges and extends two conference contributions, one by Bredereck et al.~\cite{bre-fal-fur-kac-lac:c:strategic-campaign-management-in-apportionment-elections} and the other by Laußmann et al.~\cite{lau-rot-see:c:apportionment-with-thresholds}. 
Apportionment methods have been deeply studied in mathematics and political science (see, e.g., the classical texts of Balinski and Young~\cite{bal-you:j:quota-method-of-apportionment,bal-you:b:polsci:fair-representation} and Pukelsheim~\cite{puk:b-2nd-edition:proportional-representation}),
but much less so in computer science and computational social choice---a field at the interface of theoretical computer science and artificial intelligence on the one hand and economics and social choice theory on the other hand that studies computational properties of elections and has applications, for example, in multiagent systems.  

Within computational social choice, the research line on \emph{bribery}, a.k.a.\ \emph{strategic campaigning}, was initiated by Faliszewski et al.~\cite{fal-hem-hem:j:bribery,fal-hem-hem-rot:j:llull-copeland-full-techreport}.
Bribery is linked to both \emph{manipulation} of elections by strategic voters and to \emph{electoral control}.
Manipulation was introduced and studied by Bartholdi et al.~\cite{bar-tov-tri:j:manipulating} for unweighted elections and a single manipulator in the constructive case (where the goal is to make a distinguished candidate win), and by Conitzer et al.~\cite{con-san-lan:j:when-hard-to-manipulate} for weighted, coalitional manipulation in both the constructive and the destructive case (where a coalition of strategic voters aims at preventing the victory of a despised candidate in a weighted election).
For more background on manipulation, we refer to the book chapters by Conitzer and Walsh~\cite{con-wal:b:handbook-comsoc-manipulation} and Baumeister and Rothe~\cite{bau-rot:b-2nd-edition:economics-and-computation-preference-aggregation-by-voting}.
Electoral control was introduced and studied by Bartholdi et al.~\cite{bar-tov-tri:j:control} in the constructive case and by Hemaspaandra et al.~\cite{hem-hem-rot:j:destructive-control} in the destructive case.
In this scenario, an election chair changes the structure of an election by actions such as adding, deleting, or partitioning either voters or candidates so as to reach the constructive or destructive control goal.
Bribery and control have been studied for a wide range of voting rules, as surveyed by Faliszewski and Rothe~\cite{fal-rot:b:handbook-comsoc-control-and-bribery} and Baumeister and Rothe~\cite{bau-rot:b-2nd-edition:economics-and-computation-preference-aggregation-by-voting} in their book chapters.

While both bribery and control have been mainly investigated for single-winner and multiwinner voting rules, there is not so much literature yet regarding strategic campaigns in \emph{apportionment elections}, even though these are arguably the most important elections in many countries.
Nonetheless, apportionment can be viewed as a special case of approval-based multiwinner voting~\cite{bri-las-sko:j:approval-as-apportionment,lac-sko:c:approval-based-multiwinner-rules-and-strategic-voting}, especially so when using no electoral threshold. 
As a consequence, for single-district settings with few parties only, one could use the $\fpt$ bribery algorithms of Faliszewski et al.~\cite{fal-sko-tal:c:bribery-success}.
Yet, our algorithms are faster and more direct (for other ways of manipulating approval-based multiwinner rules, see, e.g., the work of Yang~\cite{yan:c:multiwinner-control}, Peters~\cite{peters2018proportionality}, and Lackner and Skowron~\cite{lac-sko:c:approval-based-multiwinner-rules-and-strategic-voting}).

The work of G\"uney~\cite{gue:c:mixed-integer-linear-program-for-election-campaign-optimization} is most closely related to ours: The author studies a campaign management problem similar to ours (specifically for the D'Hondt rule), but instead of providing polynomial-time algorithms, he models it as a mixed integer linear program.
The paper shows that Turkish parliamentary elections (which are based on a multi-district system) are susceptible to strategic campaign management.
Finally, Ostapenko et al.~\cite{ostapenko2012mathematical} developed a loosely related game-theoretic model of campaign management.

\paragraph{Our Contribution}
We propose efficient algorithms for (variations of) the following problem both in the single- and the multi-district setting: Given an apportionment election consisting of a set of parties and a list of votes, a threshold, a number of seats, a distinguished party, and a desired number of seats, what is the smallest number of votes (i.e., cost) that need to be moved
to ensure that the distinguished party obtains at least (in the constructive case)---or at most (in the destructive case)---a desired number of seats.
To formulate this as a decision problem, we are also given a budget that must not be exceeded by the cost spent to reach this goal.
We also consider the problem where we ask for the lowest cost of ensuring that the distinguished party has more seats than any other one, i.e., that this party is a \emph{winner} of the election (in the constructive case), or that this party is not the strongest party (in the destructive case).
Indeed, sometimes winning an election this way may be better for a party than just obtaining more seats when the strongest opponent has even more.
We focus on the family of divisor apportionment methods (which includes
the well-known D'Hondt and \sainte{} methods), and on the largest-remainder method.
As a border case, we also study the first-past-the-post method.

We find that for the single-district setting, all our problems can be solved in polynomial time.
Depending on the apportionment method, these algorithms range from straightforward iterative procedures to involved dynamic programming.
For the multi-district case, the complexity spectrum is more interesting: While the single-district algorithms can be used to obtain at least a given number of seats for the distinguished party in the constructive case, the problem of getting the largest number of seats is $\np$-hard for all methods. 
In the destructive case, we show similar results with the exception that under first-past-the-post, the problem of ensuring that the distinguished party does not receive the largest number of seats is computable in polynomial time.
We further strengthen these results by proving W[1]-hardness for parameterization by the number of districts.

We complement our study by performing a series of experiments on real-world
election data. 
First, we ask how useful it is to use our optimal algorithms, as opposed to using simple heuristic bribing strategies.
We find that in most cases, moving votes optimally is (significantly) more effective than simple bribing strategies: Often, they require between $1.2$ and $3$ times as many votes as the optimal strategy, and in some cases even up to over $6.5$ times as many.
In our second experiment, we study the effect of a threshold on the effectiveness of the strategic campaigns. We find that they can exploit the electoral threshold and significantly benefit from it: 
There are some ``spikes'' at thresholds with a party being directly above the threshold where with bribing only $0.25\%$ of the votes one can make the distinguished party gain sometimes $4.4\%$ or even $12.5\%$ of all seats on top in the constructive case. In the destructive case, we observe similar results at thresholds where a party is directly below it.
Third, we ask how the number of districts affects the easiness (in terms of cost) of changing the election result. Our results show that for both the constructive and destructive case, the fewer districts there are, the more vote moves are required to obtain (respectively, take away) a certain number of seats, e.g., up to more than $3.5$ times as many votes need to be changed in the constructive case and up to almost $6$ times as many in the destructive case when reducing the number of districts to a quarter.

As our third contribution, we introduce and study a simple extension of the usual apportionment procedure with electoral threshold: 
Voters who supported a party below the threshold can reuse their vote for one of the remaining parties above the threshold.
These voters thus get a second chance. 
We provide complexity results showing that this modification renders the corresponding problems intractable (i.e., $\np$-hard) both in the single- and the multi-district setting.
\section{Preliminaries}
\label{sec:preliminaries}
For an integer $t$, we write $[t]$ to denote the set $\{1, \dots, t\}$, and by $[t]_0$ we mean the set $\{0\} \cup [t] = \{0, 1, \dots, t\}$.
We use the Iverson bracket notation, i.e., for a logical expression~$P$, we write $[P]$ to mean $1$ if $P$ is true, and to mean $0$ if $P$ is false.

\subsection{Useful Computational Problems}
\label{sec:preliminaries:useful-problems}

To show various computational hardness results in our subsequent analysis of the
computational complexity of problems related to apportionment procedures, we will use three fundamental decision problems to reduce from.
The first one is \textsc{Cubic\dash{}Vertex\dash{}Cover}, which is a simplified variant of \textsc{Vertex\dash{}Cover} and was shown to be \np-complete by Garey~et~al.~\cite{gar-joh-sto:j:simplified-np-completeness}.

\EP{Cubic-Vertex-Cover}
   {A graph $G=(V,E)$, where each vertex has exactly three neighbors, and a
    positive integer~$k$.}
   {Is there a vertex cover of size at most~$k$ in~$G$, i.e., is it possible to choose at most $k$ vertices in $G$ such that each edge of $G$ is incident to at least one of them?}

\newcommand{\universe}{\ensuremath{U}}
\newcommand{\universeSize}{\ensuremath{n}}
\newcommand{\element}{\ensuremath{u}}
\newcommand{\binSize}{\ensuremath{b}}
\newcommand{\binsCount}{\ensuremath{k}}
In several proofs, we will use
a variant of \textsc{Unary-Bin-Packing}, which itself is the variant of
the classical \textsc{Bin-Packing} problem in which all numbers are encoded in
unary. Notably, \textsc{Unary-Bin-Packing} is not only
\np-hard~\cite{gar-joh:b:int} but also~\wonehard{} when
parameterized by the number of
bins~\cite{jan-kra-mar-sch:j:bin-packing-revisited}.
Specifically, we define the variant \unarybinpackingstar{} of this
problem, where the star in the problem name indicates that our definition
of the problem slightly differs from the classical one.

	\EP{\unarybinpackingstar}
  {A bin size~\binSize{}, a
    collection \namedorderedsetof{\universe}{\element}{\universeSize}
    of $\universeSize \geq 3$
    even positive integers encoded in unary
    such that $u_i \leq b$ for all~$i \in [\universeSize]$, and the
    number $\binsCount{} \geq 3$ of bins such that $\sum_{i=1}^{\universeSize}
  \element_i=\binsCount{}\binSize{}$.}
	{Is it possible to fit all elements from~\universe{} in these~\binsCount{}~bins not exceeding the size of any bin?} 

        Our definition differs from the classical one by assuming the
        evenness of the 
elements of~$\universe$, the lower bound of~$3$ on the numbers~$\universeSize$
and~$\binsCount$, and that all elements of~$\universe$ exactly fit in all bins.
We make these (standard) assumptions without loss of generality, in order to simplify our proofs and their exposition, and it is easy to see that the slightly simplified problem variant remains \nphard\ and \wonehard.

Finally, we define the well-known \textsc{Hitting-Set} problem, which was shown to be \nphard\ by Karp~\cite{kar:b:reducibilities}. 

\EP{Hitting-Set}
{A set $U = \{u_1, \dots, u_p\}$, a collection $S = \{S_1, \dots, S_q\}$ of nonempty subsets of $U$, and an integer~$\budget$, $1 \leq \budget \leq \min\{p,q\}$.}
{Is there a hitting set $U' \subseteq U$, $|U'| \leq \budget$, i.e., a set $U'$ of size at most $\budget$ such that $U' \cap S_i \neq \emptyset$ for each $S_i \in S$?}

\subsection{Apportionment Procedures}
\label{sec:apportionment-procedures}

An election $E = (\parties, \votes)$ consists of a set of $m$ \emph{parties} $\parties = \{P_1, \ldots, P_m\}$, and a list of $n$ \emph{votes} $\votes$ over the parties in~$\parties$.
For a party $X \in \parties$, we denote the set $\parties \setminus \{X\}$ by~$\parties_{-X}$.
Each vote in $\votes$ is a strict ranking of the parties from most to least preferred, and we write $A \succ_v B$ if voter $v$ prefers party $A$ to $B$ (where we omit the subscript $v$ when it is clear from the context).
We sometimes refer to the most preferred party of a voter~$v$ as $v$'s \emph{top choice}.
From an election we can derive an \emph{apportionment problem} $I = (\supportalloc, \threshold, \seats)$ consisting of a \emph{support allocation} $\supportalloc : \parties \to \mathbb{N}$, a \emph{threshold} $\threshold \in \mathbb{N} = \{0, 1, 2, \ldots\}$, and the \emph{seat count} $\seats \in \mathbb{N}$.
The support allocation maps a party to the numerical value of support it receives in the election. 
Let
$(p_1, ..., p_m)$ be the vote distribution, i.e., $p_i$ gives the number of votes for which $P_i \in \parties$ is the top choice. 
Unless specified otherwise, we assume that the support of a party $P_i$ is
\[
\supportalloc(P_i) = \begin{cases}
	p_i \ \ \ \text{if }p_i \geq \threshold \\
	0 \ \ \ \ \text{otherwise,}
\end{cases}
\]
i.e., the number of votes for which
$P_i$ is the top choice if
$P_i$ receives at least $\threshold$ top choices, and~$0$ otherwise.
That is, votes for parties that receive
fewer than $\threshold$ top choices are ignored, and the voters have no opportunity to change their vote.\footnote{Note that also other support functions are possible, e.g., support functions using the Borda count.
  However, one would never do so.
  The problem with Borda is that if everyone has preference order $P_1 \succ P_2 \succ \cdots \succ P_m$ (i.e., everyone is absolutely sure that $P_1$ is the best party), this party would still get only very few seats.
  Indeed, its number of seats is proportional to $m-1$, while the total number of seats that the other parties get is proportional to $O(m^2)$.
  One could say that it is unfair to compare one party against all the other ones, but we are talking about the completely unanimous situation here. So it is quite paradoxical that the other parties can form a successful coalition against the clear, unanimous winner.}
This ensures that parties that do not reach the threshold are not assigned seats by our apportionment methods.
We refer to this setting as the \emph{top-choice mode}.
An alternative mode (which we call the \emph{second-chance mode}) will be proposed in Section~\ref{sec:second-chance}.

\begin{example}[support allocation in the top-choice mode]\label{ex:supportallocation}
	Consider the election \[E=(\{P_1, \dots, P_5\}, \{v_1, \dots, v_{2052}\})\] with the votes
	\begin{align*}
		v_1, \dots, v_{604} :&\quad P_1 \succ P_2 \succ P_3 \succ P_5 \succ P_4\\
		v_{605}, \dots, v_{819} :&\quad P_2 \succ P_1 \succ P_5 \succ P_4 \succ P_3\\
		v_{820}, \dots, v_{1174} :&\quad P_3 \succ P_2 \succ P_1 \succ P_4 \succ P_5\\
		v_{1175}, \dots, v_{1474}  :&\quad P_1 \succ P_5 \succ P_3 \succ P_2 \succ P_4\\
		v_{1475}, \dots, v_{1652} :&\quad P_4 \succ P_2 \succ P_3 \succ P_1 \succ P_5\\
		v_{1653}, \dots, v_{1800} :&\quad P_2 \succ P_4 \succ P_3 \succ P_1 \succ P_5\\
		v_{1801}, \dots, v_{2000} :&\quad P_1 \succ P_3 \succ P_4 \succ P_5 \succ P_2\\
		v_{2001}, \dots, v_{2052} :&\quad P_5 \succ P_2 \succ P_3 \succ P_4 \succ P_1
	\end{align*}
	and the apportionment problem $I = (\supportalloc, 100, 6)$, i.e., we have a threshold of $100$ and want to distribute $6$ seats.
	In total, $P_1$ is ranked in the first position in $1104$ votes, $P_2$ in $363$, $P_3$ in $355$, $P_4$ in $178$, and $P_5$ has
        $52$ top choices.
        As $P_5$ does not reach the threshold, its support is set to zero.
        All other parties reach the threshold and thus their support is simply the number of their top choices.
        Overall, in the top-choice mode, this results in the following support allocation for these five parties:
	\begin{align*}
		\supportalloc(P_1) &= 1104,\\
		\supportalloc(P_2) &= 363, \\
		\supportalloc(P_3) &= 355,\\
		\supportalloc(P_4) &= 178,\\
		\supportalloc(P_5) &= 0.
	\end{align*}
\end{example}

We make the very natural assumption that we have more votes than we have parties and more votes than we have seats.
Note that in reality an electoral threshold is usually given as a relative threshold in percent (e.g., a $5\%$ threshold).
However, we can easily convert such a relative threshold into an absolute threshold, as required by our definition.

Given an apportionment problem, we can determine the \emph{seat allocation} by employing an \emph{apportionment method}.
The seat allocation $\seatalloc: \parties \to \{0, \dots, \seats\}$ must satisfy $\sum_{A \in \parties} \seatalloc(A) = \seats$, i.e., exactly $\seats$ seats are allocated to parties.
Note that
if no party reaches the threshold, the seat allocation is undefined, and in case of ties, an apportionment method may output more than one apportionment (in real-world applications and in our computational problems, a specific tie-breaking rule will be assumed, so our apportionment rules in fact are resolute).

In this paper, we mostly consider \emph{divisor sequence methods} (with a focus on the \emph{D'Hondt} and \emph{\sainte{}} methods), the \emph{largest-remainder} method, and the \emph{first-past-the-post} method. 
\paragraph{Divisor sequence methods}
A divisor sequence method is defined by a (finite) sequence $d = (d_1, d_2, \dots, d_{\seats}) \in \mathbb{R}^{\seats}$ with $d_i < d_j$ for all $i,j \in \{1,\dots,\seats\}$ with $i<j$, and $d_1 \geq 1$.
We require that each value
$d_i$ is computable in polynomial time with respect to $i$ and is polynomially bounded in $i$ (indeed, the standard divisor sequence methods use very simple divisor sequences).
For each party~$P \in \parties$, we compute its \emph{fraction list}
$\left[\frac{\supportalloc(P)}{d_1},  \frac{\supportalloc(P)}{d_2}, \dots, \frac{\supportalloc(P)}{d_{\seats}}\right]$.
Then we go through the fraction lists of all parties to find the
$\seats$ highest values. In case of a tie, there exists more than one correct seat assignment.
Each party receives one seat for each of its list values among the $\seats$ highest values.

The \emph{D'Hondt method} (also known as the \emph{Jefferson method}) is defined by the divisor sequence $1,2,3,\dots$; the \emph{\sainte{} method} (also known as the \emph{Webster method}) is defined by the divisor sequence $1,3,5,\dots$. There are also further divisor sequence methods, such as Huntington--Hill, Adams, and Dean~\citep{bal-you:b:polsci:fair-representation}, but we do not further consider them in this paper.

\begin{example}[D'Hondt]\label{ex:dhondt}
	 Consider the election and apportionment problem of Example~\ref{ex:supportallocation}. %
	Then, by dividing the support allocation of each party by the divisor sequence $1,2,3,\dots$, the resulting D'Hondt fraction lists are
	\[
	\begin{array}{r@{\hspace*{2.1mm}}c@{\hspace*{2.1mm}}l@{\hspace*{2.1mm}}r@{\hspace*{2.1mm}}r@{\hspace*{2.1mm}}r@{\hspace*{2.1mm}}r@{\hspace*{2.1mm}}r}
		\text{$P_1$} & :  & [\mathbf{1104}, & \mathbf{552},\phantom{.5} & \mathbf{368},\phantom{.3} & \mathbf{276},\phantom{.5} & 220.8, & 184\phantom{.5}],\\
		\text{$P_2$} & :  & [\phantom{\mathbf{1}}\mathbf{363}, & 181.5, & 121,\phantom{.3} & 90.8, & 72.6, & 60.5],\\
		\text{$P_3$} & :  & [\phantom{\mathbf{1}}\mathbf{355}, & 177.5, & 118.3, & 88.8, & 71,\phantom{.8} & 59.2],\\
		\text{$P_4$} & :  & [\phantom{\mathbf{1}}178, & \phantom{1} 89,\phantom{.5} & 59.3, & 44.5, & 35.6, & 29.7],\\
		\text{$P_5$} & :  & [\phantom{110}0, & \phantom{17} 0,\phantom{.5} & 0,\phantom{.3} & 0,\phantom{.5} & 0,\phantom{.8} & 0\phantom{.5}].
	\end{array}
	\]
	The $\seats = 6$ highest values among the fraction lists of all parties are highlighted in boldface.
	Party~$P_1$ thus receives four seats, parties~$P_2$ and~$P_3$ receive one seat each, and parties~$P_4$ and $P_5$ receive no seats at all.
\end{example}

Note that it is possible that a party does not receive any seats in a parliament when an allocation procedure such as D'Hondt or Sainte-Lagu\"e is applied, even if it receives enough votes to exceed the threshold, i.e., we only have the implication stating that if a party does not get enough votes to pass the threshold, it will not get any seats.
However, the reverse implication does not apply.
This can also be seen in the previous example, where party~$P_4$ has support greater than zero, i.e., exceeds the threshold, but it does not receive any seats when D'Hondt is applied.

\paragraph{Largest-remainder method (LRM)} 
The \emph{largest-remainder method} (also known as the \emph{Hamilton method}) is based directly on the fair share $\fairshare(P_i) = \seats\cdot \frac{\supportalloc(P_i)}{n}$, where $n = \sum_{P_i \in \parties} \supportalloc(P_i)$ (i.e., in the top-choice mode it is simply the number of votes for parties above the threshold).
Intuitively, the fair share is the number of seats the party would receive in a perfectly proportional parliament.
Unfortunately, this number is usually not an integer, so we now have to round the fair shares of the parties.
Therefore, each party $P_i$ is assigned at least
$\quota(\supportalloc(P_i),n,\seats) = \lfloor \seats \cdot \frac{\supportalloc(P_i)}{n}\rfloor$ seats, i.e., the integer part of $\fairshare(P_i)$ (we refer to them as \emph{lower quota seats}).
The remaining seats are distributed to the parties with the largest remainder values
$\remainder(\supportalloc(P_i),n,\seats) =
\seats \cdot \frac{\supportalloc(P_i)}{n} - \lfloor\seats \cdot \frac{\supportalloc(P_i)}{n}\rfloor$ (if there are $r$ seats left to be assigned, then $r$ parties with the highest remainder values get one seat each; we refer to them as \emph{remainder seats}).\footnote{It is possible (though not very likely) that the remainder value of enough parties is exactly zero, so that a party below the threshold could receive one of the remainder seats due to tie-breaking.
  We therefore exclude all parties below the threshold from the distribution of remainder seats.}

\begin{example}[LRM]\label{ex:lrm}
	Consider the election and apportionment problem of Example~\ref{ex:supportallocation}. 
	Recall that the support allocation is $\supportalloc(P_1) = 1104$,  $\supportalloc(P_2) = 363$, $\supportalloc(P_3) = 355$, $\supportalloc(P_4) = 178$, and $\supportalloc(P_5) = 0$.
	Recall further that $n=2000$.
	First, we calculate the parties' fair share that is (rounded to two decimal places): $6 \cdot \frac{1104}{2000} = 3.31$ seats for~$P_1$, $6 \cdot \frac{363}{2000} = 1.09$ seats for~$P_2$, $6 \cdot \frac{355}{2000} = 1.07$ seats for~$P_3$, $6 \cdot \frac{178}{2000} = 0.53$ seats for~$P_4$, and $6 \cdot \frac{0}{2000} = 0$ seats for~$P_5$.
	By lower quota, $P_1$ receives three seats while $P_2$ and $P_3$ receive one seat each.
	One seat is left and since party $P_4$ has the largest remainder of $0.53$, this seat is allocated to~$P_4$.
\end{example}

\paragraph{First-past-the-post (FPTP)}
Finally, in a given district, \emph{first-past-the-post} simply gives all seats to the strongest party, i.e., the party with the highest support. 
The FPTP method is not an apportionment method (as it is not proportional). However, we consider it as a corner case and due to its use in US presidential elections, where each state holds its own election and the winner receives all the electoral votes (seats, in our language), and in the UK for electing Members of Parliament to the House of Commons. 
Note that, since we assume there to be at least one party to reach the threshold, there is no difference in the outcome under FPTP whether or not there exists a threshold $\threshold > 0$.

\begin{example}[FPTP]\label{ex:fptp}
	Again, consider the election and apportionment problem of Example~\ref{ex:supportallocation}. Since $P_1$ has the highest support, i.e., is the strongest party, all $\seats= 6$ seats are assigned to $P_1$.
\end{example}

\subsection{Strategic Campaigns and Apportionment Bribery Problems}

Now we define strategic campaigns, modeled as bribery scenarios.
In the following, we assume that apportionment methods are resolute, i.e., they employ a tie-breaking scheme and return a single apportionment.  Specifically, we consider lexicographic tie-breaking (unless specified otherwise), where ties are broken in favor of the party with the lowest index (this simplifies the formulation of some of our algorithms; our
results can easily be adapted to other tie-breaking orders).

Let $\appmethod$ be an apportionment method.
In our (constructive) bribery problem, we are given an election~$E$, a threshold~$\threshold$,\footnote{If there is no threshold for a given apportionment election, $\threshold$ can simply be set to zero.}
a seat count~$\seats$, a budget~$\budget$, and a desired number of seats~$\desiredseats$ for a distinguished party~$P^*$.
By bribing at most $\budget$ voters to change their vote in our favor (i.e., we can alter their votes as we like), we seek to ensure that party
$P^*$ receives at least $\desiredseats$ seats.

\EP{$\appmethod$-Apportionment-Bribery ($\appmethod$-AB)}
{An election $E=(\parties, \votes)$, a threshold $\threshold$, a total number of seats~$\seats$, a distinguished party $P^* \in \parties$, a desired number of seats~$\desiredseats$, $1 \leq \desiredseats \leq \seats$, and a budget~$\budget$, $0 \leq \budget \leq |\votes|$.}
{Is there a successful campaign, i.e., is it possible to make $P^*$ receive at least $\desiredseats$ seats using apportionment method $\appmethod$ by changing at most $\budget$ votes in~$\votes$?}

The destructive variant of the problem
is defined analogously:

\EP{$\appmethod$-Destructive-Apportionment-Bribery ($\appmethod$-DAB)}
   {An election $E=(\parties, \votes)$, a threshold $\threshold$, a total number of seats~$\seats$, a distinguished party $P^* \in \parties$, a desired number of seats~$\desiredseats$,
     $0 \leq \desiredseats \leq \seats$, and a budget~$\budget$, $0 \leq \budget \leq |\votes|$.}
{Is there a successful destructive campaign, i.e., is it possible to make $P^*$ receive at most $\desiredseats$ seats using apportionment method $\appmethod$ by changing at most $\budget$ votes in~$\votes$?}

Note that we require $0 \leq \desiredseats \leq \seats$ in the destructive case, so the distinguished party may also be not allowed to receive any seat at all.
Moreover, since
$\desiredseats \leq \seats \leq |\votes|$ and $\budget \leq |\votes|$, for the complexity analysis it does not matter whether $\seats$, $\desiredseats$, and $\budget$ are encoded in binary or unary. 
However, we assume that the number of voters is encoded in unary because we primarily focus on parliamentary elections where each vote corresponds to a single person casting it.
Thus the number of votes is of 
reasonable size (typically, no more than around~$10^9$).
Additionally, each vote is physically stored as a separate ballot.

The just-defined problems can be slightly modified, to ask for a vote distribution (within the specified bribing budget) so that the distinguished party receives more seats than each of the other parties, i.e., wins the election. 
Note that in contrast to the previously defined problems, the input for these problems therefore does not contain a desired number of seats:

\EP{$\appmethod$-Apportionment-Winner-Bribery ($\appmethod$-AWB)}
{An election $E=(\parties, \votes)$, a threshold $\threshold$, a total number of seats~$\seats$, a distinguished party $P^* \in \parties$, and a budget~$\budget$, $0 \leq \budget \leq |\votes|$.}
{Is there a successful winner campaign, i.e., is it possible to ensure $\seatalloc(P^*)>\seatalloc(P_i)$ for all $P_i \in \parties_{-P^*}$ using apportionment method $\appmethod$ by changing at most $\budget$ votes in~$\votes$?}

The destructive variant
of this problem
is defined analogously by asking if it is possible to ensure that $P^*$ does \emph{not} win the election:

\EP{$\appmethod$-Destructive-Apportionment-Winner-Bribery ($\appmethod$-DAWB)}
{An election $E=(\parties, \votes)$, a threshold $\threshold$, a total number of seats~$\seats$, a distinguished party $P^* \in \parties$, and a budget~$\budget$, $0 \leq \budget \leq |\votes|$.}
{Is there a
  successful campaign to preclude the victory of~$P^*$, i.e., is it possible to ensure that there exists at least one party $P_i$ with $\seatalloc(P^*)<\seatalloc(P_i)$ using apportionment method $\appmethod$ by changing at most $\budget$ votes in~$\votes$?}

Some countries, such as Poland, are partitioned into districts and hold separate apportionment elections in each.
The results of these elections are then summed up.
We model
the problem variant for multiple districts as follows.

\EP{$\appmethod$-Multidistrict-Apportionment-Bribery ($\appmethod$-MAB)}
   {A number of districts $q$, a list of elections $[E_1, E_2, \dots, E_{q}]$ with $E_i=(\parties_i, \votes_i)$ and the corresponding thresholds~$\threshold_i$ and seat counts $\seats_i$, $1 \leq i \leq q$, a distinguished party $P^* \in \parties=\bigcup_{i=1}^{q} \parties_i$, a desired number of seats~$\desiredseats$, $1 \leq \desiredseats \leq \seats =\sum_{i=1}^{q}\seats_i$, and a budget~$\budget$,
     $0 \leq \budget \leq \sum_{i=1}^{q} |\votes_i|$.}
{Is there a successful campaign, i.e., is it possible to make $P^*$ receive at least $\desiredseats$ seats using apportionment method $\appmethod$ in each district individually by changing in total at most $\budget$ votes in $\votes_1, \dots, \votes_{q}$, where the seat count for each party is the sum of seat counts in all districts?}

Again, the destructive variant
is defined analogously.

\EP{$\appmethod$-Destructive-Multidistrict-Apportionment-Bribery ($\appmethod$-DMAB)}
   {A number of districts $q$, a list of elections $[E_1, E_2, \dots, E_{q}]$ with $E_i=(\parties_i, \votes_i)$ and the corresponding thresholds~$\threshold_i$ and seat counts $\seats_i$, $1 \leq i \leq q$, a distinguished party $P^* \in \parties=\bigcup_{i=1}^{q} \parties_i$, a desired number of seats~$\desiredseats$,
     $0 \leq \desiredseats \leq \seats =\sum_{i=1}^{q}\seats_i$, and a budget~$\budget$,
     $0 \leq \budget \leq \sum_{i=1}^{q} |\votes_i|$.}
{Is there a successful
  destructive campaign, i.e., is it possible to make $P^*$ receive at most $\desiredseats$ seats using apportionment method $\appmethod$ in each district individually by changing in total at most $\budget$ votes in $\votes_1, \dots, \votes_{q}$, where the seat count for each party is the sum of seat counts in all districts?}

As before, we also consider the
corresponding winner problem, where we ask for $P^*$ to have more seats than each of the other parties, and its destructive variant.
Again, note that the input for these problems in the winner variant does not contain a desired number of seats.

\EP{$\appmethod$-Multidistrict-Apportionment-Winner-Bribery ($\appmethod$-MAWB)}
   {A number of districts $q$, a list of elections $[E_1, E_2, \dots, E_{q}]$ with $E_i=(\parties_i, \votes_i)$ and the corresponding thresholds~$\threshold_i$ and seat counts $\seats_i$, $1 \leq i \leq q$, a distinguished party $P^* \in \parties=\bigcup_{i=1}^{q} \parties_i$,
     and a budget~$\budget$,
     $0 \leq \budget \leq \sum_{i=1}^{q} |\votes_i|$.}
{Is there a successful
  winner campaign, i.e., is it possible to ensure $\seatalloc(P^*)>\seatalloc(P_i)$ for all $P_i \in \parties_{-P^*}$ using apportionment method $\appmethod$ in each district individually by changing in total at most $\budget$ votes in $\votes_1, \dots, \votes_{q}$, where the seat count for each party is the sum of seat counts in all districts?}

\EP{$\appmethod$-Destructive-Multidistrict-Apportionment-Winner-Bribery ($\appmethod$-DMAWB)}
   {A number of districts $q$, a list of elections $[E_1, E_2, \dots, E_{q}]$ with $E_i=(\parties_i, \votes_i)$ and the corresponding thresholds~$\threshold_i$ and seat counts $\seats_i$, $1 \leq i \leq q$, a distinguished party $P^* \in \parties=\bigcup_{i=1}^{q} \parties_i$,
     and a budget~$\budget$,
     $0 \leq \budget \leq \sum_{i=1}^{q} |\votes_i|$.}
{Is there a 
  successful campaign to preclude the victory of~$P^*$, i.e., is it possible to ensure that there exists at least one party $P_i$ with $\seatalloc(P^*)<\seatalloc(P_i)$
  using apportionment method $\appmethod$ in each district individually by changing in total at most $\budget$ votes in $\votes_1, \dots, \votes_{q}$, where the seat count for each party is the sum of seat counts in all districts?}

\section{Classical Top-Choice Mode}\label{sec:top-choice} 

We now present our results for the classical top-choice mode, starting with the single-district case.

\subsection{Single-District Case}\label{subsec:top-choice-single-district}
	
	We start the discussion by considering the single-district cases,
	i.e., by considering the
        \advantageGaining{$\calR$} and
	\pluralityAG{$\calR$} problems. We find polynomial time algorithms 
	for all our apportionment methods, for these two problems.
	As a warm-up, let us consider the FPTP method, for which a simple
	greedy algorithm suffices for both problems.
	
	\begin{proposition}\label{prop:fptp}
		\advantageGaining{FPTP} and \pluralityAG{FPTP} are both in~$\p$.
	\end{proposition}
		\begin{proof}
		  Under FPTP the party with the highest vote count gets
                  all $\seats$ seats,
			so the same algorithm works for both our problems.
                        Consider the support allocation of a given
                        apportionment problem and 
			let $\budget$ be the bribing budget. Our algorithm executes $\budget$ iterations, where in each iteration we proceed as follows: 
			\begin{enumerate}
				\item We identify the party, say~$P_i$,
				with the currently highest vote count (breaking ties in favor of the
				party with the lowest index; if this is distinguished party~$P^*$, then we terminate and
				accept), and 
				\item we move one vote from party $P_i$ to $P^*$. 
			\end{enumerate}
			If in the end party $P^*$ has the highest number of votes, then we accept.
                        Otherwise, we reject. 
			The correctness and polynomial running time are straightforward to verify.
		\end{proof}

	        A similar approach also works for the destructive variants, where instead of moving votes from the strongest party to $P^*$, we now move them from $P^*$ to the
                currently strongest party in each iteration.
		
	\begin{proposition}\label{prop:fptpdestructive}
		FPTP-\textsc{DAB} and FPTP-\textsc{DAWB} are both in~$\p$.
	\end{proposition}
		\begin{proof}
		  Again, since under FPTP the party with the highest vote count gets
                  all $\seats$ seats, the same algorithm works for both problems. 
                  Let $\supportalloc$ be the support allocation of a given apportionment problem
                  and let $\budget$ be the bribing budget. 
			Our algorithm executes $\budget$ iterations, where in each iteration we proceed as follows: 
			\begin{enumerate}
				\item We identify the party $P_i \in \parties_{-P^*}$ with the currently highest vote count (breaking ties in favor of the party with the lowest index; if $\supportalloc(P_i) > \supportalloc(P^*)$, then we terminate and accept), and 
				\item we move one vote from party $P^*$ to $P_i$. 
			\end{enumerate}
			If in the end party $P^*$ does not have the highest number of votes, then we
			accept.
                        Otherwise, we reject. 
			The correctness and polynomial running time again are straightforward to verify.
		\end{proof}
	
Unfortunately, greedy algorithms do not seem to work for the other apportionment methods and, therefore, we resort to dynamic programming for them.
We present algorithms that are nearly identical for the divisor sequence methods and LRM.\footnote{Note that Bredereck et al.~\cite{bre-fal-fur-kac-lac:c:strategic-campaign-management-in-apportionment-elections} have previously proposed an alternative polynomial-time algorithm that does not handle thresholds.}
		The general structure of our algorithms is that first we add as many
		votes as the bribing budget allows to party~$P^*$, and then we compute
		from which parties these votes have to be taken to guarantee the
		desired number of seats for~$P^*$.

The proof of polynomial-time computability for
$\appmethod$-\textsc{AB}, where $\appmethod$ is a divisor sequence method, relies on the following lemma. 

	 \begin{lemma}\label{lem:moving-voters}
	 	The following two statements hold for all divisor sequence methods.
	 	\begin{enumerate}[label=(\roman*)]
	 		\item The maximum number of additional seats for party $P^*$ by bribing at most $\budget$ votes can always be achieved by moving exactly $\budget$ voters from parties in $\parties_{-P^*}$ to $P^*$. %
	 		\item The maximum number of seats we can remove from party $P^*$ by bribing at most $\budget$ votes can always be achieved by moving exactly $\budget$ voters from $P^*$ to parties in $\parties_{-P^*}$. %
	 	\end{enumerate}
	 \end{lemma}
	 
	 \begin{proof}
           Note that when a party receives additional support, the parties fraction list values increase, while they decrease when the support is decreased.
	 	
	   We start with the proof of the first claim.
	 	Assume we found a way to convince voters to change their vote (which, in the following, is called a \emph{bribery action}) such that $P^*$ receives $X$ additional seats.
	 	\begin{description}%
	 		\item[\normalfont{Case 1:}] There are parties other than $P^*$ that receive additional votes.
	 		Let $P_i \neq P^*$ be such a party. 
	 		Now consider that we move all votes that were moved to $P_i$ to $P^*$ instead.
	 		Clearly, $P^*$'s fraction list values increase while those of $P_i$ decrease.
	 		Thus $P^*$ receives at least as many seats in
                        the modified election resulting from this bribery action as in the original election.
	 		\item[\normalfont{Case 2:}] Votes are only moved from parties in $\parties_{-P^*}$ to~$P^*$.
	 		By moving (some) votes to party $P_i$ instead of to~$P^*$, the fraction list values of $P^*$ would decrease while those of $P_i$ would increase.
	 		Therefore, $P^*$ cannot receive more seats in the modified election resulting from this bribery action than in the original election.
	 	\end{description}
	 	Finally, note that by the monotonicity of divisor sequence methods, moving more votes to $P^*$ never makes $P^*$ lose any seats.
	 	Thus we can spend the whole budget $\budget$ on moving votes to~$P^*$.
	 	This together with the two given cases implies that the best possible number of additional seats can always be achieved by moving $\budget$ votes only from $\parties_{-P^*}$ to $P^*$ (although there might be other solutions that are equally good). 
	 	
	 	To prove the second claim, just swap the roles of $P^*$ and the other parties.
	 \end{proof}
	 
	 Lemma~\ref{lem:moving-voters} is crucial for the correctness of our polynomial-time algorithms because it implies that we should exhaust the whole budget $\budget$ for moving votes from other parties to $P^*$ in the constructive case, and for moving votes from $P^*$ to other parties in the destructive case.
	 That is, we do not need to consider moving votes within $\parties_{-P^*}$.
	 
	 We now introduce our algorithm (Algorithm~\ref{alg:bribery-constructive}) and describe it intuitively. 
	 Then we proceed by showing that the algorithm indeed is correct and
      runs in polynomial time.
      For better readability, without loss of generality, for the remainder of this Section~\ref{subsec:top-choice-single-district}, we assume that $P_1$ is our distinguished party, i.e., $P^* = P_1$ (whenever there exists a distinguished party).
	 
	 \begin{algorithm}[h]
	 	\caption{for $\appmethod$-\textsc{AB}, where $\appmethod$ is a divisor sequence method}
	 	\label{alg:bribery-constructive}
	 	\textbf{Input}: $\parties$, $\votes$, $\threshold$, $\seats$, $P_1$, $\budget$, $\desiredseats$
	 	\begin{algorithmic}[1] 
	 		\STATE $\budget \gets \min\{n-p_1, \budget\}$
	 		\IF{$p_1 + \budget < \threshold$}
	 		\RETURN \textsc{No}
	 		\ENDIF
	 		\STATE $\supportalloc(P_1) \gets p_1 + \budget$
	 		\STATE compute $\gamma$
	 		\COMMENT{\textcolor{gray}{\emph{For each party $P \in \parties$, the entries of array $\gamma[P]$ are pairs $(x,\mathrm{cost})$, where $\mathrm{cost}$ is the bribery budget needed to ensure that $P$ receives exactly $x$ seats before $P_1$ gets $\desiredseats$ seats}}}
	 		\STATE initialize table $\mathrm{tab}$ with $\seats - \desiredseats$ columns and $m$ rows, where $\mathrm{tab}[0][0] \gets 0$ and the other entries are~$\infty$
	 		\STATE let $o:\{1, \dots, |\parties_{-P_1}|\} \to \parties_{-P_1}$ be an ordering
	 		\FOR{$i \gets 1$ to $|\parties_{-P_1}|$}
	 		\FOR{$s \gets 0$ to $\seats - \desiredseats$}
	 		\FOR{$(x, \mathrm{cost}) \in \gamma[o(i)]$}
	 		\IF{$s-x \geq 0$}
	 		\STATE $\mathit{tmp} \gets \mathrm{tab}[i-1][s-x] + \mathrm{cost}$
	 		\IF{$\mathit{tmp} < \mathrm{tab}[i][s]$}
	 		\STATE $\mathrm{tab}[i][s] = \mathit{tmp}$
	 		\ENDIF
	 		\ENDIF
	 		\ENDFOR
	 		\ENDFOR
	 		\ENDFOR
	 		\FOR{$s \gets 0$ to $\seats - \desiredseats$}
	 		\IF{$\mathrm{tab}[|\parties_{-P_1}|][s] \leq \budget$}
	 		\RETURN \textsc{Yes}
	 		\ENDIF
	 		\ENDFOR
	 		\RETURN \textsc{No}
	 	\end{algorithmic}
	 \end{algorithm} 

	 Intuitively, we first set $\budget$ to the minimum of $n-p_1$ and $\budget$ because this is the maximum number of votes we can move from other parties to $P_1$.
	 Should it be impossible for $P_1$ with this
         new $\budget$ to reach the threshold, we can already answer \textsc{No}, as $P_1$ never receives any seat at all.
	 The crucial part of the algorithm is computing the dictionary~$\gamma$.
	 As commented in Algorithm~\ref{alg:bribery-constructive}, $\gamma$ is defined as a list of arrays $\gamma[P]$: For each party $P \in \parties$, the entries of their corresponding array $\gamma[P]$ are (implementable) pairs $(x,\mathrm{cost})$, where $\mathrm{cost}$ is the minimum number of votes that must be removed from party~$P$ so that $P$ receives only $x$ seats before $P_1$ receives the $\desiredseats$-th seat, assuming $P_1$ has exactly $\budget$ additional votes in the end.	 
	 We later show in the proof of Theorem~\ref{thm:bribery-top-choice} how $\gamma$ can be efficiently computed.
	 We now define an order $o$ over the parties
         $\parties_{-P_1}$.
	 This can be any order; we just use it to identify each party with a row in the table which we now begin to fill.
	 For each~$i$, $1 \leq i \leq |\parties_{-P_1}|$, and each~$s$, $0 \leq s \leq \seats-\desiredseats$, the cell $\mathrm{tab}[i][s]$ contains the minimum number of votes needed to be moved away from parties $o(1), \dots, o(i)$ such that $o(1), \dots, o(i)$ receive $s$ seats in total before $P_1$ is assigned its $\desiredseats$-th seat (again, assuming that $P_1$ has exactly $\budget$ additional votes in the end).
	 This table can also be efficiently computed with dynamic programming, as we describe later.
	 Finally, we check if there exists a value of at most $\budget$ in the last row of the table.
	 If this holds, we answer \textsc{Yes} because there do exist bribes that do not exceed $\budget$ and ensure that the other parties leave the $\desiredseats$-th seat for~$P_1$.
	 
	 Note that by tracing back through the table $\mathrm{tab}$ we can find the individual numbers of votes we need to move from each party to $P_1$ for a successful campaign.
	 This number does not sum up to $\budget$ in many cases.
	 If so, we can simply remove the remaining votes from arbitrary parties (except~$P_1$).

	 \begin{theorem}\label{thm:bribery-top-choice}
	 	For each divisor sequence method~$\appmethod$,
	 	$\appmethod$-\textsc{AB} is in~$\p$.
	 \end{theorem}
	 \begin{proof}
		 Algorithm~\ref{alg:bribery-constructive} decides whether a successful campaign exists in the constructive case.
		 The algorithm works with every divisor sequence method.
	 	We first prove that the algorithm indeed runs in polynomial time.
	 	For most parts of the algorithm this is easy to see: We essentially fill a table with $|\parties|$ rows and at most $\seats$ columns.
	 	Since $\seats \leq |\votes|$, the table size is indeed polynomial in the input size.
	 	However, it is yet unclear how $\gamma$ is computed.
	 	Computing $\gamma$ works with a binary search for the jumping points of a function~$\phi$, which is defined to give the number of seats a party with $y$ votes receives before $P_1$ receives $\desiredseats$ seats, assuming that $P_1$ has exactly $\budget$ additional votes in the end.
	 	Let $q$ be the final vote count of $P_1$ (i.e., with the $\budget$ additional votes).
	 	Then, for a divisor sequence method with the sequence $d = (d_1, d_2, \dots, d_{\seats})$, we have
	 	\[
                \phi(y) = \begin{cases} 0 &\text{ if } y < \threshold\\0 &\text{ if } y \leq \nicefrac{q}{d_{\desiredseats}}\\ \max\{z \in \{1, \dots, \seats\} \mid \nicefrac{y}{d_{z}} > \nicefrac{q}{d_{\desiredseats}}\} & \text{ otherwise. }\end{cases}
                \]
	 	The jumping points can be found by binary search in time $\mathcal{O}(\seats \cdot \log(\budget))$.
	 	
	 	We now prove the correctness of the algorithm.
	 	Starting in the beginning, setting $\budget$ to the minimum of $n-\supportalloc(P_1)$ and $\budget$ is necessary to ensure that we never move more votes from parties in $\parties_{-P_1}$ to $P_1$ than allowed; in particular, setting $\budget$ higher than that would result in false positive results.
	 	For the remainder of this proof, we assume that all $\budget$ votes are moved from $\parties_{-P_1}$ to $P_1$, i.e., $P_1$ receives $\budget$ additional votes in the end.
	 	This is optimal according to Lemma~\ref{lem:moving-voters}.
	 	The first if-statement returns \textsc{No} if $P_1$ cannot reach the threshold. This answer is correct since $P_1$ can never get any seat as long as it is below the threshold, i.e., in this case the bribe is unsuccessful.
	 	
	 	In the middle part of the algorithm, we fill a table.
	 	Recall that for each~$i$, $1 \leq i \leq |\parties_{-P_1}|$, and each~$s$, $0 \leq s \leq \seats-\desiredseats$, the cell $\mathrm{tab}[i][s]$ contains the minimum number of votes needed to be removed from parties $o(1), \dots, o(i)$ such that $o(1), \dots, o(i)$ receive $s$ seats in total before $P_1$ is assigned its $\desiredseats$-th seat.
	 	The values are computed dynamically from the previous row to the next row.
	 	This is possible because the
                number of seats that parties $o(1), \dots, o(i)$ receive in total before $P_1$ is assigned its $\desiredseats$-th seat is exactly the sum of the seat counts the parties have individually before $P_1$ receives its $\desiredseats$-th seat.
	 	Further, since this number can be computed directly by comparing the divisor list of
                a party with the divisor list of $P_1$ (i.e., the $\phi$ function of each party is independent of other parties' support), the required bribery budget is also exactly the sum of the individual bribes.
	 	Thus the values in the
                table are indeed computed correctly.
	 	
	 	Finally, if in the last row there exists a value of at most $\budget$, we correctly answer \textsc{Yes}, by the following argument.
	 	Suppose we have a value of at most $\budget$ in cell $\mathrm{tab}[|\parties_{-P_1}|][s]$.
	 	Then there are bribes that do not exceed $\budget$ and ensure that the other parties receive at most $s$ seats before $P_1$ is assigned its $\desiredseats$-th seat.
	 	Since there are a total of $\seats$ seats available, and the other parties get $s\leq \seats-\desiredseats$ seats before $P_1$ receives the $\desiredseats$-th seat, $P_1$ will indeed receive its $\desiredseats$-th seat.
	 	However, if all cells of the last row contain a value greater than~$\budget$, the given budget is too small to ensure that the other parties receive at most $\seats-\desiredseats$ seats before $P_1$ receives its $\desiredseats$-th seat.
	 	Thus the other parties receive at least $\seats-\desiredseats +1$ seats in this case, which leaves at most $\desiredseats-1$ seats for~$P_1$, so we correctly answer \textsc{No}.
	 \end{proof}

	 We can easily adapt Algorithm~\ref{alg:bribery-constructive} for the destructive case.
	 Now, we \emph{remove} $\min\{\budget, \supportalloc(P_1)\}$ votes from party $P_1$ and \emph{add them to the other} parties.
	 Of course, when $P_1$ is pushed below the threshold, we immediately answer \textsc{Yes}.
	 For the destructive case, $\phi$ and $\gamma$ need to be defined slightly differently.
	 Here, we define $\phi(y)$ as the number of seats a party with $y$ votes receives before $P_1$ is assigned its $(\desiredseats+1)$-th seat (which we try to prevent).
	 And
         $(x,cost)$ in $\gamma[P]$ is defined as the minimum number of votes we need to add to party $P$ such that it receives at least $x$ seats before $P_1$ is assigned its $(\desiredseats+1)$-th seat.
	 Again, we fill the table with dynamic programming but this time, whenever we have filled a row completely, we check if it is possible for the parties corresponding to all
         already filled rows to receive at least $\seats - \desiredseats$ seats before $P_1$ receives its $(\desiredseats+1)$-th seat by bribery.
	 That is, we check whether we filled the last cell in the current row with a value of at most~$\budget$ or whether it would be possible to fill a cell beyond the last table cell in the current row with such a value.
	 In that case we answer \textsc{Yes}, since there are not enough seats left for $P_1$ to be assigned its $(\desiredseats + 1)$-th seat with this bribery action.
	 If this was never possible, we answer \textsc{No} because the best we could do is to occupy at most $\seats - \desiredseats-1$ seats with parties in~$\parties_{-P_1}$, which still leaves the $(\desiredseats + 1)$-th seat for~$P_1$.
	 This proves the following result.

 	 \begin{theorem}\label{thm:destr-bribery-top-choice}
	   For each divisor sequence method~$\appmethod$,
           $\appmethod$-\textsc{DAB} is in~$\p$.
	 \end{theorem}

	 Note that since we can decide in polynomial time whether there exists a successful campaign guaranteeing $P_1$ at least $\desiredseats$ seats by bribing at most $\budget$ voters, we can also find the maximum number of seats we can guarantee for $P_1$ with a budget of $\budget$ (using a simple binary search) in polynomial time.
	 Analogously, for the destructive variant, it is possible to efficiently determine the maximum number of seats we can steal from $P_1$.
	 We demonstrate this experimentally in Section~\ref{sec:experiments}, testing the effectiveness of the campaigns on real-world elections.
	 
	 We now show the LRM analogues of all previous results, starting with the LRM analogue of Lemma~\ref{lem:moving-voters}.
	 
	 \begin{lemma}\label{lem:moving-voters-LRM}
	 	The following two statements hold for LRM.\footnote{Note that these two statements also hold for moving votes from other parties to $P_1$ and from $P_1$ to other parties.}
	 	\begin{enumerate}[label=(\roman*)]
	 		\item Adding votes to support the distinguished party $P_1$ can never make $P_1$ lose seats.
	 		Further, removing votes from $P_1$ can never make $P_1$ gain seats.
	 		\item Adding votes to support another party than $P_1$ cannot increase the number of seats for $P_1$ by more than if we would add the same number of supporters to $P_1$.
	 		Conversely, removing votes from another party than $P_1$ cannot decrease the number of seats for~$P_1$ by more than if we would remove the same number of supporters directly from $P_1$.
	 	\end{enumerate}
	 \end{lemma}
	\begin{proof}
		The first claim, which expresses that LRM satisfies some kind of ``monotonicity,'' is easy to see.
		For the second claim, note that adding or removing votes changes $n$ and thus the fair shares of all parties.
		It is in fact possible (although rare) that $P_1$ receives more seats than before by adding a vote to \emph{another} party than $P_1$ due to the changed remainders with the changed value of~$n$.\footnote{Party $P_1$'s remainder may be decreased a little less than the remainder of a much bigger party.}
		However, note that we would achieve the same $n$ by adding the same votes to $P_1$ instead of to the other party.
		In that case, $P_1$'s fair share is even higher, so $P_1$ will not receive fewer seats than before.
		An analogous argumentation applies to the converse case since removing votes from $P_1$ reduces $P_1$'s seat count at least as much as removing votes from another party.
	\end{proof}
	
	Next, we show the LRM analogue of Theorem~\ref{thm:bribery-top-choice}.

	\begin{theorem}\label{thrm:bribery-top-choice-lrm}
		LRM-\textsc{AB} is in~$\p$.
	\end{theorem}

	\begin{algorithm}[h!]
	\caption{for LRM-\textsc{AB}}
	\label{alg:bribery-constructive-lrm}
	\textbf{Input}: $\parties$, $\votes$, $\threshold$, $\seats$, $P_1$, $\budget$, $\desiredseats$
	\begin{algorithmic}[1]
		\STATE $\budget \gets \min\{n-p_1, \budget\}$
		\IF{$p_1 + \budget < \threshold$}
		\RETURN \textsc{No}
		\ENDIF
		\STATE $\supportalloc(P_1) \gets p_1 + \budget$
		\WHILE{True}
		\STATE $n \gets (\sum_{P \in \parties} \supportalloc(P)) - \budget$
		\IF{$\seats \cdot \frac{\supportalloc(P_1)}{n} \geq \desiredseats$}
		\RETURN \textsc{Yes}
		\ENDIF
		\IF{$\seats \cdot \frac{\supportalloc(P_1)}{n} == \desiredseats-1$}
		\STATE compute $\gamma$
		\STATE initialize table $\mathrm{tab}$ with $\seats - \desiredseats$ columns and $m$ rows, where $\mathrm{tab}[0][0] \gets 0$ and the other entries are~$\infty$
		\STATE let $o:\{1, \dots, |\parties_{-P_1}|\} \to \parties_{-P_1}$ be an ordering
		\FOR{$i \gets 1$ to $|\parties_{-P_1}|$}
		\FOR{$s \gets 0$ to $\seats - \desiredseats$}
		\FOR{$(x, \mathrm{cost}) \in \gamma[o(i)]$}
		\IF{$s-x \geq 0$}
		\STATE $\mathit{tmp} \gets \mathrm{tab}[i-1][s-x] + \mathrm{cost}$
		\IF{$\mathit{tmp} < \mathrm{tab}[i][s]$}
		\STATE $\mathrm{tab}[i][s] = \mathit{tmp}$
		\ENDIF
		\ENDIF
		\ENDFOR
		\ENDFOR
		\ENDFOR\\
		\FOR{$s \gets 0$ to $\seats - \desiredseats$}
		\IF{$\mathrm{tab}[|\parties_{-P_1}|][s] \leq \budget$}
		\RETURN \textsc{Yes}
		\ENDIF
		\ENDFOR
		\ENDIF
		\STATE let $P_i$ be the party with the smallest positive support
		\IF{$\supportalloc(P_i) - \threshold + 1 \leq \budget$}
		\STATE\COMMENT{\textcolor{gray}{\emph{push $P_i$ below the threshold}}}
		\STATE $\supportalloc(P_i) \gets 0$
		\STATE $\budget \gets \budget - (\supportalloc(P_i) - \threshold + 1)$
		\ELSE
		\RETURN \textsc{No}
		\ENDIF
		\ENDWHILE	  
	\end{algorithmic}
	\end{algorithm} 
	
	\begin{proof}
		We prove that Algorithm~\ref{alg:bribery-constructive-lrm} (which obviously runs in polynomial-time) correctly decides LRM-\textsc{AB}.
		First, note that according to Lemma~\ref{lem:moving-voters-LRM} we should move as many votes as possible to~$P_1$.
		In case $P_1$ still does not reach the threshold, we can answer \textsc{No} since there is no way for $P_1$ to get any seat.
		
		Otherwise, the while-loop will now first check whether it is possible to make $P_1$ get the desired seats without pushing additional parties below the threshold, then whether this is possible when pushing exactly one party below the threshold, and so on. Note that a party $P_i$ is pushed below the threshold by removing $\supportalloc(P_i) - \threshold + 1$ votes from $P_i$.
		This step-by-step approach is important because every time we push a party below the threshold, $n$ changes, and so will the fair shares of all parties
                still in the race.
		In each iteration of the while-loop, we distinguish three cases:
		\begin{description}
			\item[\normalfont{Case 1:}] If $P_1$ gets $\desiredseats$ seats by lower quota, we can answer \textsc{Yes} since it does not matter from which parties we move the votes to $P_1$.
			\item[\normalfont{Case 2:}] In case $P_1$ receives
                          fewer than $\desiredseats - 1$ seats as lower quota, we cannot decide the instance in the current iteration,
                          which is why we remove (i.e., push below the threshold) the next party and continue with the next iteration.
			This is correct, as $P_1$ can only get up to one additional (remainder) seat by strategically moving votes from the other parties to $P_1$.
		        \item[\normalfont{Case 3:}] Finally, in the third case, $P_1$ gets exactly $\desiredseats - 1$ seats as lower quota.
                          This is the most challenging part, since now we have to check whether it is possible to move the votes strategically such that $P_1$ has one of the largest remainders, and thus receives a remainder seat.
		
			This is done using a similar dynamic programming approach as we did for divisor sequence methods.
			However, $\gamma$ is defined slightly differently this time.
			That is,
                        $(x,cost)$ in $\gamma[P]$ gives the minimum number of votes that have to be removed from party $P$ such that $P$ receives only $x$ seats before $P_1$ receives \emph{the remainder seat} (assuming $P_1$ has exactly $\budget$ additional votes in the end), and $P$ is still above the threshold.
			Computing the $\gamma$ values once again works with a binary search for the jumping points of the function $\phi$.
			Here, $\phi(y)$ is defined as the number of seats a party with $y$ votes receives before $P_1$ receives a remainder seat, i.e., the lower quota whenever the remainder is higher than the remainder of $P_1$, or else the lower quota plus one. 
			With dynamic programming we fill a table with $\seats - \desiredseats$ columns.
			If we find a value in the last row of the table which is at most our budget, we can answer \textsc{Yes} since this value certifies that it was possible to decrease the support of the other parties such that they receive at most $\seats - \desiredseats$
                        seats before $P_1$ receives a remainder seat.
			Due to the lower quota of $\desiredseats - 1$ seats $P_1$ receives, there is at least $\seats - \desiredseats - (\desiredseats - 1) = 1$ remainder seat left, which naturally goes to $P_1$.
			If no value in the last table row is less than the budget, we continue to remove the smallest party, as it is impossible with the current $n$ (in combination with the support allocation) to provide $P_1$ with a remainder seat.
		
			Note that by pushing the least supported party below the threshold in each succeeding iteration, we eventually cover every possible number of
                        removed parties (i.e., of parties pushed below the threshold).
		\end{description}
		Thus, if
                apportionment bribery by LRM can be successful, a successful campaign will be found in one iteration of the while-loop.
		If no iteration of the while-loop returns \textsc{Yes} (i.e., when we cannot push further parties below the threshold), we can be sure that apportionment bribery by LRM cannot be successful, and thus we answer \textsc{No}.
	\end{proof}

	Algorithm~\ref{alg:bribery-constructive-lrm} can easily be adapted to decide the destructive variant of the problem as well.
	This is done in the same way as Algorithm~\ref{alg:bribery-constructive} was adapted for the destructive case using divisor sequence methods.
	In the destructive
        variant, we successively \emph{add} parties, i.e., bring them above the threshold by adding supporters (instead of removing parties by pushing them below the threshold as in Algorithm~\ref{alg:bribery-constructive-lrm}).

	\begin{theorem}\label{thm:destr-bribery-top-choice-lrm}
		LRM-\textsc{DAB} is in~$\p$.
	\end{theorem}

	 Note that both algorithms can also easily be adapted to actually compute a campaign (if one exists). 	 
	 
	 Similar algorithms also work for the cases of ensuring that~$P_1$ has
	 strictly more seats than any other party (but without requiring a
	 particular number of them), i.e., for \pluralityAGnp{}. However, in
	 this scenario the algorithms require more care. The issue is that
	 while our algorithms for \AGnp{} find ways to give $\desiredseats$ seats to
	 party $P_1$ (provided that this is possible), as a side effect they
	 may also increase the numbers of seats allocated to other parties. It
	 is also possible that to solve \pluralityAGnp{} it suffices to move
	 votes in such a way that a party with the most seats loses some of
	 them in favor of parties other than $P_1$. Finally, we need to pay
	 more attention to tie-breaking.

	\begin{theorem}\label{thm:winner-single-bribery}
	  For $\calR$ being either a divisor sequence method
	  or LRM, \pluralityAG{$\calR$} is in~$\p$.
	\end{theorem}
		\begin{proof}
		  Let $\supportalloc$ be our input support allocation,
                  let $\seats$ be the number of seats to be allocated,
                  and let $\budget$ by the bribing budget.
                  We first describe the algorithm for the case of divisor
			methods with a divisor sequence $d= (d_1, d_2, ..., d_{\seats})$ and then argue what modifications are needed for LRM.

			Our algorithm first guesses the number of seats $\desiredseats$ that $P_1$ is
			to end up with
                        for a yes-instance of the problem.
                        All the other parties should end up with fewer seats in that case.
                        Next, it guesses some party~$P_s$ with $\supportalloc(P_s) > 0$, the number of
			seats~$\desiredseats_s$ that~$P_s$ will end up with ($\desiredseats_s < \desiredseats$), and
			the number of votes $b'_s$ that we move from $P_s$ to $P_1$
			(throughout the proof, we assume that $s \neq 1$; the case of
			$s = 1$ can be handled analogously). The idea is that the value
			$x = \nicefrac{\left(\supportalloc(P_s)-b'_s\right)}{d_{\desiredseats_s}}$ will be the lowest in the
			final $k$-division sequence
                        that warrants a seat (up to
			tie-breaking). Finally, we guess party $P_t$ ($t \geq s$) such that
			for parties $P_1, \ldots, P_t$, the value $x$ in the $k$-division
			sequence suffices to get a seat, whereas for parties
			$P_{t+1}, \ldots, P_m$ only values larger than $x$ suffice for a
			seat (in other words, $P_t$ is the last party for which tie-breaking
			is favorable).  Our goal is to decide if it is possible to move $\budget$
			votes from the other parties to $P_1$, so that all the above
			intentions
                        can be realized.
			
			To this end, we need some additional notation. For an integer
			$i \in [m]$ and 
			support $\supportalloc(P)$,
			let $\seats_i(\supportalloc(P))$ be an
			integer~$\lambda$, such that $\nicefrac{\supportalloc(P)}{d_{\lambda+1}} < x$ and it holds that
			\begin{enumerate}
				\item if $i \leq t$ then  $\nicefrac{\supportalloc(P)}{d_{\lambda}} \geq x$, and
				\item if $i > t$ then $\nicefrac{\supportalloc(P)}{d_{\lambda}} > x$.
			\end{enumerate}
			Intuitively, $\seats_i(\supportalloc(P))$ gives the number of seats that party
			$P_i$ would get if it had a support of $\supportalloc(P)$, provided that the lowest value
			in the division sequence that warrants a seat is $x$ (modulo
			tie-breaking).
			
			Let us consider some $i \in [m]$,
                        $j \leq \seats$, and $b \in [\budget]_0$. The
			interpretation of these values is that we consider the first $i$
			parties, i.e., $P_1, \ldots, P_i$, who will be assigned $j$ seats,
			and $b$ is the number of votes that we need to move from parties
			$P_2, \ldots, P_i$ to~$P_1$.
                        We define a boolean function $f(i,j,b)$
			to be \emph{true} exactly if there is a sequence $(b_2, \ldots, b_i)$ of
			nonnegative integers (which, intuitively, give the numbers of votes
			moved from parties $P_2, \ldots, P_i$ to~$P_1$) such that:
			\begin{enumerate}
				\item $\sum_{r=2}^i b_r = b$ (i.e., we move exactly $b$ votes from $P_2, \ldots, P_i$ to~$P_1$).
				\item If $i \geq s$ then $b_s = b'_s$ and
				$\seats_s(\supportalloc(P_s)-b'_s) = \desiredseats_s$ (i.e., we move exactly $b'_s$ votes from
				$P_s$ to $P_1$, and $P_s$ gets exactly $\desiredseats_s$ seats),
				\item $\seats(p_1+\budget) \geq \desiredseats$ and for each
				$r \in [i] \setminus \{1\}$, $\seats_1(\supportalloc(P_i)-b_i) < \desiredseats$ (i.e., $P_1$
				gets at least $\desiredseats$ seats and every other party gets fewer
				seats).
				\item $\seats_1(p_1+\budget) + \sum_{r=2}^i \seats_r(\supportalloc(p_r)-b_r) = j$
				(i.e., altogether, $j$ seats are allocated to parties
				$P_1, \ldots, P_i$).
			\end{enumerate}
			
			Our main algorithm accepts exactly if $f(m,k,\budget)$ is \emph{true}
			which, indeed, means that $P_1$ gets at least $\desiredseats$ seats and all
			the other parties get fewer seats.
                        It remains to show how to compute the
			values $f(i,j,b)$ in polynomial time.
			
			Even though function $f$ may look somewhat complicated, it is quite
			easy to compute it using dynamic programming. First, we set
			$f(i,j,b) = \mathit{false}$ for all $j < 0$ and all $b < 0$.
			Second, we note that for $i=1$, computing $f(i,j,b)$ follows by
			directly implementing the definition.  Third, for values $i \geq 2$,
			we express $f(i,j,b)$ recursively.  If $i \geq 2$ but $i \neq s$,
			then
                        $f(i,j,b)$ equals:
			\[
			\bigvee_{\substack{b_i \in [\budget]_0,\\ t_{i-1} \in T(t)}}
			\begin{array}{l}
			  f(i-1,j-\seats_i(\supportalloc(p_i)-b_i), b-b_i)
				\land (\seats_i(\supportalloc(p_i)-b_i)  < \desiredseats).
			\end{array}
			\]
			Further,
                        $f(s,j,b)$ is equal to:
			\[
			f(i-1,j-\desiredseats_s, b-b'_s) \land
                        (\seats_s(\supportalloc(p_s)-b'_s)
                        = \desiredseats_s).
			\]
			
			Using these recursions and standard dynamic programming techniques,
			we obtain a polynomial-time algorithm for computing function~$f$.
			This completes the proof for the case of divisor sequence methods.\medskip
			
			For the LRM method, our algorithm is the same as for the
			divisor sequence methods, except of the following:
			\begin{enumerate}
				\item Instead of guessing the lowest value $x$ in the division
				sequence that warrants a seat, we guess the lowest value $x$ of a
				remainder that warrants a remainder seat (as in the divisor sequence methods
				case, we do so by guessing a party and the number of votes that we
				take away from it, and then we compute its remainder).
				\item We redefine $\seats_i(\supportalloc(P))$ to either be
				$\quota(\supportalloc(P),n,\seats) + [\remainder(\supportalloc(P),n,\seats) \geq y]$ (for parties $P_i$ for
				which the guessed tie-breaking is favorable) or to be
				$\quota(\supportalloc(P),n,\seats) + [\remainder(\supportalloc(P),n,\seats) > y]$ (for parties for which
				tie-breaking is not favorable).
			\end{enumerate}
			This completes the proof.
		\end{proof}

	To conclude, we note that our algorithms from
	Theorems~\ref{thm:bribery-top-choice}
	and~\ref{thrm:bribery-top-choice-lrm} can be adapted to the case where for
	each party $P_i$, $i \in \{2, \ldots, m\}$, there is a
	polynomially-bounded, nondecreasing function $\cost_i(x)$ (provided
	as part of the input), specifying the cost of moving $x$ votes from
	$P_i$ to $P_1$, and where we ask if there is a bribery of total cost
	at most $\budget$. This way we can, e.g., model the increasing difficulty of
	convincing larger groups of voters to vote for $P_1$.  To take such
	functions into account, in our algorithms we would modify the $f$
	function to not return a \emph{true/false} value, but the lowest cost
	of achieving a particular effect (where cost $\infty$ would correspond
	to impossibility).
	
	The destructive cases of winner bribery are somewhat simpler, as we must only find \emph{one} party $P_i \in \parties_{-P_1}$ that we can
        make receive
        more seats than our distinguished (i.e., in this case, despised) party~$P_1$.
	However, while a simple greedy algorithm suffices for divisor sequence methods, LRM needs more care. We thus provide two separate proofs.
	
	\begin{proposition}\label{prop:destr-winner-brib}
	  For each divisor sequence method~$\appmethod$,
          $\appmethod$-\textsc{DAWB} is in~$\p$.
	\end{proposition}
	\begin{proof}
		First, we check if there is a party $P_i$ with $\seatalloc(P_i) > \seatalloc(P_1)$. If this is the case, we can answer YES immediately.
		Otherwise, let us assume, without loss of generality, that for each $P_i$ and~$P_j$, $i<j$, we have $\supportalloc(P_i)\geq\supportalloc(P_j)$ and that $P_1$ is one of the parties with the highest number of seats.  
		We now simply move all $\budget$ votes from $P_1$ to $P_2$, i.e., to a party with
                at most as much support as $P_1$ has.\footnote{If $\budget > \supportalloc(P_1)$, we move
                $\budget - \supportalloc(P_1)$ votes from parties $\parties_{-\{P_1,  P_2\}}$ to~$P_2$.}
		If this implies $\seatalloc(P_2) > \seatalloc(P_1)$, we answer YES; otherwise, we answer NO.
	
		It is easy to see that the algorithm runs in polynomial time. 
		Due to second part of Lemma~\ref{lem:moving-voters} we know that in the destructive case, the whole budget $\budget$ should be exhausted to move votes from $P_1$ to parties $P_i \in \parties_{-P_1}$.
		It remains to show that it is never better to give the $\budget$ votes from $P_1$ to another party $P_i \in \parties_{-\{P_1,  P_2\}}$ with less support than $P_2$ or to distribute the $\budget$ votes among several parties $P_i \in \parties_{-P_1}$.
		This follows from the first part of Lemma~\ref{lem:moving-voters} stating that given a budget $\budget$, the maximum number of additional seats for a party $P_i$ can always be achieved by moving $\budget$ votes from parties in $\parties_{-P_i}$ (here, as established above, from party $P_1$) to $P_i$. 
		Since in the destructive case of the winner problem we only need one party to beat $P_1$ (i.e., to have a higher number of seats), a party with the least distance
                $\supportalloc(P_1)-\supportalloc(P_i)$ to $P_1$ (here party $P_2$) will always be a party that needs the minimum number of votes to beat~$P_1$.
	\end{proof}
	
	\begin{proposition}\label{prop:destr-winner-brib-lrm}
		LRM-\textsc{DAWB} is in~$\p$.
	\end{proposition}
\begin{proof}
	We prove this proposition by first analyzing where the votes should be moved to and then checking where we should take the votes from.
	Without loss of generality, assume that for each $P_i$ and $P_j$, $i<j$, we have $\supportalloc(P_i)\geq\supportalloc(P_j)$ and that $P_1$ is one of the parties with the highest number of seats and our distinguished party.
	
	First, we show that transferring all possible votes to $P_2$ gives us the greatest chance to make $P_1$ not a winner, i.e., not have the greatest number of seats.
	Let us give $\budget'$ votes to $P_2$ and $\budget''$ votes to
	some $P_t$ for $t>2$ with $p_t\ge\tau$, where $\budget'+\budget''\le \budget$, without changing the denominator $n=\sum_{i=1}^{|\mathcal{P}|}\supportalloc(P_i)$.
	Let us assume that after this bribery
	\[
	\frac{\supportalloc(P_2)+\budget'}{n} < \frac{\supportalloc(P_t)+\budget''}{n}.
	\]
	But then,
	\[
	\frac{\supportalloc(P_t)+\budget''}{n} \le \frac{\supportalloc(P_2)+\budget''}{n} \le \frac{\supportalloc(P_2)+\budget'+\budget''}{n},
	\]
	so if we transferred all the votes directly to~$P_2$, $P_2$ will not receive
	fewer seats than~$P_t$.
	If we bribe voters also voting for parties below the threshold~$\threshold$ to vote for parties above~$\threshold$ instead, the denominator increases by some value $L$ and the situation will be analogous to the previous one. 
	That is, we can transfer the votes directly to~$P_2$: For any $P_t$ with greater new support than $P_2$ after getting additional $\budget''$ votes, it holds that
	\[
	\frac{\supportalloc(P_t)+\budget''}{n+L} \le \frac{\supportalloc(P_2)+\budget''}{n+L} \le \frac{\supportalloc(P_2)+\budget'+\budget''}{n+L}.
	\]
	If we bribe some of the voters to vote for some parties below the threshold, such that after the bribery their support is positive, then for $\budget'''<\budget'$ and $L'>L$, it holds that
	\[
	\frac{\supportalloc(P_2)+\budget'-\supportalloc(P_1)+\budget''}{n+L} - \frac{\supportalloc(P_2)+\budget'''-\supportalloc(P_1)+\budget''}{n+L'} > 0.
	\]
	Thus also in this case, it is not better than transferring the votes directly to $P_2$.
	The last case is when we bribe voters voting for parties above the threshold in a way that there exists at least one party whose new support is zero---let us call them~$P_t$. 
	Hence, the denominator decreases by some value $D$.
	If we transfer at least one vote not to~$P_2$, we have that if
	\[
	\frac{\supportalloc(P_s)+\budget'}{n-D}\ge \frac{\supportalloc(P_2)+\budget''}{n-D}
	\]
	for any $P_s\notin \{P_1,P_2\}$, then it holds that
	\[
	\frac{\supportalloc(P_s)+\budget'}{n-D} \le \frac{\supportalloc(P_2)+\budget'}{n-D}\le \frac{\supportalloc(P_2)+\budget'+\budget''}{n-D}.
	\]
	So again,
	we can instead transfer the votes directly to~$P_2$.
	Now let us consider a mix of the bribes above: For some $K',K''$ with $0 \le K',K'' < K$, some $L, D\ge 0$, and some $P_t\notin \{P_1,P_2\}$  with $p_t \ge \threshold$,
	\[
	\frac{\supportalloc(P_t)+\budget''}{n+L-D} \le \frac{\supportalloc(P_2)+\budget''}{n+L-D} \le \frac{\supportalloc(P_2)+\budget'+\budget''}{n+L-D}, 
	\]
	i.e., we can again instead transfer the votes directly to~$P_2$.
	Finally, note that for $K' > K''$, where $K'-K''$ votes were moved from the parties with positive supports to some parties below the threshold, such that they stay below the threshold, and for $n'$ being the denominator after a bribery action when $P_2$ received the additional \budget' votes, it holds that 
	\[
	\frac{\supportalloc(P_2)+\budget'}{n'} \ge \frac{\supportalloc(P_2)+\budget''}{n'-(K'-K'')}
	\]
	if and only if
	\begin{align*}
		n'\supportalloc(P_2) + n'K' - K'(K'-K'') - 	n'\supportalloc(P_2)- n'K'' & \ge 0,
	\end{align*}
which in turn is equivalent to
	\begin{align*}
		(n'- K')(K'-K'')  & \ge 0.
	\end{align*}
	That is, the initial inequality holds if $n'\ge K'$, which is always true (also note that for any support for $P_1$, its $\fairshare(P_1)$ is smaller when the denominator is larger).
	
	Let us now analyze whether we should
	bribe all $\budget$ voters supporting $P_1$ if $\supportalloc(P_1)\ge \budget$.
	If we bribe all of them, the denominator will not change
	(unless it makes $P_1$ have its support equal to zero, but then it is easy to check if $P_2$ receives more seats than the distinguished party).
	Let us assume that in the cases we consider from now on, the situations where fewer than $\budget$~votes are moved from $P_1$ are successful briberies.
	It is easy to see that when $P_2$ would get more votes than~$P_1$, the distance between the new number of votes for $P_2$ and the support of $P_1$ (if we take from $P_1$ some $\budget'''<\budget$ votes) is smaller than the distance between the former and $\supportalloc(P_1)-\budget$ while not changing~$n$.
	Therefore, in the latter situation, $P_2$ has better chances to get more seats than~$P_1$.
	If we bribe also some $\budget'<\budget$ voters voting for parties below the threshold and/or the parties above the threshold but without decreasing their support to zero instead of those voting for~$P_1$, the denominator will increase, and for $\budget''<\budget$ and $\budget'+\budget''\le \budget$, it holds that
	\[
	\frac{\supportalloc(P_2)+\budget-\supportalloc(P_1)+\budget}{n} - \frac{\supportalloc(P_2)+\budget-\supportalloc(P_1)+\budget''}{n+\budget'} > 0.
	\]
	That is, this strategy is not better than taking more votes from~$P_1$.
	We could also transfer votes to $P_2$ from parties with support at least the threshold in such a way that their new support is zero, i.e., the denominator decreases. Since $P_2$ cannot have a smaller support than $P_1$ to be able to receive more seats than $P_1$, for $\budget''< \budget$ and for some $P_t\notin\{P_1,P_2\}$ (or a group of such parties), it holds that  
	\[
	\frac{\supportalloc(P_2)+\budget}{n-\supportalloc(P_t)+K-K''} - \frac{\supportalloc(P_1)-\budget''}{n-\supportalloc(P_t)+K-K''} \geq 0.
	\]
	That is, if $P_2$ should be the winner, then $\supportalloc(P_2)+\budget\ge \supportalloc(P_1)-\budget''$. Moreover, this distance is also greatest if $K''$ is as large as possible and/or the support of $P_t$ is as large as possible. For each coalition $P_t$ whose support is to be reduced to $0$, it is sufficient to take $p_t-\threshold+1$ votes from it: For any $L<L'$, where both are enough to decrease $P_t$'s support to zero, it holds that
	\[
	\frac{\supportalloc(P_2)+\budget-\supportalloc(P_1)+K-L}{n-\supportalloc(P_t)+L} > \frac{\supportalloc(P_2)+\budget-\supportalloc(P_1)+K-L'}{n-\supportalloc(P_t)+L'}.
	\]
	Therefore, it is enough to check the bribery from $m$ parties with the smallest support reaching the threshold, where we iterate $m$ from $1$ until we reach our budget.
	
	In the case when $\supportalloc(P_1)<\budget$, we can decrease $P_1$'s support to zero, i.e., it will not receive any seats, while there exists a party with at least one seat.
	Therefore, it can be checked in polynomial time whether we can bribe up to $\budget$ voters to give $P_2$ more seats than~$P_1$, i.e., LRM-\textsc{DAWB} is in~$\p$.
\end{proof}

	\subsection{Multi-District Case}\label{subsec:multi-district}
	
	It turns out that we can use the single-district algorithms to solve
	the \multiBribery{$\calR$} problem. Briefly put, we can compute the
	cost of getting each possible number of seats for the
        distinguished party $P^*$ in each
	district separately, and then solve a Knapsack-like problem to find
	out if by moving at most $B$ votes we can obtain~$\desiredseats$ seats for~$P^*$
	(which we can do in polynomial time due to our assumption that $B$ is 
	given in unary).
	
	\begin{proposition}\label{prop:divisor_lrm_multi_bribery}
		For $\calR$ being FPTP, a divisor sequence method, or LRM, \multiBribery{$\calR$} is in~$\p$.
	\end{proposition}
	\begin{proof}
		\newcommand{\function}{\ensuremath{F}}
		\newcommand{\currcost}{\ensuremath{c}}
		\newcommand{\currdist}{\ensuremath{j}}
		\newcommand{\currgain}{\ensuremath{g}}
		Consider an instance of \multiBribery{$\calR$} with $\delta$ districts, where
		we can move at most $\budget$ voters to obtain $\desiredseats$~seats for the distinguished party $P^*$.
		Further, for each $j \in [\delta]$ let $\seats_j$ be the number of seats that are
		assigned within the $j$-th district. For ease of presentation and without the
		loss of generality, we assume an arbitrary but fixed numbering of the
		districts from~$1$ to~$\delta$.
		
		In general, our algorithm proceeds in the following two steps:
		\begin{enumerate}
			\item For each $j \in [\delta]$ and each $s \in [\seats_j]_0$, we compute
			the smallest number of votes that we need to move in the $j$-th
			district to ensure that $P^*$ gets (at least) $s$ seats in this
			district. We denote this value by $\cost(j,s)$.\footnote{Following the definition of FPTP in the multidistrict setting, in each district $\delta_j$, all seats $\seats_j$ are assigned to the strongest party within~$\delta_j$. Thus the only possible values for $s$ under FPTP are zero and~$\seats_j$. Note that for all studied apportionment methods, we only consider the larger value of $s$ if there is the same value of $\cost(j,s)$ for more than one~$s$.}
			
			\item We check if there is a bribery that gives $s_1, s_2, \ldots, s_\delta$
			additional seats for~$P^*$ in the respective $\delta$ districts such that
			\begin{itemize}
				\item  [(a)] $\sum_{j=1}^{\delta} s_j \geq \desiredseats$ ($P^*$ receives at least $\desiredseats$ seats in all districts in total), and
				\item  [(b)] $\sum_{j=1}^\delta	\cost(j,s_j) \leq \budget$ (at most $\budget$ votes are moved).
			\end{itemize}
		 	We accept if such a bribery exists and we reject otherwise.
		\end{enumerate}
		
		The first step is clearly polynomial-time solvable by
		our~\Cref{thm:bribery-top-choice} or~\Cref{thrm:bribery-top-choice-lrm}. Thus we proceed with showing a
		polynomial-time dynamic program for the second step. The program is in spirit
		similar to the one for the \textsc{Knapsack} problem, but where the items
		belong to disjoint groups and we have to take exactly one item from each
		group.
		
		We begin with defining a boolean function~$\function(\currdist, \currgain,
		\currcost)$ that outputs \emph{true} if within the first $\currdist \in
		[\delta]_0$ districts there is a bribery of cost at most~$\currcost{} \in
		[\budget]_0$ that lets~$P^*$ get at least~$\currgain{} \in [\desiredseats]_0$ additional
		seats; the output is \emph{false} in the opposite case. Then, a desired
		bribery exists if~$\function(\delta, \desiredseats, \budget)$ is \emph{true}. 
		
		For technical reasons (and with a slight abuse of notation), we assume that
		$\function(0, \currgain, \currcost)$ is \emph{false} if~$\currgain > 0$ ($P^*$
		cannot have any additional seats in a bribery including zero districts) or
		if~$\currcost < 0$. Conversely, the value of~$\function$ in question is
		\emph{true} for $\currgain \leq 0$ (naturally, assuming that $\currcost \geq
		0$).

		Having the above interpretation and assumptions, we give the following
		recursive formula to compute the values of~$\function$:
		$$
		\function(\currdist, \currgain, \currcost) = \bigvee_{\currgain'
			\in [\seats_j]_0} \function(\currdist-1, \currgain-\currgain', \currcost -
		\cost(\currdist, \currgain')).
		$$
		Essentially, the above formula says that to check whether we can
		get $\currgain$~additional seats for~$P^*$ in the first~$\currdist$~districts
		for at most~$\budget$~units of budget we proceed as follows. Assuming that we know
		values of $\function$ for the first~$\currdist-1$~districts, we try each possible
		number~$\currgain'$ of additional seats that we can get in the $\currdist$-th
		district, which costs~$\cost(\currdist, \currgain)$. For each trial, we verify
		whether we can get $\currgain-\currgain'$~additional seats for the price of at
		most~$\budget-\cost(\currdist, \currgain)$ in the first~$\currdist-1$~districts. If
		the verification is positive in at least one of such trials, then we assign
		\emph{true} to the being-computed value of~$F$; otherwise, we let it
		be~\emph{false}.
		
		The correctness of our approach clearly follows from the fact that it
		exhaustively considers all possible cases potentially leading to a desired
		bribery. We compute all $O(\budget\delta\desiredseats)$~possible values of~$F$ through the
		standard dynamic programming approach, spending $O(\budget)$~time to compute 
		each of them. Thus, by our convention that the number of voters is encoded in
		unary, it follows that our dynamic program runs in polynomial time. Hence,
		so applies to our whole two-step algorithm.
		\let\function\undefined
		\let\currcost\undefined
		\let\currdist\undefined
		\let\currgain\undefined
	\end{proof}

	The same approach can be used to solve the destructive cases. Here, we define $\cost(j,s)$ to be the smallest number of votes that we need to move in the $j$-th
	district to ensure that $P^*$ gets \emph{at most} $s$ seats in this
	district and we check for a sequence $s_1, \ldots, s_\delta$ that satisfies $\sum_{j=1}^{\delta} s_j \leq \desiredseats$ and $\sum_{j=1}^\delta \cost(j,s_j) \leq \budget$.

	\begin{proposition}\label{prop:destr_divisor_lrm_multi_bribery}
	  For $\calR$ being FPTP, a divisor sequence method, or LRM, $\calR$-\textsc{DMAB} is in~$\p$.
	\end{proposition}
	
	We now move on to the winner problems for multiple districts.
  Unfortunately, \multiBriberyAG{$\calR$} is $\np$-hard for all considered
  apportionment methods, namely for divisor sequence methods, LRM, and even FPTP.
  Intuitively, this is so because the problem gives us flexibility to require
  that all parties from a given set lose one seat each, and form districts so
  that in each of them only subsets of these parties can lose seats.

	\begin{theorem}\label{thm:pAG-np-h} 
	  Let $\calR$ be either a divisor sequence method with sequence $d_1, d_2, \ldots$ or LRM. \multiPluralityAG{$\calR$} is~$\np$-complete, even if the number of seats per district as well as the number of votes transferred per district is at most three and the total number of votes per district is a constant depending on $d_1$ and~$d_2$.
	\end{theorem}
	\begin{proof}
		\newcommand{\cGraph}{\ensuremath{G}}
		\newcommand{\cVertices}{\ensuremath{V}}
		\newcommand{\cVertex}{\ensuremath{v}}
		\newcommand{\cEdges}{\ensuremath{E}}
		\newcommand{\cEdge}{\ensuremath{e}}
		\newcommand{\cVCSize}{\ensuremath{k}}
		\newcommand{\partyOf}[1]{\ensuremath{P(#1)}}
		\newcommand{\districtOf}[1]{\ensuremath{d(#1)}}
		\newcommand{\dummiesOf}[1]{\ensuremath{D(#1)}}
		\newcommand{\preferredParty}{\ensuremath{P_1}}
		\newcommand{\partyNumber}[1]{\ensuremath{P_{#1}}}
		\newcommand{\solutionBribery}{\ensuremath{S}}
		\newcommand{\solutionScore}{\ensuremath{s}}
		\newcommand{\solutionElection}{\ensuremath{E(\solutionBribery)}} 
		Membership in $\np$ is obvious as divisor sequence methods and LRM are polynomial-time computable.
		
		We show \nphardness{} by a reduction from the~\cubicVertexCover{}
		problem.  
		Below we provide a reduction for the case of the D'Hondt
		apportionment method. %

		\noindent \emph{Construction.}  Let $(\cGraph,k)$ be our
		input instance of~\cubicVertexCover{}, where
                $\cGraph = (\cVertices,\cEdges)$ is a graph and $k$ is an
                integer (without loss of generality, we assume that $k \geq 3$).
                Further, let
		$\cVertices=\{v_1,\dots,v_r\}$ and $\cEdges=\{e_1,\dots,e_t\}$.  We
		form an instance of \multiPluralityAG{D'Hondt} as follows.  First,
		we set our bribing budget to be $\budget = 3k$. Next, we create
		$1+t+10r+1$ parties according to their tie-breaking order, i.e., the firstly defined party is also chosen first by the tie-breaking rule, and so on:
		\begin{enumerate}
			\item Let $P^*= \preferredParty{}$ be our preferred party.
			\item For each edge~$\cEdge_i$, %
			we form an \emph{edge party}~$\partyOf{\cEdge_i}=\partyNumber{1+i}$
			\item For each vertex~$\cVertex_j$ %
			we form $10$~\emph{dummy parties}, so for each $\desiredseats \in [10]$ we
			let $\partyOf{\cVertex_j,\desiredseats}=\partyNumber{1+t+10(j-1)+\desiredseats}$ be
			the $\desiredseats$-th dummy party for vertex $v_j$.
			\item We also form one \emph{blocker}
			party~$\partyOf{b}=\partyNumber{1+t+10r+1}$.
		\end{enumerate}
		Then we form $r+t(k-2)+(k-1)$ districts, where each district has a threshold of zero:
		\begin{enumerate}
			\item For every vertex~$\cVertex_j$ we form a~\emph{vertex district}
			\districtOf{\cVertex_j} with three seats to allocate, where $P_1$
			gets $7$ votes, each dummy party associated with $\cVertex_j$ gets
			$10$ votes, and for each edge $\cEdge_i$ %
			incident to
			$\cVertex_j$, party \partyOf{\cEdge_i} %
			gets $10$ votes.
			
			\item We form $t(k-2)$ dummy districts, each with a single seat to
			allocate.  For each edge party there are exactly $k-2$ dummy
			districts where this party gets $2\budget+1$ votes and all the other
			parties get zero votes (this way this edge party gets the seat and
			a bribery of cost at most $\budget$ cannot change that).
			
			\item We form $k-1$ blocker districts, each with a single seat to
			allocate. In each of these districts the blocker party gets $2\budget+1$
			votes and all the other parties get zero votes.
		\end{enumerate}
		This completes the description of our construction.\smallskip
		
		\noindent\emph{Initial Seat Allocation.}
		Let us now describe the initial seat allocation, prior to any
		bribery.  The blocker party gets exactly $k-1$ seats from the
		blocker district (it has zero votes in every other district so it
		cannot get seats anywhere else). Similarly, each edge party gets
		$k-2$ seats from the dummy districts.  For the vertex districts, let
		us recall that the tie-breaking prefers party \preferredParty{} over
		the edge parties, which are then preferred over the dummy parties
		and the blocker party. Now consider some vertex district
		\districtOf{\cVertex_j}, and let $\cEdge_a$, $\cEdge_b$, and
		$\cEdge_c$ be the three edges incident to $\cVertex_j$.  There are
		three seats to allocate and one can verify that, due to
		tie-breaking, each of the parties $\partyOf{\cEdge_a}$,
		$\partyOf{\cEdge_b}$, and $\partyOf{\cEdge_c}$ gets one of them.  In
		total, party $\preferredParty$ and the dummy parties get no seats,
		the blocker party gets $k-1$ seats, and the edge parties get $k$
		seats each. In particular, \preferredParty{} does not have the
		majority of seats in the unbribed election.\smallskip

		\noindent \emph{Main Idea.}  The intuition behind our construction
		is that with budget $\budget=3k$, $\preferredParty$ can obtain $k$~seats
		by getting one seat in $k$ vertex districts.  We ensure that this is
		enough to have more seats than any other party exactly if the $k$
		districts correspond to a vertex cover in $\cGraph$.\smallskip

		\noindent\emph{Correctness.}
		To show the correctness of our reduction, we will show that
		graph~$G$ has a vertex cover of size~$\cVCSize$ if and only if there
		is a valid bribery that moves at most~$\budget$ votes and ensures that
		party~$\preferredParty$ gets more seats than any other party.
		
		For the ``only if'' direction, assume that $\cGraph$~has a vertex
		cover of size~$\cVCSize$. In particular, let $S$ be a set of $k$
		vertices that forms a vertex cover for $\cGraph$.  We form a bribery
		as follows.  For each vertex district that corresponds to a vertex
		$v \in S$, we transfer to~\preferredParty{} one vote from each of
		the three edge parties corresponding to edges incident to~$v$.  In
		every such district, \preferredParty{} now has 10 votes, every edge
		party has at most 9 votes, and the dummy parties still have 10 votes
		each.  Thus, due to tie-breaking, in each such district one seat goes
		to~\preferredParty{} and the other two seats go to the dummy parties
		(each of them to a different dummy party).  It directly follows that
		this gives \preferredParty{} $k$~ seats.  Moreover, since every edge
		is covered by at least one vertex from $S$, every edge party looses
		at least one seat and ends up with $k-1$ seats. The blocker party
		also gets $k-1$ seats, so \preferredParty{} has more seats than any
		other party.

	For the ``if'' direction, assume that there is a bribery that moves
	at most~$\budget$ votes and ensures that party~$p$ gets more seats than
	any other party.  To beat the blocker party, \preferredParty{} must
	obtain at least $k$~seats, implying that, on average, one must move
	at most $\nicefrac{\budget}{k}=3$~votes for each of these seats.  We will
	now argue that for each additional seat that $\preferredParty$ gets,
	one must move at least (and, hence, exactly) three votes.  
	
	By construction, \preferredParty{} can obtain additional seats only
	in vertex districts.  Assume towards a contradiction that
	\preferredParty{} obtained $x\ge1$~seats in some vertex district by
	moving $y\le3x-1$ votes.  With just three available seats, it holds
	that $x\le 3$ and $y\le 8$.  Moreover, as we move $y \leq 8$ votes,
	it must be the case that after the bribery there are more than three
	parties with $10$ votes in the district.\footnote{Without loss of generality, we can assume that all the votes are moved to support~$\preferredParty$, so no other party ends up with more than 10 votes.}
        As a
	consequence, we note that we must move at least $y\ge 3$ votes, so
	that $\preferredParty$ at least ties with the parties that receive
	10 votes. This implies that $x > 1$, i.e., \preferredParty{} obtains
	at least two seats in the district. However, under the D'Hondt
	method this would require $\preferredParty$ to have at least $20$
	votes.  This is a contradiction, because we would have to move $y
	\geq 13$ votes and we assumed that $y \leq 8$.

	To summarize, we know that to obtain $k$~seats, we have to move
	exactly three votes per district in exactly $\cVCSize$~vertex
	districts.  Assume towards a contradiction that the corresponding
	vertex districts do not form a vertex cover, that is, there is an
	edge~$\cEdge_i$ that is not incident to any of these vertices.  This
	means that party \partyOf{\cEdge_i} still obtains $k$ seats ($k-2$
	seats from the dummy districts and $2$ seats from the vertex
	districts) and, thus, $\preferredParty$ does not get more seats than
	any other party.  \smallskip

	Let us now briefly mention how to adapt the proof for the case of
	other apportionment methods. First, we note that the dummy districts
	and the blocker district do not need to be changed---they still
	guarantee that the appropriate parties get the required number of
	seats. Vertex districts may need to be modified, but for each of our
	apportionment methods it suffices to find a constant $x$ such that
	$\preferredParty$ gets $x-3$ votes and the other parties (with
	nonzero numbers of votes) get $x$ votes each. In particular, for
	\sainte{} our current constant, i.e., $x=10$, is sufficient. One can
	also verify that $x=10$ is sufficient for LRM 
	(specifically, there are $137$ voters in each vertex district and, so,
	prior to bribery only remainder seats are assigned in these
	districts; to get two seats, a party would need to obtain a quota
	seat and that requires at least $46$ votes, which is sufficient for
	our proof to still hold).
	\let\cGraph\undefined
	\let\cVertices\undefined
	\let\cVertex\undefined
	\let\cEdges\undefined
	\let\cEdge\undefined
	\let\cVCSize\undefined
	\let\partyOf\undefined
	\let\districtOf\undefined
	\let\dummiesOf\undefined
	\let\preferredParty\undefined
	\let\solutionBribery\undefined
	\let\solutionScore\undefined
	\let\solutionElection\undefined
	\end{proof}

	$\np$-hardness also holds for the destructive cases.
        This can be shown by extending the previous reduction.

\begin{theorem}\label{thm:dest-multi-winner}
	Let $\calR$ be either a divisor sequence method with sequence $d_1, d_2, \ldots$ or LRM. $\calR$-\textsc{DMAWB} is~$\np$-complete, even if the number of seats per district as well as the number of votes transferred per district is at most three and the total number of votes per district is a constant depending on $d_1$ and~$d_2$.
\end{theorem}
\begin{proof}
	Membership in $\np$ is obvious as divisor sequence methods and LRM are polynomial-time computable.
	
	We show $\np$-hardness by a reduction from the~\cubicVertexCover{}
	problem. As before, we provide a reduction for the D'Hondt apportionment method. A short discussion about the details regarding how to adapt it for other methods can be found in the proof of Theorem~\ref{thm:pAG-np-h}.
	
	Let $(\ensuremath{G},k)$ be our
	input instance of~\cubicVertexCover{}, where $\ensuremath{G} =
	(\ensuremath{V},\ensuremath{E})$ with $\ensuremath{V}=\{v_1,\dots,v_r\}$ and $\ensuremath{E}=\{e_1,\dots,e_t\}$  is a graph and $k$ is an integer (without loss of generality, we assume that $k \geq 3$). 
	We
	form an instance of D'Hondt-\textsc{DMAWB} 
	as we did it in the proof of Theorem~\ref{thm:pAG-np-h} with the following adaptations. To the set of parties, we add an auxiliary party $P_1$, with $P^*\neq P_1$ being our distinguished party. In the tie-breaking order, we set $P_1$ between $P^*$ and $P_2$. 
	Then we form $r+t(k-2)+(k-1)$ districts as in the reduction of the previous proof, with the only difference that in the vertex districts the auxiliary party $P_1$ is the one that receives $7$ votes. The distinguished party receives no vote in these districts. Next, we add two additional districts, with each district also having a threshold of zero. Note that since each of these districts contains only one party, all voters can vote only for this party. Accordingly, it does not matter how many voters there are in the districts, as the seats can only go to the only existing party either way. The two additional district are defined as follows: 
	\begin{enumerate}

		\item We form a district only with $P^{*}$ where there are $3k$ seats to allocate. So $P^{*}$ gets all the seats and the seats cannot be reallocated to any other party.
		
		\item We form a district only with $P_1$  where there are $2k+1$ seats to allocate. So $P_1$ gets all the seats and the seats cannot be reallocated to any other party.
	\end{enumerate}
	This completes the %
	adaptations for
	our construction.\smallskip

	Let us now describe the initial seat allocation. In the first $r+t(k-2)+(k-1)$ districts, the seats are allocated as in the proof of Theorem~\ref{thm:pAG-np-h} prior to any bribery:
	The blocker party gets exactly $k-1$ seats from the
	blocker districts. Similarly, each edge party, the edge of form $\{v_i,v_j\}$, gets
	$k-2$ seats from the dummy districts and $2$ from the districts $d(v_i)$ and $d(v_j)$, i.e., $k$ seats in total. The dummy parties get no seats. $P_1$ only gets seats in the last district, i.e., $2k+1$ seats, and $P^{*}$ gets $3k$ seats in the district where $P^{*}$ is also the only party 
	to be chosen. In particular, $P^{*}$ has the largest number of seats in the unbribed election.\smallskip
	
	We now want to bribe the election to make $P^{*}$ have less seats than some other party. First note that in all districts but in the vertex districts, there is no possibility to move seats to other parties. Therefore, each edge party cannot get more than $k+4$ seats (the edge party is in exactly two vertex districts and therefore can get maximal $6$ seats in them), each dummy party cannot get more than $3$ seats and the blocker party cannot get more than $k-1$ seats. That means that only $P_1$ may get more seats than $P^{*}$. Because our distinguished party $P^*$ cannot have less seats than $3k$, $P_1$ has to get at least $k$ more seats and it can get them only from the vertex districts. The problem is thus equivalent to the problem presented in the proof of Theorem~\ref{thm:pAG-np-h} and therefore it follows that $\calR$-\textsc{DMAWB} is \nphard{}.
\end{proof}

	It is also interesting to consider the parametrized complexity of
	\multiBriberyAG{$\calR$}. We show that the problem is~\wonehard{} for
	the parameterization by the number of districts.

	\newcommand{\preferredParty}{\ensuremath{P^*}}
	\newcommand{\partyOf}[1]{\ensuremath{P(#1)}}
	\newcommand{\dummyParties}{\ensuremath{Y}}

	\begin{theorem}\label{thm:wone-hard-districts}
	  Let $\calR$ be a divisor sequence method or LRM.
          \multiPluralityAG{$\calR$} is~$W[1]$-hard with respect
	  to the number of districts.
	\end{theorem}
\begin{proof}
  We give a parameterized, polynomial-time reduction from~\unarybinpackingstar{}. To
  this end, let us fix an instance of~\unarybinpackingstar{}
  with~\namedorderedsetof{\universe}{\element}{\universeSize}, a bin
  size~\binSize{} and the number~\binsCount{} of bins.
	For brevity, let~$\delta = \max\{\universeSize, \binSize\}$.

  For an arbitrary (but fixed) divisor sequence method~$\calR$ with sequence~$d_1, d_2,
  \ldots$, we reduce the above instance to a new instance
  of~\multiPluralityAG{$\calR$} with budget~$\budget = \binsCount
  \binSize$ as follows. We create a preferred party~\preferredParty{},
	$\delta$~\emph{dummy parties}, and, for every element~$\element$ in the
	universe, we create an \emph{element party}~\partyOf{\element}. The central
	part of our construction is a \emph{bin gadget}. The bin gadget is a district
	with \universeSize{}~seats to allocate. In the district, the preferred
	candidate gets $d_2\budget - \binSize$~votes, the dummy parties get
	$d_2\budget$~votes each, and, for each element~\element{}, its
	party~$\partyOf{\element}$ gets exactly $d_2\budget + \element$~votes. We
	create the whole election by copying the bin gadget \binsCount{}~times and we
	label the parties in a way that the tie-breaking favors the preferred party
	over the dummy parties who are preferred over the element parties. Clearly, the
	described instance can be constructed in polynomial-time.

	We collect four crucial properties of the bin gadget in the following claim.
	For readability, we prove~Claim~\ref{claim:bin-gadget-properties} in the end of
	the proof.
	\begin{claim}\label{claim:bin-gadget-properties}
		Consider an instance of the \unarybinpackingstar{} problem with a
		collection \namedorderedsetof{\universe}{\element}{\universeSize}, bin
		size~\binSize{}, and the number~\binsCount{}. Then it holds that in the
		district of the bin gadget:
		\begin{enumerate}
			\item Party~\preferredParty{} cannot get two
			seats after moving $2\binSize$ or less votes;\label{claim:point:no-two-seats}
			\item At least $\binSize$~votes needs to be moved to~\preferredParty{} to
			make it gain a single seat;\label{claim:point:price-of-single-seat}
			\item For each element~$\element{}_i \in \universe$, at least
			$\nicefrac{\element_i}{2}$~vote moves are required to
			make~\partyOf{\element_i} lose a seat;\label{claim:point:bound-element}
			\item Given that~$\preferredParty{}$ gets a seat, there is an
			element~$\element{} \in \universe$ for which at least $\element$~vote
			moves are required to make~\partyOf{\element} lose a
			seat.\label{claim:point:conditioned-bound-element}
		\end{enumerate}
	\end{claim}
	
	The created election results in each element's party having exactly
	\binsCount{}~seats and the preferred party having no seats at all. Because
	there are~\binsCount{} bin gadgets and \preferredParty{} gains at most one seat
	from a single gadget, \preferredParty{}~can have at most \binsCount{}~seats.
	Thus every element's party has to lose at least one seat in order to
	let~\preferredParty{}~have the most representatives.

	If there is a bin packing, then for each bin, we choose one (distinct) district
	and transfer \binSize{}~votes to~\preferredParty{} in a way that all parties
	corresponding to the elements in this bin have~$d_2 \budget$~votes. This is
	possible because each bin is filled with elements of size exactly~\binSize{}
	and, due to the construction, each vote moved corresponds to one unit of size.
	We end up with~\preferredParty{} and the corresponding elements' parties having
	$d_2 \budget$~votes. So the seats are allocated to the other elements'
	parties,~\preferredParty{}, and some dummy candidates. Repeating this procedure
	for every bin results in \preferredParty{}~having
	\binsCount{}~seats, some (all different) dummy parties having one seat each,
	and every element's party having $\binsCount{}-1$~seats. This clearly gives a
	``yes''-instance.
	
	Now we show that in fact in every ``yes''-instance of the
	created~\multiPluralityAG{$\calR$}, 
	party~\preferredParty{} has to get
	exactly~\binsCount{} seats. If~\preferredParty{} gets $\binsCount-1$~seats,
	then every element party has to lose at least two~seats. By
	Point~\ref{claim:point:bound-element} of
	Claim~\ref{claim:bin-gadget-properties} and the fact that our budget is~$\budget =
	\binsCount\binSize$, we know that this case requires that the corresponding
	party of \emph{every} element~$\element \in \universe$ has to lose a seat at
	most by~$\nicefrac{\element}{2}$~vote moves. Otherwise, the number of vote
	changes would exceed~$\budget$. However, \preferredParty{} has to win at least
	one seat and by
	Point~\ref{claim:point:conditioned-bound-element} of
	Claim~\ref{claim:bin-gadget-properties} there must be at least one element
	party~$\partyOf{\element}$, $\element \in \universe$, which loses seat
	by~$\element$~vote moves. Note that this argument is in fact valid any other number~$x
	\in [\binsCount-2]$ that~\preferredParty{} could end up with. Hence, it follows
	that~\preferredParty{} has to get \binsCount{}~seats.
	Eventually, by~Points~\ref{claim:point:no-two-seats}
	and~\ref{claim:point:price-of-single-seat} of
	Claim~\ref{claim:bin-gadget-properties} and the value of~$\budget =
	\binsCount\binSize$ it follows that exactly \binSize{}~units of budget are
	spent in each of $\binsCount$~district and~\preferredParty{}, with the
	number~$d_2\budget$~of voters, gets one seat in all of them.
	Thus every element's party has at most $\binsCount-1$~seats. By construction,
	if~\preferredParty{} has~$d_2\budget$~votes in a district after
	getting~$\binSize$~votes as a result of bribery, each element
	party~$\partyOf{\element}$, $\element \in \universe$, has to lose
	exactly~$\element$~votes to lose a seat in the district. So, interpreting each
	district as a bin and taking the elements whose corresponding element parties
	have~$\binSize$~votes as those that are put in the corresponding bin, we get a
	solution to the original instance of~\unarybinpacking{}.
	
	It is clear that the reduction runs in polynomial time and that the parameter
	``number of districts'' is a function of solely parameter~$k$ of
	\unarybinpacking{}. To conclude the proof of the reduction it now remains to
	prove the properties of the bin gadget as stated
	by~Claim~\ref{claim:bin-gadget-properties}.
	\begin{proof}[Proof of Claim~\ref{claim:bin-gadget-properties}]
		
		We will prove each point of the claim separately. In all of the arguments we
		aim at maximizing the number of seats that party~$\preferredParty$. For the
		rules we study, it is optimal to always move vote \emph{to}
		party~$\preferredParty$, so we follow this assumption in the proofs of all
		claims. Before we start, recall that in a bin gadget district yields
		$\universeSize$~seats to allocate and none of these is allocated to
		party~$\preferredParty$ prior to a bribery.  Furthermore, note that, except of the
		preferred party, we have at least~$2\universeSize$~parties with at
		least~$d_2\budget$~votes each.
		
		\noindent\emph{Proof of Point~\ref{claim:point:no-two-seats}.}
		To prove the claim, we show that even if $\preferredParty$ gets~$2\binSize$
		more votes, then~$\preferredParty$ does not obtain two seats, independently
		of which parties have their votes taken away. Towards contradiction, assume
		that~$\preferredParty$ has $d_2\budget - \binSize + 2\binSize$~votes and gets
		two~seats. Observe that $d_2\budget - 2\binSize = \binSize(d_2\binsCount -
		2)$ is the smallest number of votes that any party~$q$ can have after
		moving~$2\binSize$~votes (in particular, this would be the case for any of
		the dummy parties). Thus, as a necessary condition for the preferred party to
		get the second seat, it must hold that
		$$
		\frac{\binSize(d_2\binsCount + 1)}{d_2} \geq \frac{\binSize(d_2\binsCount
			-2)}{d_1}.
		$$
		In words, $\preferredParty$ has to have enough votes to get its second seat
		before $q$ gets its first seat (the inequality is weak due to the
		tie-breaking preferring~$\preferredParty$ over~$q$) assuming that
		all~$2\binSize$~votes are taken from~$q$ and given to~$p$. Otherwise, there
		is enough parties with the number of votes that is strictly bigger than those
		of~$q$ to distribute $\universeSize-1$~seats to them, one seat per party,
		leaving~$p$ with at most one seat.
		We transform the above expression as follows:
		\begin{align*}
			\frac{d_2\binsCount + 1}{d_2} &\geq \frac{d_2\binsCount
				-2}{d_1}\\
			d_1(d_2\binsCount + 1) &\geq d_2(d_2\binsCount-2)\\
			d_1d_2\binsCount + d_1 &\geq (d_2)^2\binsCount-2d_2\\
			\binsCount{}d_2(d_1-d_2) &\geq -2d_2-d_1.
		\end{align*}
		From the fact that the right hand side is negative and both~$\binsCount$
		and~$d_2$ must be positive, it follows that~$d_1 - d_2 \geq 0$, which is a
		contradiction to the fact that the divisor sequence $d_1, d_2, \ldots$ is
		increasing. Eventually, we conclude that more than than~$2\binSize$~budget
		units are required to make party~$\preferredParty$ gain at least two seats in
		a gadget district.
		
		\noindent\emph{Proof of Points~\ref{claim:point:price-of-single-seat}
			and~\ref{claim:point:conditioned-bound-element}.}
		Assume that we perform $\binSize-x$~vote moves resulting in $\preferredParty$
		having $p^* = d_2\budget - x$~votes.
		
		We first show that for $x>0$, party~$\preferredParty$ cannot get any seat.
		Assuming~$x>0$, a necessary condition (yet not sufficient) for
		$\preferredParty$ to obtain a seat is that there is at least $\universeSize +
		\binSize - (\universeSize - 1) = \binSize + 1$~parties with the number of
		votes smaller or equal to~$p^* < d_2\budget$. Indeed,
		since~$\preferredParty$ is first in the tie-breaking order, at most
		$\universeSize -1$~parties can have a higher number of votes
		than~$\preferredParty$. Observe that all parties except of~$\preferredParty$ have
		initially at least $d_2\budget > p^*$~votes. It follows then, that we need
		to move at least $\binSize + 1$~votes to decrease their vote counts. Hence,
		we reach contradiction as for~$x>0$ we have only $\binSize-x <
		\binSize$~moves available.
		
		We now move on to the case of $x=0$, for which we show that $\preferredParty$
		can possibly obtain one seat. This time, the final number of votes of
		$\preferredParty$ is $p^* = d_2\budget$. Importantly, $\preferredParty$ is
		first in the tie-breaking order and the dummy parties have initially exactly
		$d_2\budget$~votes. Hence, we do not need to steal voters from the dummy
		parties as none of them gets allocated a seat before so happens
		for~$\preferredParty$. Because there is exactly $\universeSize$
		element~parties, who have more voters than~$d_2\budget$, to
		make~$\preferredParty$ get one seat it is sufficient to steal votes from one
		element party. Thanks to the assumption that for each $i \in
		[\universeSize]$ it holds that $\binSize \geq \element_i$, we can arbitrarily
		choose such a party~$\partyOf{\element_i}$ and steal~$\element_i$ of its votes. We
		steal the remaining $\binSize - \element_i$~votes from arbitrarily chosen
		parties. This way, we guarantee that~$\preferredParty$ gets one seat as a
		result of bribery at cost~$\binSize$. Additionally, $\partyOf{\element_i}$ is
		exactly the party that loses its seat as a result of not fewer
		than~$\element_i$~vote moves, which proves
		Point~\ref{claim:point:conditioned-bound-element} of the claim.

		\noindent\emph{Proof of Point~\ref{claim:point:bound-element}.} We consider
		the initial state of the district. Without loss of generality, let us fix
		some party~$\partyOf{\element}$ for some~$\element \in \universe$. Recall
		that~$\partyOf{\element}$ initially gets one seat. Due the tie-breaking and
		the number~$\universeSize$ of seats, $\partyOf{\element}$ could lose a seat
		only if either a dummy party or the preferred party gets at least
		the same number of voters after a bribery action. Since $\partyOf{\element}$ initially
		has~$d_2\budget + \element$~votes and each dummy party
		has~$d_2\element$~votes, $\partyOf{\element}$ will tie with some dummy party
		as a result of exactly~$\nicefrac{\element}{2}$~votes moves. Indeed, both
		parties then would have~$d_2\budget + \nicefrac{\element}{2}$~votes (note
		that $\element$ is even by our assumptions on \unarybinpackingstar{} instances).
		Clearly, to make $\partyOf{\element}$ tie with the preferred
		party~$\preferredParty$ would require more vote moves, as $\preferredParty$
		initially has $\binSize$ fewer votes than every dummy party. Importantly, if
		there is another element~$\element' \in \universe$ such that~$\element' <
		\nicefrac{\element}{2}$, then $\partyOf{\element}$ would anyway get a set.
		Namely, in such a case, $\partyOf{\element}$ would be considered by the
		divisor rule earlier than $\partyOf{\element'}$ as the former has a greater
		number of votes. Effectively, $\partyOf{\element}$ would retain its seat.
		This behavior only strenghtens the observation that the
		number~$\nicefrac{\element}{2}$ of vote moves is only a best-case lower
		bound for~$\partyOf{\element}$ losing their seat. We have presented the
		argument only for the case that considered the initial state of the district
		(prior to any bribery action). However, because letting any element party
		lose their seat requires exactly one losing party to gain a seat, we can
		apply the same argument for any achievable situation. Since we have proven
		Point~\ref{claim:point:conditioned-bound-element} above, this concludes the
		proof of the claim.
	\end{proof}	
	
	Notably, the above reduction requires only tiny adaptations to also yield the
	claimed result for the case of the Largest Remainder method. Let us denote
	by~$N$ the number of votes and by~$S$ the number of seats to allocate in the
	bin gadget. Recall that LRM works in two stages, independently in each
	district. In the first stage, it assigns each party~$P$ in the gadget exactly
	$\quota(\sigma(P),N,S) = \lfloor \sigma(P) \cdot \nicefrac{S}{N}\rfloor$ lower~quota~seats,
	where~$p$ denotes the number of votes for~$P$. In the second stage, if there
	are still $r>0$~unallocated seats, $r$ parties~$P$ with the highest number
	$\remainder(\sigma(P),N,S) = \sigma(P)\cdot\nicefrac{S}{N} - \quota(\sigma(P),N,S)$ get
	$\remainder(\sigma(P),N,S)$ remainder seats. 
	
	First, we will show that no seats are allocated in the first stage, i.e., $\lfloor \sigma(P) \cdot \nicefrac{S}{N}\rfloor = 0$. Since $d_2>1$, let us adapt the above reduction by taking~$d_2 = 1$ and let us redefine $\delta$
	to be~$\max\{\universeSize^2, \binSize\}$. Now, calculating a lower bound on the
	value~$N$ for the updated bin gadget,
	\begin{align*}
		N &= \budget - \binSize + \budget\delta +
		\sum_{\element_i \in \universe} (\budget+\element_i) \\
		&= \budget - \binSize + \budget\delta + \universeSize\budget + \budget \\
		&= \budget(2+\delta+\universeSize) -
		\binSize  \\
		&\geq \budget(\universeSize + \universeSize^2),
	\end{align*}
	
	Let us take $S = \universeSize$. Then, we get that
	$$
	\nicefrac{S}{N} \leq \nicefrac{\universeSize}{(\budget(\universeSize +
		\universeSize^2))}.
	$$
	The construction of the reduction ensures that the maximum number~$\hat{\sigma}(P)$ of
	votes that a party in the bin gadget can get after a bribery is that of some
	element party increased by budget~$\budget$ (as a result of a budget-exceeding
	bribery). Hence, taking~$\binSize$ as a trivial upper-bound of a single element
	size, we have that~$\hat{\sigma}(P) \leq \budget + \binSize + \budget \leq 3\budget$.
	It follows that each party in the gadget cannot get any lower quota seats
	because
	\begin{align*}
		\quota(\hat{\sigma}(P),N,S) &= \lfloor \sigma(P) \cdot \nicefrac{S}{N}\rfloor \leq
		\left\lfloor\frac{3\budget\universeSize}{\budget(\universeSize+\universeSize^2)}
		\right\rfloor \\ 
		&= \left\lfloor\frac{3}{\universeSize + 1}\right\rfloor
		= 0
	\end{align*}
	for every positive value of~$\budget$ and $\universeSize \geq 3$.
	Consequently, no party can obtain more than two seats, which, in particular,
	proves Point~\ref{claim:point:no-two-seats} of
	Claim~\ref{claim:bin-gadget-properties}. Other consequence of the fact that no
	lower quota seats are assigned is that in each district LRM always assign seats
	to $\universeSize$~parties with the largest support (as $\nicefrac{S}{N}$ is a
	constant independent from parties support), breaking ties accordingly to the
	tie-breaking order. Even though we did not state it explicitly, in our proof of
	Claim~\ref{claim:bin-gadget-properties} (except that of
	Point~\ref{claim:point:no-two-seats}), we always considered such cases, in
	which no party obtains two seats. In these situations, LRM and divisor sequence methods
	behave exactly the same, so Claim~\ref{claim:bin-gadget-properties} carries
	over intact. For the same reason, the whole argument for the correctness of the
	reduction is valid for the LRM case. 
	\let\universe\undefined
	\let\universeSize\undefined
	\let\element\undefined
	\let\binSize\undefined
	\let\binsCount\undefined
	\let\preferredParty\undefined
	\let\partyOf\undefined
	\let\dummyParties\undefined
\end{proof}

	Again, the hardness result also remains for the destructive case, which we show by adapting the previous proof.

	\begin{theorem}\label{thm:wone-hard-multi-destructive}
	  Let $\appmethod$ be a divisor sequence method or LRM.
          $\appmethod$-\textsc{DMAWB} is~$W[1]$-hard with respect to the number of districts.
	\end{theorem}
\begin{proof}
  \newcommand{\despised}{\ensuremath{P_{\textrm{d}}}}
  We extend the reduction from the proof of
  Theorem~\ref{thm:wone-hard-districts}, and obtain a parameterized,
  polynomial-time reduction from
  \unarybinpackingstar{} to $\appmethod$-\textsc{DMAWB}.
  Using the same notation, fix an instance of~\unarybinpackingstar{}
  with~\namedorderedsetof{\universe}{\element}{\universeSize}, a bin
  size~\binSize{}, and the number~\binsCount{} of bins. 
	
  Let us fix a divisor sequence method~$\calR$ with sequence~$d_1, d_2, \ldots$. We form
  a new instance of~\multiPluralityAG{$\calR$}, with budget~$\budget =
  \binsCount \binSize$ as follows. First, we copy the construction from the proof
  Theorem~\ref{thm:wone-hard-districts}. In addition to the preferred
  party~$P^*$, $\delta$~\emph{dummy parties}, and \universeSize{} \emph{element
  parties} from the copied construction, we introduce the \emph{despised}
  party~$\despised$. 
  Further, we add two \emph{auxiliary districts} to the election. One, where there
  is only the despised party~$\despised$ getting $2\universeSize{}\binsCount$ seats, and another
  one only with the preferred party~$P^*$ getting $(2\universeSize{}-1)\binsCount +1$
  seats. 
  We set the tie-breaking order such that it favors, in order, the preferred
  party, the dummy parties, element parties, and the despised party. The order
  within the groups of the dummy and the element parties is arbitrary but fixed.
  Clearly, the described instance can be constructed in polynomial-time.
	
  It is easy to verify (see also the proof of
  Theorem~\ref{thm:wone-hard-districts}) that the created election results in
  each element's party having exactly \binsCount{}~seats, the dummy parties
  having no seats at all, the preferred party having
  $(2\universeSize{}-1)\binsCount+1$ (no seats from the bin gadgets), and the
  despised party having $2\universeSize{}\binsCount$~seats.
  Our goal is to make some party to have more seats than $\despised$, which
  originally gets the most seats.

  First, note that the seats allocated in the auxiliary districts cannot be
  reallocated to any other party. Observe that, compared to the construction from
  the proof of Theorem~\ref{thm:wone-hard-districts}, these were the only newly
  introduced districts. So, from now on we focus on the remaining districts,
  i.e.,\ the bin gadgets. Due to the properties of the bin gadgets stated
  in~Claim~\ref{claim:bin-gadget-properties}, no
  parties except $P^*$ and $\despised$ can have
  more than $\universeSize\binsCount$ seats, therefore, only $P^*$ may get more seats than
  $\despised$.
  Since there are~\binsCount{} bin gadgets and $P^*$ gains at most one seat from
  a single gadget, $P^*$~can have at most $(2\universeSize{}-1)\binsCount +1 +
  \binsCount{}=2\universeSize\binsCount +1$~seats. Importantly,
  party~$\despised$ has exactly~$2\universeSize\binsCount$~seats (and cannot
  lose any seat in its respective auxiliary district). Thus $P^*$ has to get
  exactly~$k$ additional seats, one per each bin gadget, in order to have more seats
  than $P^*$ has. This condition on the victory of~$P^*$ is equivalent to that from the
  proof of Theorem~\ref{thm:wone-hard-districts}, so by its proof it follows that
  $\appmethod$-\textsc{DMAWB}
  is~\wonehard{} with respect to the number of districts.
  \let\despised\undefined
\end{proof}

In the constructive case, \np-completeness and $W[1]$-hardness for the
parameterization by the number of districts also holds under FPTP, which we
again show by reducing from \unarybinpackingstar{}.

\begin{theorem}\label{thm:fptp-np-hard-wone-hard-winner}
	\multiPluralityAG{FPTP} is $\np$-complete, and 
	$W[1]$-hard with respect to the number of
        districts.
\end{theorem}
	\begin{proof}
		\newcommand{\party}{\ensuremath{P}}
		\newcommand{\despisedParty}{\ensuremath{\party_\mathrm{d}}}
		\newcommand{\districtOf}[1]{\ensuremath{D(#1)}}
		\newcommand{\constant}{\ensuremath{X}}
		Membership in $\np$ is obvious as FPTP is polynomial-time computable. 	
    To prove the claimed hardness results, we give a parameterized,
    polynomial-time reduction from \unarybinpackingstar{}. We thus fix an
    instance of~\unarybinpackingstar{} consisting of a
    collection \namedorderedsetof{\universe}{\element}{\universeSize}, a bin
    size~\binSize{}, and~\binsCount{}~bins. 

    In a new instance
		of~\multiPluralityAG{FPTP} that we create, we use budget~$\budget =
		\binsCount
		\binSize$. We create a preferred party~\preferredParty{}, a despised
		party~\despisedParty{}, and a \emph{bin party}~$\party_i$ for each bin~$i \in
		[\binsCount]$. We now proceed with the reduction description
		using a ``large enough'' constant~$\constant{} = 4\budget$ to ease the
		presentation. For each element~$\element \in \universe$, we build an
		\emph{element district}~$\districtOf{u}$ with exactly~$\element$~seats, in
		which the preferred party gets no votes, the despised party gets $(\constant +
		\element)$~votes, and each bin party gets $(\constant-\element)$~votes.
		Further, we create a collection of \emph{dummy districts}, whose role is to
		secure a certain number of seats to particular parties. The first two districts
		have, respectively, $(\constant-1)$ and $(\constant)$~seats to allocate and have
		exactly $\constant$~voters. In the first district all voters
		support~$\despisedParty$, and in the second one the respective voters
		support~$\preferredParty$. Analogously, we prepare further $\binsCount$~dummy
		districts, one for every~$i \in [\binsCount]$, each of them with
		$(\constant-\binSize - 1)$~seats and with exactly $\constant$~voters, all voting
		for the respective bin party~$\party_i$. Our goal in the new instance is to let
		party $\preferredParty$ get the highest number of seats after performing the
		winning campaign of cost at most~$\budget$. We assume that in case of internal
		ties, $\preferredParty$ is favored over all other parties, and the bin parties
		are preferred over the despised party. Due to the unary encoding of the
		elements in the original instance, the new instance is naturally constructed in
		polynomial-time.
		
		Prior to proving the correctness of the reduction, let us stress important
		features of the constructions. First of all, note a very large margin of
		victory~$\constant > 2\budget$ of parties winning in their respective dummy
		districts. As a result, independently of a performed campaign costing at most
		budget~$\budget$, the preferred and despised parties get, respectively, at least
		$\constant$ and $(\constant-1)$~seats each, and the element parties get at least
		$(\constant-\binSize - 1)$~seats (from the dummy districts) each. Second,
		$\constant$ is also the upper-bound on the total number of seats
		that $\preferredParty$ gets after every possible campaign with cost at
		most~$\budget$,
		i.e., it is impossible for the party to get any additional seat. 
		It holds because in each element district, after implementing
		any campaign, party~$\preferredParty$ gets at most $\budget$~votes, while the
		other parties have at least~$(\constant - \binSize - \budget) \geq
		2\budget$~votes. Eventually, from now on we assume that~$\preferredParty$
		always gets exactly $\constant$~seats, and so, it is never optimal to transfer
		votes to or from~$\preferredParty$ in the course of a successful campaign. 
		
		Let $S_1, S_2, \ldots, S_\binsCount$ be some packing~$\mathcal{S}$ of elements
		from~$\universe$ into $\binsCount$~bins witnessing a solution to the
		original instance. We transform this packing to the solution of the
		corresponding \multiPluralityAG{FPTP}~instance as follows. For each
		element~$\element \in \universe$ and its respective bin~$S_j$
		containing $\element$ in packing~$\mathcal{S}$, we transfer $\element$~votes
		from~$\despisedParty$ to~party~$\party_j$ in district~$\districtOf{\element}$
		corresponding to~$\element$. Since the sum of elements of~$\universe$
		is~$\budget$, such a bribery does not exceed budget~$\budget$. Furthermore,
		after the campaign, in each element district~$\districtOf{\element}$, $\element
		\in \universe$, it holds that party~$\party_j$ corresponding to bin~$S_j$
		containing $\element$ as well as the despised party have exactly $\constant$~vote. Thus
		the $\element$~seats of district~$\districtOf{\element}$ are allocated
		to~$\party_j$. Eventually, since each bin~$S_i$, $i \in [\binsCount]$, contains
		elements whose total sum is~$\binSize$, we obtain that $\despisedParty$
		loses $\binsCount\binSize = \budget$~seats ending up with $(\constant - 1)$~seats in
		total. Conversely, each bin party gets $\binSize$~additional seats, which gives
		it altogether~$\constant - 1$~seats. Since the preferred party has the largest
		number of seats, namely $\constant$, we obtained a valid solution to the
		corresponding constructed instance of~\multiPluralityAG{FPTP}.
		
		We now prove the opposite directions, that is, that if there is a successful
		campaing that makes the preferred party have more seats that any other party,
		then this campaign can be transformed to a sought packing of the elements
		of~$\universe$. First, recall that prior to a bribery the preferred
		party gets $\constant$~seats, the despised party
		gets~$(\constant+\binsCount\binSize-1)$~seats, and each of the bin parties
		gets~$(\constant-\binSize-1)$~seats. Additionally, as we have shown earlier,
		the preferred party has to get exactly $\constant$~seats after every campaign
		that does not exceed the budget. Hence, a successful campaing decreases the
		number of seats allocated to the despised party to at most~$(\constant-1)$~seats.
		Since we already have shown that the despised party has to get at least this
		number of seats, it follows that the despised party eventually gets
		exactly~$(\constant-1)$~seats. That is, it loses $\binsCount\binSize$~seats
		as a result of a successful campaign. Let us now focus on the element
		districts, which, prior to any campaign, in total give
		$\binsCount\binSize$~seats to the despised party. Consider some
		district~$\districtOf{\element}$, $\element \in \universe$. 
		To take $\element$~seats from~$\despisedParty$ in this discrict, one needs to move at
		least~$\element$ votes from~$\despisedParty$; furthermore, at least~$\element$
		of these votes need to be moved to one, arbitrarily chosen bin
		party~$\party_i$, $i \in [\binsCount]$. In such a case, $\despisedParty$ is
		left with at most $\constant$~votes and the chosen bin party obtains at
		least~$(\constant - \element + \element)=\constant$~votes. As a result, due to
		the tie-breaking, $\party_i$ gets all~$\element$~seats of the district and
		party~$\despisedParty$~loses all its~$\element$~seats. Extending this argument
		to all element districts, we obtain that each successful campaing uses exactly
		$\binsCount\binSize = \budget$~vote moves to take, in total,
		$\binsCount\binSize$~seats from~$\despisedParty$ in the element districts.
		Furthermore, by the fact the we use $\binsCount\binSize$~units of budget to
		take away exactly~$\binsCount\binSize$~seats from~$\despisedParty$, we know
		that for each element~$\element \in \universe$, a successful campaing spends
		exactly $\element$~units of budget in the respective element
		district~$\districtOf{\element}$. Consequently, after the campaign, the seats
		from district~$\districtOf{\element}$ are allocated to some party~$P_i$, $i
		\in [\binsCount]$. By the fact that prior to any campaingn each bin party has
		exactly~$(\constant - \binSize - 1)$~seats and due to the condition that the
		preferred party has exactly $\constant$~seats, we know that a successful
		campaingn results in at most $\binSize$~additional seats for each bin party. In
		fact, it must be exactly $\binSize$~additional seats as there is
		exactly $\binsCount$~bin parties and a successful campaing takes
		exactly $\binSize\binsCount$~seats from~$\despisedParty$ in total (and no other
		seats). Concluding the proof, we obtain a correct packing by assigning each
		element~$\element \in \universe$ to the bin that represents the bin party which
		gets seats in the element's district~$\districtOf{\element}$ in a successful
		campaign action.
		\let\universe\undefined
		\let\universeSize\undefined
		\let\element\undefined
		\let\binSize\undefined
		\let\binsCount\undefined
		\let\preferredParty\undefined
		\let\partyOf\undefined
		\let\dummyParties\undefined
	\end{proof}
For the above problems, i.e., for $\appmethod$-\textsc{MAWB} when $\appmethod$ is a divisor sequence method, LRM, or FPTP, and for $\appmethod$-\textsc{DMAWB} when $\appmethod$ is a divisor sequence method or LRM, 
we have not been able to find a parameterized hardness reduction nor to find an $\fpt$ algorithm for the parameter ``number of parties.'' A natural approach would be, e.g., to form an integer linear program and solve it using an $\fpt$ algorithm; unfortunately, for Lenstra's algorithm~\cite{len:j:integer-fixed} we end up with too many
variables, and for many other approaches, including $n$-fold IP, we end up with two large coefficients in the constraints; see the overview of Gavenciak et al.~\cite{gav-kno-kou:c:ilp-guidebook}.

In contrast to all other multi-district winner problems, FPTP-\textsc{DMAWB} is the only one that can be solved in polynomial time.

\begin{theorem}\label{thm:fptp-dmawb-in-p}
	FPTP-\textsc{DMAWB} is polynomial-time solvable.
\end{theorem}
\begin{proof}
	\newcommand{\party}{\ensuremath{P}}
	\newcommand{\difference}{\ensuremath{\Delta}}
	\newcommand{\bribery}{\ensuremath{b}}
	Let $\preferredParty$ be the distinguished party winning a given election with $q$ districts. We seek for a successful destructive campaign
	against~$\preferredParty$ using
	at most $\budget$~units of budget. We first guess a party~$\party \in
	\parties_{-\preferredParty}$ and compute a difference~$d$ of
	the number of seats allocated to party~$\preferredParty$ and~$\party$. 
	We add one to this difference and denote the resulting
	number by~$\difference$, i.e., $\difference = d + 1$. Clearly $\difference$ is positive integer, as
	otherwise~$\party$ already gets more seats than~$\preferredParty$. Next we apply
	the algorithm from~\Cref{prop:fptp} and, for each district~$i \in [q]$, we
	compute the smallest number~$\bribery_i$ of vote changes to make~$\party$
	win the seats in district~$i$; we set~$\bribery_i$ to infinity if $\party$
	already wins in~$i$. Using the computed values, we construct a
	\textsc{Knapsack} instance with the items corresponding to the districts.
	Specifically, for some item~$i \in [q]$, we let its weight be~$\bribery_i$
	and its value be either~$\seats_i$ or~$2\seats_i$, depending on whether
	$\preferredParty$ was initially losing or winning in district~$i$, respectively, which corresponds to the difference between the difference of seats allocated to $\preferredParty$ and $\party$ after the seats' transfer in district $i$ and before that. We
	then solve the \textsc{Knapsack} instance for value~$\difference$ subject to
	weight limit~$\budget$ and record the output. After running the above
	described procedure for each possible guess of party~$\party$, we output
	``yes'' if at least one of the guesses led to a positive \textsc{Knapsack}
	instance.
	
	Since there is a polynomial number of guesses and the weights of the
	constructed \textsc{Knapsack} instances are polynomially-bounded with
	respect to the input of our original instance, our algorithm runs in
	polynomial time. The correctess of the algorithm is a straightforward
	conclusion from the way how we compute the weights and the sizes of the
	items for the constructed \textsc{Knapsack} instances.
	\let\bribery\undefined
	\let\difference\undefined
	\let\party\undefined
	\let\preferredParty\undefined
\end{proof}

	\subsection{Experiments}\label{sec:experiments}

        As we have shown in the previous section, computing successful and even optimal campaigns for
        apportionment bribery in single or multiple districts is computationally tractable for both the constructive and the destructive case.
This indicates that such an attack would be relatively simple for a campaign manager to execute (at least from a computational standpoint). 
However, it remains to be investigated how \emph{effective} strategic campaigns are in real-world scenarios.
In this section, we thus experimentally study the effectiveness of such campaigns with regard to the following three questions:

\begin{enumerate}
	\item How effective are our optimal algorithms compared to simple heuristic bribing strategies, measured by the additional budget required %
	to increase or decrease a distinguished party’s number of seats by one?
	\item How does the electoral threshold impact the effectiveness of strategic campaigns, measured by the number of seats that can be gained for or taken away from a distinguished party in an optimal %
	campaign?
	\item How does the number of districts affect the difficulty of changing the election outcome, measured by the minimal budget required?
\end{enumerate}

We use three datasets from single-district elections and three datasets from multi-district elections.\footnote{The code and data used for our experiments can be found at
\url{https://github.com/bredereck/Strategic-Campaigns-in-Apportionment-Elections}.}
As the single-district elections, we use the Austrian parliamentary election from 2019, the Israel parliamentary election from 2022, and the Dutch parliamentary election from 2023.
For the multi-district elections, we use the Polish parliamentary election from 2023, the Argentinian parliamentary election from 2021, and the Portuguese parliamentary election from 2024.
These countries were chosen as they all use the D'Hondt method.
An overview of the number of districts, the number of seats in the parliament, and the threshold for each of those countries is given in Table~\ref{tab:datasets_overview}.\footnote{Note that the elections for the Argentinian \emph{Chamber of Deputies} are held every two years for half of the seats, i.e., there are 257 seats in total.}

\begin{table}[t]
	\caption{Overview of the number of districts, seats, thresholds, and total vote counts for all datasets chosen for our experiments.}
  \label{tab:datasets_overview}
  \centering
	{\scriptsize
	\begin{tabular}{lcccccc}
		& \textbf{Austria} & \textbf{Israel} & \textbf{The Netherlands} & \textbf{Poland} & \textbf{Portugal} & \textbf{Argentinia} \\
		& \textbf{2019} & \textbf{2022} & \textbf{2023} & \textbf{2023} & \textbf{2024} & \textbf{2021} \\
		\midrule
		num.\ of districts & 1 & 1 & 1 & 41 & 22 & 24 \\
		num.\ of seats & 183  & 120  & 150 & 460 & 230 & 127 \\
		threshold & $4\%$  & $3.25\%$  & $\nicefrac{1}{150}$  & $5\%$ & $0\%$ & $3\%$\\
		num.\ of (valid) votes & 4,777,246 & 4,764,742 & 10,432,526 & 21,596,674 & 6,194,290 & 23,602,493
	\end{tabular}
}
\end{table}

\subsubsection{Effectiveness of Optimal Campaigns}\label{subsec:experiments-heuristics}

In our initial experiment, we evaluate the effectiveness of an optimal bribing strategy compared to simpler heuristics in both the constructive and the destructive scenario.
Specifically, we assess the additional budget required by the heuristics to achieve the same goal of increasing (in the constructive case) or decreasing (in the destructive case) a given party's number of seats by one, compared to our optimal algorithm.  
Our exact algorithms sometimes result in rather complicated strategic campaigns (mostly in terms of voter movement). 
A simple heuristic produces minimally complex campaigns and is thus easy to realize.

We start with the constructive case, where we consider the following bribing strategies for party $P_i$:
\begin{enumerate}%
	\item \emph{Optimal bribery}: We move votes from other parties to $P_i$ optimally (as in the algorithms from Sections~\ref{subsec:top-choice-single-district} and~\ref{subsec:multi-district}, but using $P_i$ in the place of~$P_1$).
	\item \emph{Balanced bribery}: We move votes from the other parties to $P_i$ proportionally to their vote counts (this is a derandomized analogue of bribing votes uniformly at random).
	\item \emph{Weakest rival}: We move votes from the weakest party to $P_i$ (and if this rival party does not have enough votes, we move votes from the second-weakest rival, and so on).\footnote{By the weakest party we mean the one that originally had the fewest votes.}	
	\item \emph{Strongest rival}: We move votes from the strongest party to $P_i$ (and if this rival party does not have enough votes, we move votes from the second-strongest rival, and so on).\footnote{By the strongest party we mean the one that originally had the most votes.}	
\end{enumerate}
Note that an optimal bribery has very high information
requirements (exact vote counts), whereas the other forms require much less information. Thus one can also view the latter three strategies
as bribing strategies based on realistic information assumptions.

In each dataset (i.e., in each country), we assume the same threshold for all parties and all districts
(see Table~\ref{tab:datasets_overview} for an overview of the actual thresholds of the countries).
To gain seats in a district, a party must reach the threshold in that district.  
For each country, we compute the number of votes that need to be changed for the distinguished party to win an additional seat by applying the above mentioned bribing strategies $1$ to $4$. Note that in some cases, the number of additional seats is greater than one because a party has reached the threshold.
We calculate the efficiency of the different strategies as follows:
From each of the algorithms, we calculate the budget needed to achieve that goal, i.e., the number of votes that need to be changed for the chosen bribery strategy for the distinguished party to gain an additional seat, and we divide it by the value of the optimal algorithm, i.e., the minimum number of votes that need to be changed for the distinguished party to gain an additional seat. 

The results for the constructive case are shown in Table~\ref{tab:eff-one-seat} for the single-district elections, and in Table~\ref{tab:eff-one-seat-multi} for the multi-district elections.
Intuitively, we want to know how many more votes we need to change with a simple heuristic bribery strategy compared to the optimal one.
For example, a concrete number like ``$2$'' in Tables~\ref{tab:eff-one-seat} and~\ref{tab:eff-one-seat-multi} indicates that twice as many votes would have to be changed in the corresponding heuristic algorithm than in the optimal one in order for the distinguished party to gain an additional seat. 
Note that in some cases doubling the number of votes can even mean a difference of thousands of votes that would have to be bribed more in the simple strategy, even though it is only a matter of changing one seat. 
Lastly, in the ``average'' column of the tables, we treat the average results across all parties as the data for the distinguished party;\footnote{More specifically, for each country, we compute the average results over all parties by calculating the arithmetic mean of all budgets obtained with a specific algorithm and dividing this number by the arithmetic mean of the results of the optimal algorithm.} and we choose as our distinguished party the strongest party in the next column and the weakest party in the rightmost column. 
\begin{table*}[h]
	\caption{Effectiveness of different bribing strategies for one additional seat.}
	\scriptsize
	\begin{subtable}{0.5\textwidth}		
		\caption{Single-district countries.}
		\label{tab:eff-one-seat}
		\centering
		\begin{tabular}{cllll}
			& &         & strongest & weakest \\
			& & average & party & party\\
			\midrule
			\parbox[t]{2mm}{\multirow{4}{*}{\rotatebox[origin=c]{90}{Austria}}} &
			\ \ optimal bribery & 1 & 1 &  1  \\
			& \ \ balanced bribery & 1.01406  & 2.07465 & 1  \\
			& \ \ weakest rival &  1.0256 & 3.3224 &1 \\
			& \ \ strongest rival & 1.01547 & 1 & 1 \\
			\midrule
			\parbox[t]{2mm}{\multirow{4}{*}{\rotatebox[origin=c]{90}{Israel}}} &
			\ \ optimal bribery & 1 & 1 & 1 \\
			& \ \ balanced bribery & 1.01520 & 3.84078 & 1 \\
			& \ \ weakest rival & 1.01984 & 5.02418 & 1 \\
			& \ \ strongest rival &  1.01659 & 1 & 1 \\
			\midrule
			\parbox[t]{2mm}{\multirow{4}{*}{\rotatebox[origin=c]{90}{Netherlands}}} &
			\ \ optimal bribery & 1 & 1 &  1 \\
			& \ \ balanced bribery &  1.17935 &  4.99234 & 1 \\
			& \ \ weakest rival  & 1.22679  & 6.552335 & 1 \\
			& \ \ strongest rival & 1.22383 & 6.46223 & 1 
		\end{tabular}
	\end{subtable}
	\begin{subtable}{0.5\textwidth}		
		\caption{Multi-district countries.}
		\label{tab:eff-one-seat-multi}
		\centering
		\begin{tabular}{cllll}
			& &         & strongest  & weakest \\
			& & average & party  & party\\
			\midrule
			\parbox[t]{2mm}{\multirow{4}{*}{\rotatebox[origin=c]{90}{Poland}}} &
			\ \ optimal bribery & 1 & 1 & 1  \\
			& \ \ balanced bribery & 1.08521  & 1.65625 & 1.07160  \\
			& \ \ weakest rival & 1.17646 & 3.18750 & 1.16582 \\
			& \ \ strongest rival & 1.10409  & 1 & 1.07399 \\
			\midrule
			\parbox[t]{2mm}{\multirow{4}{*}{\rotatebox[origin=c]{90}{Portugal}}} &
			\ \ optimal bribery & 1 & 1 & 1 \\
			& \ \ balanced bribery & 1.14033 & 1.447489 &  1.12533 \\
			& \ \ weakest rival & 1.34459 & 2.32420 &1.31880 \\
			& \ \ strongest rival & 1.24686 & 1 & 1.24862\\
			\midrule
			\parbox[t]{2mm}{\multirow{4}{*}{\rotatebox[origin=c]{90}{Argentina}}} &
			\ \ optimal bribery & 1 & 1  & 1 \\
			& \ \ balanced bribery & 1.26643 & 1.07874 & 1.46548 \\
			& \ \ weakest rival & 1.46549 & 1.50394 &   2 \\
			& \ \ strongest rival & 1.21806 & 1 & 2 
		\end{tabular}
	\end{subtable}			
\end{table*} 

One can see that in most cases, optimal bribery is (significantly) more effective than simple bribing strategies, i.e., it requires a smaller budget for one additional seat. 
The largest difference in effectiveness can usually be seen when the strongest party is chosen as the distinguished one. 
For example, for the \emph{Party for Freedom} in the Dutch parliamentary election in 2023, the other simple bribery strategies require between $4.99$ and $6.55$ times as many votes as the optimal bribery by our algorithm.
However, there are also cases such as those where the weakest party is chosen as the distinguished party in the single-district countries, in which all algorithms have to move the same (minimum) number of votes for the distinguished party to win a seat. 
The reason for this is that the number of votes the parties received is below the threshold, and it is enough to make it reach the threshold to win a seat (and sometimes even more directly), without it making much difference from which parties the votes are taken. 
When choosing
among simple bribing strategies, taking votes proportionally from the other parties (i.e., using balanced bribery) appears to be the most promising one.

For the destructive case, we consider similar bribing strategies as in the constructive case: 		
\begin{enumerate}%
	\item \emph{Optimal bribery}: We move votes from $P_i$ to other parties optimally (as in the algorithms from Sections~\ref{subsec:top-choice-single-district} and~\ref{subsec:multi-district}, but using $P_i$ in the place of~$P_1$).
	\item \emph{Balanced bribery}: We move votes from $P_i$ to the other parties proportionally to their vote counts (this is a derandomized analogue of bribing votes uniformly at random).
	\item \emph{Weakest rival}: We move votes from $P_i$ to the weakest party.
	\item \emph{Strongest rival}: We move votes from $P_i$ to the strongest party.
\end{enumerate}

We have conducted the experiments analogously to the constructive case, except that our goal now is to calculate the number of votes that must be changed in order for the distinguished party to \emph{lose} a seat.
Note that in the destructive case, we do not take into account the parties that do not receive any seats, i.e., we do not select those parties as a distinguished party. 
The results are summarized in Table~\ref{tab:eff-minus-one-seat} for the elections with a single district and in Table~\ref{tab:eff-minus-one-seat-multi} for the elections with multiple districts.

The results for the destructive case are similar to those for the constructive case. 
Again, optimal bribery is (significantly) more effective than simple bribery strategies in almost all cases, and the largest difference in effectiveness occurs with the strongest party as the distinguished party, e.g., in the case of \emph{Likud}---the strongest party in the Israeli election in 2022---our algorithm delivers much cheaper bribery in comparison to the other strategies: Here, the simple bribing strategies need up to $5.59$ times as many votes as our optimal algorithm.
Again, there are cases where the budget is independent of the algorithm used (but fewer than in the constructive case), e.g., if a party only has one seat and only needs to lose enough votes to be pushed below the threshold.
This can be observed, for example, in the Dutch election, where the weakest party with at least one seat---the \emph{JA21}---is chosen as the distinguished party: Here, all algorithms have to move the same (minimum) number of votes for the distinguished party to lose a seat.
Lastly, when choosing among simple bribing strategies, using balanced bribery or taking votes from the strongest rival appears to be the most promising depending on the situation. 
However, it should be noted that the values for all simple strategies are close to each other in several cases.

\begin{table*}[h]
	\caption{Effectiveness of different bribing strategies for taking one seat from a party.}
	\scriptsize
	\begin{subtable}{0.5\textwidth}		
		\caption{Single-district-countries.}
		\label{tab:eff-minus-one-seat}
		\centering
		\begin{tabular}{c@{\hspace*{2mm}}l@{\hspace*{2mm}}l@{\hspace*{2mm}}l@{\hspace*{2mm}}l}
			& &         & strongest & weakest party  \\
			& & average & party & with seat \\
			\midrule
			\parbox[t]{2mm}{\multirow{4}{*}{\rotatebox[origin=c]{90}{Austria}}} &
			\ \ optimal bribery & 1 & 1 & 1   \\
			& \ \ balanced bribery &  1.26931 & 1.69816 & 1.08033 \\
			& \ \ weakest rival &  1.66797  & 2.69417 &  1.17211  \\
			& \ \ strongest rival &  1.22228 & 1.43293 &  1.08417  \\
			\midrule
			\parbox[t]{2mm}{\multirow{4}{*}{\rotatebox[origin=c]{90}{Israel}}} &
			\ \ optimal bribery & 1 & 1 & 1 \\
			& \ \ balanced bribery & 1.93924 & 4.30920 & 1.56716  \\
			& \ \ weakest rival & 2.16676 & 5.59157 & 1.625 \\
			& \ \ strongest rival & 2.02504 & 5.44377 & 1.625 \\
			\midrule
			\parbox[t]{2mm}{\multirow{4}{*}{\rotatebox[origin=c]{90}{Netherlands}}} &
			\ \ optimal bribery & 1 & 1  & 1 \\
			& \ \ balanced bribery & 1.1388 & 1.63824 & 1  \\
			& \ \ weakest rival & 1.18093 &  2.13672 & 1 \\
			& \ \ strongest rival & 1.17158  & 2.13672  & 1 
		\end{tabular}
	\end{subtable}
	\begin{subtable}{0.5\textwidth}		
		\caption{Multi-district countries.}
		\label{tab:eff-minus-one-seat-multi}
		\centering
		\begin{tabular}{c@{\hspace*{2mm}}l@{\hspace*{2mm}}l@{\hspace*{2mm}}l@{\hspace*{2mm}}l}
			& &         & strongest &  weakest party \\
			& & average & party & with seat \\
			\midrule
			\parbox[t]{2mm}{\multirow{4}{*}{\rotatebox[origin=c]{90}{Poland}}} &
			\ \ optimal bribery & 1 & 1 & 1   \\
			& \ \ balanced bribery & 1.17703 & 2.02857 &  1.12121   \\
			& \ \ weakest rival & 1.40546 & 2.99286 & 1.21212  \\
			& \ \ strongest rival & 1.17703 & 2.99286 & 1   \\
			\midrule
			\parbox[t]{2mm}{\multirow{4}{*}{\rotatebox[origin=c]{90}{Portugal}}} &
			\ \ optimal bribery & 1 & 1 & 1 \\
			& \ \ balanced bribery & 1.21057 & 1.25109 & 1.0703  \\
			& \ \ weakest rival & 1.40685 & 1.5 & 1.09050  \\
			& \ \ strongest rival & 1.38736 & 1 & 1.09050  \\
			\midrule
			\parbox[t]{2mm}{\multirow{4}{*}{\rotatebox[origin=c]{90}{Argentina}}} &
			\ \ optimal bribery & 1 &  1 & 1  \\
			& \ \ balanced bribery & 1.1331 & 1.46196  & 1.08754 \\
			& \ \ weakest rival & 1.27494 & 2  & 1.18748  \\
			& \ \ strongest rival & 1.19092 & 2  & 1.12473  
		\end{tabular}
	\end{subtable}
\end{table*}

\subsubsection{Influence of the Threshold}

In the second experiment, we want to find out what influence an electoral threshold has on the effectiveness of strategic campaigns. Here we define effectiveness as the number of seats that we can win for a distinguished party $P^*$ in an optimal constructive campaign, or as the number of seats that we can take away from a distinguished party $P^*$ in an optimal destructive campaign.

We conducted our experiments for the three single-district datasets as follows.
In all three elections, we focus on a budget $\budget$ equal to $0.25\%$ of the total vote count.
This is a relatively small fraction of the voters, and we find it plausible that a campaign manager could be able to influence that many voters.
To show the effect of the threshold on the effectiveness of a campaign, we gradually raise the threshold in our experiment. 
As the distinguished party $P^*$ we always choose the
party with the highest voter support in the election, since it reaches all tested thresholds and is thus always present in the parliament. 
Lastly, we use D'Hondt in our experiments as a representative of the divisor sequence methods, since it is one of the most widely used in apportionment elections.

\begin{figure*}[t]
	\begin{tikzpicture}
		\pgfplotsset{every tick label/.append style={font=\small}}
		\begin{axis}[width=0.395\textwidth, height = 4.5cm, legend image code/.code={\draw[#1, fill=#1] (0cm,0cm) circle (0.07cm);}, xmin = 0, xmax = 12]
			
			\addplot+[col_bribery, name path = bribery, no marks] plot [] table {DATA/0.0025/best/AUT2019/DHondt-Bribery-Max-Diff.dat};
		\end{axis}
	\end{tikzpicture}
	\begin{tikzpicture}
		\pgfplotsset{every tick label/.append style={font=\small}}
		\begin{axis}[width=0.395\textwidth, height = 4.5cm, legend image code/.code={\draw[#1] (0cm,0cm) circle (0.07cm);}, xmin = 0, xmax = 12]
			
			\addplot+[col_bribery, name path = bribery, no marks] plot [] table {DATA/0.0025/best/IKE2022/DHondt-Bribery-Max-Diff.dat};
		\end{axis}
	\end{tikzpicture}
	\begin{tikzpicture}
		\pgfplotsset{every tick label/.append style={font=\small}}
		\begin{axis}[width=0.395\textwidth, height = 4.5cm, legend image code/.code={\draw[#1] (0cm,0cm) circle (0.07cm);}, xmin = 0, xmax = 12]
			
			\addplot+[col_bribery, name path = bribery, no marks] plot [] table {DATA/0.0025/best/TWK2023/DHondt-Bribery-Max-Diff.dat};
		\end{axis}
	\end{tikzpicture}
	
	\begin{tikzpicture}
		\pgfplotsset{every tick label/.append style={font=\small}}
		\begin{axis}[width=0.395\textwidth, height = 4.5cm, legend image code/.code={\draw[#1, fill=#1] (0cm,0cm) circle (0.07cm);}, xmin = 0, xmax = 12, xlabel = {Austria}]
			
			\addplot+[col_bribery, name path = bribery, no marks] plot [] table {DATA/0.0025/best/AUT2019/DHondt-Bribery-Min-Diff.dat};
		\end{axis}
	\end{tikzpicture}
	\begin{tikzpicture}
		\pgfplotsset{every tick label/.append style={font=\small}}
		\begin{axis}[width=0.395\textwidth, height = 4.5cm, legend image code/.code={\draw[#1] (0cm,0cm) circle (0.07cm);}, xmin = 0, xmax = 12, xlabel = {Israel}]
			
			\addplot+[col_bribery, name path = bribery, no marks] plot [] table {DATA/0.0025/best/IKE2022/DHondt-Bribery-Min-Diff.dat};
		\end{axis}
	\end{tikzpicture}
	\begin{tikzpicture}
		\pgfplotsset{every tick label/.append style={font=\small}}
		\begin{axis}[width=0.395\textwidth, height = 4.5cm, legend image code/.code={\draw[#1] (0cm,0cm) circle (0.07cm);}, xmin = 0, xmax = 12, xlabel = {The Netherlands}]
			
			\addplot+[col_bribery, name path = bribery, no marks] plot [] table {DATA/0.0025/best/TWK2023/DHondt-Bribery-Min-Diff.dat};
		\end{axis}
	\end{tikzpicture}
	
	\caption{
		The $x$-axis shows a variety of thresholds in percent of the number of voters~$n$.
		The $y$-axis shows the maximally achievable number of additional seats (respectively, prevented seats) for the strongest party by D'Hondt-\textsc{AB} in the top row and by D'Hondt-\textsc{DAB} in the bottom row, each with a given budget of $\budget = 0.0025\cdot n$.
	}\label{fig:experiment-dhondt}
\end{figure*}
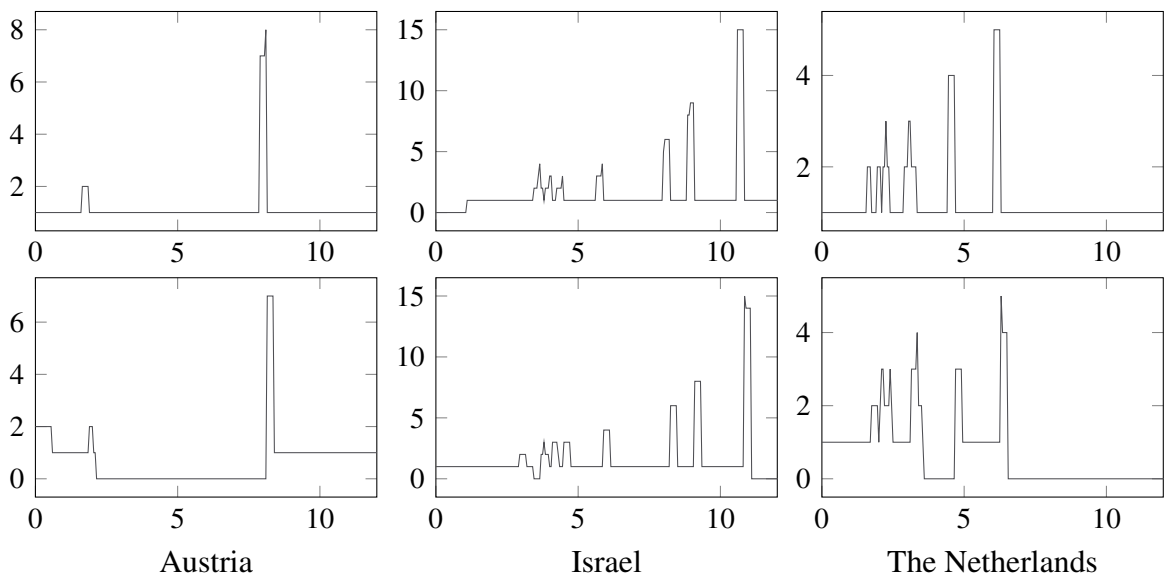

Figure~\ref{fig:experiment-dhondt} illustrates the effectiveness of both the constructive (top row) and destructive (bottom row) campaigns run on the three real-world elections.
One would expect some kind of proportionality, e.g., that $0.25\%$ of the voters control approximately $0.25\%$ of the seats.
Here, for all three countries (Austria, Israel, and the Netherlands), $0.25\%$ of the seats in the parliament correspond to between zero and one seat and this is indeed what we observe for many values of the threshold.
However, there are some spikes where with only $0.25\%$ of votes one can make the distinguished party $P^*$, i.e., the strongest party, gain sometimes $3.3\%$ (the Netherlands), $4.4\%$ (Austria), or even $12.5\%$ (Israel) of all seats on top in the constructive case.
This is considerably more than one can expect from our small budget.
Note that the spikes mainly occur at thresholds where a party is directly above the threshold.
For instance, in Austria we see a peak at thresholds $1.87\%$ and $8.10\%$, which are exactly the values where the parties \emph{JETZT~-- Pilz List} and \emph{NEOS~-- The New Austria} are slightly above the threshold.
This indicates that at these thresholds the campaign is focused on pushing a party below the threshold and freeing up its seats. 

For the destructive case, we can see similar peaks as in the constructive campaigns.
However, this time the peaks are mainly at thresholds where a party is just below the threshold.
That is, the campaign is focused on raising a party above the threshold to steal seats from $P^*$. 
For instance, in Austria there are peaks at thresholds of a little over $0\%$, $1.87\%$, and $8.10\%$, which coincides with there being seven parties with less than $1\%$ of all (valid) votes (\emph{KPÖ}, \emph{WANDL}, \emph{BIER}, \emph{GILT}, \emph{BZÖ}, \emph{SLP}, and \emph{CPÖ}), as well as \emph{JETZT~-- Pilz List} and \emph{NEOS~-- The New Austria} at thresholds $1.87\%$ and $8.10\%$.\footnote{Note that there is no peak at around the $0\%$ threshold in the constructive case, as those parties with less than $1\%$ of all votes cannot free up many seats for the distinguished party, yet they are only assigned one seat in total (specifically, \emph{KPÖ} has one seat) while all other six parties with less support than $1\%$ have no seats.}

Note that we also ran the experiments for other values of the bribery budget $\budget$ and observed the following: 
For smaller budgets, we see narrower (and sometimes lower) peaks right at the thresholds where a party is just above it (respectively, just below it, in the destructive case), while for larger budgets, the peaks become wider (and sometimes also higher).
Figure~\ref{smaller-budget} exemplarily shows the results for $0.15\%$ and Figure~\ref{larger-budget} for $0.35\%$.
Lastly, Figures~\ref{fig:second-best} and~\ref{fig:third-best} also illustrate the effects when we choose the second-strongest or third-strongest party as our distinguished party~$P^*$.
Again, the results are similar. 
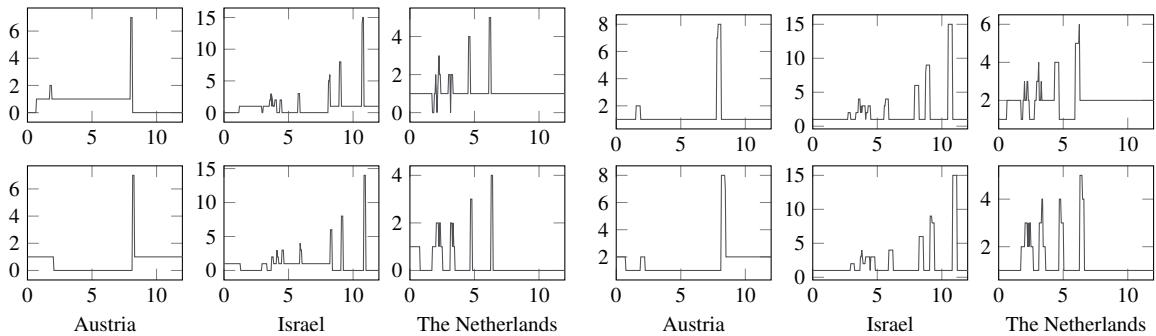
\begin{figure*}[t]
	\begin{subfigure}{0.5\textwidth}
		\scriptsize
		\begin{tikzpicture}
			\begin{axis}[width=0.47\textwidth, height = 3.1cm, legend image code/.code={\draw[#1, fill=#1] (0cm,0cm) circle (0.07cm);}, xmin = 0, xmax = 12]
				
				\addplot+[col_bribery, name path = bribery, no marks] plot [] table {DATA/0.0015/best/AUT2019/DHondt-Bribery-Max-Diff.dat};
			\end{axis}
		\end{tikzpicture}
		\begin{tikzpicture}
			\begin{axis}[width=0.47\textwidth, height = 3.1cm, legend image code/.code={\draw[#1] (0cm,0cm) circle (0.07cm);}, xmin = 0, xmax = 12]
				
				\addplot+[col_bribery, name path = bribery, no marks] plot [] table {DATA/0.0015/best/IKE2022/DHondt-Bribery-Max-Diff.dat};
			\end{axis}
		\end{tikzpicture}
		\begin{tikzpicture}
			\begin{axis}[width=0.47\textwidth, height = 3.1cm, legend image code/.code={\draw[#1] (0cm,0cm) circle (0.07cm);}, xmin = 0, xmax = 12]
				
				\addplot+[col_bribery, name path = bribery, no marks] plot [] table {DATA/0.0015/best/TWK2023/DHondt-Bribery-Max-Diff.dat};
			\end{axis}
		\end{tikzpicture}
		
		\begin{tikzpicture}
			\begin{axis}[width=0.47\textwidth, height = 3.1cm, legend image code/.code={\draw[#1, fill=#1] (0cm,0cm) circle (0.07cm);}, xmin = 0, xmax = 12, xlabel = {Austria}]
				
				\addplot+[col_bribery, name path = bribery, no marks] plot [] table {DATA/0.0015/best/AUT2019/DHondt-Bribery-Min-Diff.dat};
			\end{axis}
		\end{tikzpicture}
		\begin{tikzpicture}
			\begin{axis}[width=0.47\textwidth, height = 3.1cm, legend image code/.code={\draw[#1] (0cm,0cm) circle (0.07cm);}, xmin = 0, xmax = 12, xlabel = {Israel}]
				
				\addplot+[col_bribery, name path = bribery, no marks] plot [] table {DATA/0.0015/best/IKE2022/DHondt-Bribery-Min-Diff.dat};
			\end{axis}
		\end{tikzpicture}
		\begin{tikzpicture}
			\begin{axis}[width=0.47\textwidth, height = 3.1cm, legend image code/.code={\draw[#1] (0cm,0cm) circle (0.07cm);}, xmin = 0, xmax = 12, xlabel = {The Netherlands}]
				
				\addplot+[col_bribery, name path = bribery, no marks] plot [] table {DATA/0.0015/best/TWK2023/DHondt-Bribery-Min-Diff.dat};
			\end{axis}
		\end{tikzpicture}
		
		\caption{Experiment with a budget of $\budget = 0.0015\cdot n$.}
		\label{smaller-budget}
	\end{subfigure}
	\begin{subfigure}{0.5\textwidth}
		\scriptsize
		\begin{tikzpicture}
			\begin{axis}[width=0.47\textwidth, height = 3.1cm, legend image code/.code={\draw[#1, fill=#1] (0cm,0cm) circle (0.07cm);}, xmin = 0, xmax = 12]
				
				\addplot+[col_bribery, name path = bribery, no marks] plot [] table {DATA/0.0035/best/AUT2019/DHondt-Bribery-Max-Diff.dat};
			\end{axis}
		\end{tikzpicture}
		\begin{tikzpicture}
			\begin{axis}[width=0.47\textwidth, height = 3.1cm, legend image code/.code={\draw[#1] (0cm,0cm) circle (0.07cm);}, xmin = 0, xmax = 12]
				
				\addplot+[col_bribery, name path = bribery, no marks] plot [] table {DATA/0.0035/best/IKE2022/DHondt-Bribery-Max-Diff.dat};
			\end{axis}
		\end{tikzpicture}
		\begin{tikzpicture}
			\begin{axis}[width=0.47\textwidth, height = 3.1cm, legend image code/.code={\draw[#1] (0cm,0cm) circle (0.07cm);}, xmin = 0, xmax = 12]
				
				\addplot+[col_bribery, name path = bribery, no marks] plot [] table {DATA/0.0035/best/TWK2023/DHondt-Bribery-Max-Diff.dat};
			\end{axis}
		\end{tikzpicture}
		
		\begin{tikzpicture}
			\begin{axis}[width=0.47\textwidth, height = 3.1cm, legend image code/.code={\draw[#1, fill=#1] (0cm,0cm) circle (0.07cm);}, xmin = 0, xmax = 12, xlabel = {Austria}]
				
				\addplot+[col_bribery, name path = bribery, no marks] plot [] table {DATA/0.0035/best/AUT2019/DHondt-Bribery-Min-Diff.dat};
			\end{axis}
		\end{tikzpicture}
		\begin{tikzpicture}
			\begin{axis}[width=0.47\textwidth, height = 3.1cm, legend image code/.code={\draw[#1] (0cm,0cm) circle (0.07cm);}, xmin = 0, xmax = 12, xlabel = {Israel}]
				
				\addplot+[col_bribery, name path = bribery, no marks] plot [] table {DATA/0.0035/best/IKE2022/DHondt-Bribery-Min-Diff.dat};
			\end{axis}
		\end{tikzpicture}
		\begin{tikzpicture}
			\begin{axis}[width=0.47\textwidth, height = 3.1cm, legend image code/.code={\draw[#1] (0cm,0cm) circle (0.07cm);}, xmin = 0, xmax = 12, xlabel = {The Netherlands}]
				
				\addplot+[col_bribery, name path = bribery, no marks] plot [] table {DATA/0.0035/best/TWK2023/DHondt-Bribery-Min-Diff.dat};
			\end{axis}
		\end{tikzpicture}
		\caption{Experiment with a budget of $\budget = 0.0035\cdot n$.}
		\label{larger-budget}
	\end{subfigure}
	\caption{The $x$-axis shows a variety of thresholds in percent of the number of voters~$n$.
		The $y$-axis shows the maximally achievable number of additional seats (prevented seats) for the strongest party by D'Hondt-\textsc{AB} in the top row and by D'Hondt-\textsc{DAB} in the bottom row.}
\end{figure*}

\begin{figure*}[t]
	\begin{subfigure}{0.5\textwidth}
		\scriptsize
		\begin{tikzpicture}
			\begin{axis}[width=0.47\textwidth, height = 3.1cm, legend image code/.code={\draw[#1, fill=#1] (0cm,0cm) circle (0.07cm);}, xmin = 0, xmax = 12]
				
				\addplot+[col_bribery, name path = bribery, no marks] plot [] table {DATA/0.0025/secondbest/AUT2019/DHondt-Bribery-Max-Diff.dat};
			\end{axis}
		\end{tikzpicture}
		\begin{tikzpicture}
			\begin{axis}[width=0.47\textwidth, height = 3.1cm, legend image code/.code={\draw[#1] (0cm,0cm) circle (0.07cm);}, xmin = 0, xmax = 12]
				
				\addplot+[col_bribery, name path = bribery, no marks] plot [] table {DATA/0.0025/secondbest/IKE2022/DHondt-Bribery-Max-Diff.dat};
			\end{axis}
		\end{tikzpicture}
		\begin{tikzpicture}
			\begin{axis}[width=0.47\textwidth, height = 3.1cm, legend image code/.code={\draw[#1] (0cm,0cm) circle (0.07cm);}, xmin = 0, xmax = 12]
				
				\addplot+[col_bribery, name path = bribery, no marks] plot [] table {DATA/0.0025/secondbest/TWK2023/DHondt-Bribery-Max-Diff.dat};
			\end{axis}
		\end{tikzpicture}
		
		\begin{tikzpicture}
			\begin{axis}[width=0.47\textwidth, height = 3.1cm, legend image code/.code={\draw[#1, fill=#1] (0cm,0cm) circle (0.07cm);}, xmin = 0, xmax = 12, xlabel = {Austria}]
				
				\addplot+[col_bribery, name path = bribery, no marks] plot [] table {DATA/0.0025/secondbest/AUT2019/DHondt-Bribery-Min-Diff.dat};
			\end{axis}
		\end{tikzpicture}
		\begin{tikzpicture}
			\begin{axis}[width=0.47\textwidth, height = 3.1cm, legend image code/.code={\draw[#1] (0cm,0cm) circle (0.07cm);}, xmin = 0, xmax = 12, xlabel = {Israel}]
				
				\addplot+[col_bribery, name path = bribery, no marks] plot [] table {DATA/0.0025/secondbest/IKE2022/DHondt-Bribery-Min-Diff.dat};
			\end{axis}
		\end{tikzpicture}
		\begin{tikzpicture}
			\begin{axis}[width=0.47\textwidth, height = 3.1cm, legend image code/.code={\draw[#1] (0cm,0cm) circle (0.07cm);}, xmin = 0, xmax = 12, xlabel = {The Netherlands}]
				
				\addplot+[col_bribery, name path = bribery, no marks] plot [] table {DATA/0.0025/secondbest/TWK2023/DHondt-Bribery-Min-Diff.dat};
			\end{axis}
		\end{tikzpicture}
		
		\caption{Experiment with the
                  second-strongest party.}
		\label{fig:second-best}
	\end{subfigure}
	\begin{subfigure}{0.5\textwidth}
		\scriptsize
		\begin{tikzpicture}
			\begin{axis}[width=0.47\textwidth, height = 3.1cm, legend image code/.code={\draw[#1, fill=#1] (0cm,0cm) circle (0.07cm);}, xmin = 0, xmax = 12]
				
				\addplot+[col_bribery, name path = bribery, no marks] plot [] table {DATA/0.0025/thirdbest/AUT2019/DHondt-Bribery-Max-Diff.dat};
			\end{axis}
		\end{tikzpicture}
		\begin{tikzpicture}
			\begin{axis}[width=0.47\textwidth, height = 3.1cm, legend image code/.code={\draw[#1] (0cm,0cm) circle (0.07cm);}, xmin = 0, xmax = 12]
				
				\addplot+[col_bribery, name path = bribery, no marks] plot [] table {DATA/0.0025/thirdbest/IKE2022/DHondt-Bribery-Max-Diff.dat};
			\end{axis}
		\end{tikzpicture}
		\begin{tikzpicture}
			\begin{axis}[width=0.47\textwidth, height = 3.1cm, legend image code/.code={\draw[#1] (0cm,0cm) circle (0.07cm);}, xmin = 0, xmax = 12]
				
				\addplot+[col_bribery, name path = bribery, no marks] plot [] table {DATA/0.0025/thirdbest/TWK2023/DHondt-Bribery-Max-Diff.dat};
			\end{axis}
		\end{tikzpicture}
		
		\begin{tikzpicture}
			\begin{axis}[width=0.47\textwidth, height = 3.1cm, legend image code/.code={\draw[#1, fill=#1] (0cm,0cm) circle (0.07cm);}, xmin = 0, xmax = 12, xlabel = {Austria}]
				
				\addplot+[col_bribery, name path = bribery, no marks] plot [] table {DATA/0.0025/thirdbest/AUT2019/DHondt-Bribery-Min-Diff.dat};
			\end{axis}
		\end{tikzpicture}
		\begin{tikzpicture}
			\begin{axis}[width=0.47\textwidth, height = 3.1cm, legend image code/.code={\draw[#1] (0cm,0cm) circle (0.07cm);}, xmin = 0, xmax = 12, xlabel = {Israel}]
				
				\addplot+[col_bribery, name path = bribery, no marks] plot [] table {DATA/0.0025/thirdbest/IKE2022/DHondt-Bribery-Min-Diff.dat};
			\end{axis}
		\end{tikzpicture}
		\begin{tikzpicture}
			\begin{axis}[width=0.47\textwidth, height = 3.1cm, legend image code/.code={\draw[#1] (0cm,0cm) circle (0.07cm);}, xmin = 0, xmax = 12, xlabel = {The Netherlands}]
				
				\addplot+[col_bribery, name path = bribery, no marks] plot [] table {DATA/0.0025/thirdbest/TWK2023/DHondt-Bribery-Min-Diff.dat};
			\end{axis}
		\end{tikzpicture}
		
		\caption{Experiment with the third-strongest party.}
		\label{fig:third-best}
	\end{subfigure}
	\caption{The $x$-axis shows a variety of thresholds in percent of the number of voters~$n$.
		The $y$-axis shows the maximally achievable number of additional seats (prevented seats) by D'Hondt-\textsc{AB} in the top row and D'Hondt-\textsc{DAB} in the bottom row, each with a given budget of $\budget = 0.0025\cdot n$.}
\end{figure*}
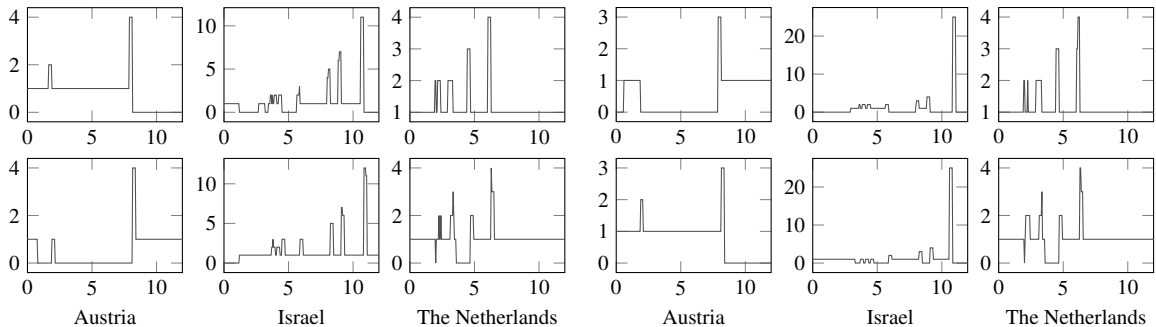

\subsubsection{Influence of the Number of Districts} 

In the third experiment, we study the influence of the number of districts on the needed budget for changing the election result.
Here, we consider the three multi-district elections, namely the Polish, Portuguese, and Argentinian datasets (again, see Table~\ref{tab:datasets_overview} for an overview of number of districts, seats, thresholds, and number of votes). 
We compute an optimal bribery for one party\footnote{For the Polish and Portuguese elections, we choose the party with the largest number of seats and for the Argentinian election, we choose the second strongest party as the strongest party is already rather close to a majority.} per dataset allowing this party to obtain a majority of seats ($50\% + 1$) in the constructive case (note that in none of the elections, there is a party that originally wins a majority of seats) and to lose half of its original seats in the destructive case.
To study the effect of the number of districts, we start with the original partition of the country and then decrease their number by merging districts.
We do this sequentially, always merging two districts picked uniformly at random, until the chosen number of districts remains. 
For the Polish election, we merge $10$, $20$, and $30$ districts, for the Portuguese and Argentinian elections (as they have a similar number of districts which is about a half of the number of districts in the Polish election), we merge $5$, $10$, and $15$ districts.
The result for a given number of districts is computed as follows:
\begin{enumerate}
	\item We take the average of the number of bribed voters over three
	trials, each for a different districting (computed as described above);
	\item we compute the average number of vote transfers per gained (respectively, lost) seat,
	and 
	\item we divide it by the number of vote transfers per gained (respectively, lost) seat in the original elections. 
\end{enumerate}
\noindent
The results are summarized in Table~\ref{tab:md-price}. 
In all cases, a smaller number of districts increases the minimum budget needed to win (respectively, lose) a seat for a party on average, i.e., the more districts there are in an election, the easier it is (in terms of cost) to carry out a successful bribery. 
For example, in the Polish elections (see Table~\ref{tab:md-price-new-pol}), if the number of districts is reduced by about half (to 21 districts), the budget needed to win a seat is more than twice as high ($2,371$ times) as with the original number of districts. In the destructive case, even more than four times as many votes ($4,213$ times) must be changed to make the party lose a seat. 
Note that even if the effect of merging districts is smaller, i.e., if the increase is slower (looking at the percentage of districts remaining), such as in the Portuguese elections (see Table~\ref{tab:md-price-new-por}), the difference between the minimum budgets is still at least a few thousand voters that one would need to bribe.

\begin{table*}[h]
	\caption{Average minimum required budgets for optimal bribery in elections with different numbers of districts to get an additional seat or lose a seat for an elected party; results are given by the ratio $\nicefrac{\text{MpS}}{\text{MpS-o}}$, where MpS is the average over the minimum budgets to get the seat (or lose it) for an elected party on average, and MpS-o is the MpS value for the original election, i.e., without merged districts.}
\begin{subtable}{0.5\textwidth}
\caption{Poland}
 \centering%
 \begin{tabular}{lrrrr}
  \toprule
  \# districts & $41$ & $31$ & $21$ & $11$ \\\midrule
  constructive & $1$ &
  $1.622$ & $2.371$ & $3.607$\\ 
  destructive & $1$ &
  $1.968$ & $4.213$ & $5.865$\\ \bottomrule
 \end{tabular}%
 \label{tab:md-price-new-pol}
\end{subtable} 
\begin{subtable}{0.5\textwidth}
\caption{Portugal}
 \centering%
 \begin{tabular}{lrrrr}
  \toprule
  \# districts & $22$ & $17$ & $12$ & $7$ \\\midrule
  constructive & $1$ &
  $1.168$ & $1.487$ & $1.843$\\ 
  destructive & $1$ &
  $1.293$ & $1.621$ & $2.260$\\ \bottomrule
 \end{tabular}%
 \label{tab:md-price-new-por}
\end{subtable}

\ \\ 

\begin{subtable}{\textwidth}
\caption{Argentina}
 \centering%
 \begin{tabular}{lrrrr}
  \toprule
  \# districts & $24$ & $19$ & $14$ & $9$ \\\midrule
  constructive & $1$ &
  $1.276$ & $1.961$ & $2.789$\\
  destructive & $1$ &
  $1.607$ & $2.805$ & $3.337$\\ \bottomrule
 \end{tabular}%
 \label{tab:md-price-new-arg}
\end{subtable} 
\label{tab:md-price}
\end{table*}

\section{Second-Chance Mode}\label{sec:second-chance}

	From the previous sections we know that it is quite problematic if optimal campaigns are easy to compute,
	because it can make it very simple for a campaign manager to exert an enormous influence on the election outcome.
	Therefore, it would be of great advantage if there were a modification to the usual apportionment setting that makes the computation of optimal campaigns intractable.
	As mentioned in the introduction, another general problem of apportionment elections with a threshold is that voters for parties below the threshold are completely ignored, which results in a less representative parliament.

	We now introduce the \emph{second-chance mode} of voting in apportionment elections which will help to resolve both of these problems.
	Unlike in the top-choice mode, in the second-chance mode voters for parties below the threshold get a second chance to use their vote.
	That is, we first determine the parties
        $\partiesabove$ that have at least $\threshold$ top choices, i.e., the parties that make it above the threshold.
	Each voter now counts as a supporter for their most preferred party in~$\partiesabove$,
        i.e., after removing the parties not in $\partiesabove$ from the ballots, votes for removed parties are transferred to the most preferred party from~$\partiesabove$.
	The second-chance voting process is somewhat reminiscent of
        iterative voting rules such as 
        \emph{plurality with runoff} or \emph{single transferable vote (STV)}.
	However, it also differs from them, as they are voting rules and not used for computing support allocations.
	
	\begin{example}[support allocation in the second-chance mode]\label{ex:secondchance} 
		Recall the election \[E=(\{P_1, \dots, P_5\}, \{v_1, \dots, v_{2052}\})\] from Example~\ref{ex:supportallocation} with the votes
		\begin{align*}
			v_1, \dots, v_{604} :&\quad P_1 \succ P_2 \succ P_3 \succ P_5 \succ P_4\\
			v_{605}, \dots, v_{819} :&\quad P_2 \succ P_1 \succ P_5 \succ P_4 \succ P_3\\
			v_{820}, \dots, v_{1174} :&\quad P_3 \succ P_2 \succ P_1 \succ P_4 \succ P_5\\
			v_{1175}, \dots, v_{1474} :&\quad P_1 \succ P_5 \succ P_3 \succ P_2 \succ P_4\\
			v_{1475}, \dots, v_{1652} :&\quad P_4 \succ P_2 \succ P_3 \succ P_1 \succ P_5\\
			v_{1653}, \dots, v_{1800} :&\quad P_2 \succ P_4 \succ P_3 \succ P_1 \succ P_5\\
			v_{1801}, \dots, v_{2000} :&\quad P_1 \succ P_3 \succ P_4 \succ P_5 \succ P_2\\
			v_{2001}, \dots, v_{2052} :&\quad P_5 \succ P_2 \succ P_3 \succ P_4 \succ P_1
		\end{align*}
		and the apportionment problem $I = (\supportalloc, 100, 6)$. %
		As we know from Example~\ref{ex:supportallocation}, $P_1$, $P_2$, $P_3$, and $P_4$
                each exceed the threshold, while party $P_5$ does not,
                so $\partiesabove = \parties_{100} = \{P_1, P_2, P_3, P_4\}$. 
		In the top-choice mode, the respective support allocation was as follows:
		\begin{align*}
			\supportalloc(P_1) &= 1104,\\
			\supportalloc(P_2) &= 363, \\
			\supportalloc(P_3) &= 355,\\
			\supportalloc(P_4) &= 178,\\
			\supportalloc(P_5) &= 0.
		\end{align*}
		However, in the second-chance mode, instead of ignoring all votes where $P_5$ is the top choice (here, votes $v_{2001}$ through~$v_{2052}$), we
                remove $P_5$ in these votes and then check which of the parties
                above the threshold---$P_1$, $P_2$, $P_3$, or~$P_4$---is most preferred in them
                after $P_5$'s removal.
               In this example, this is $P_2$ in all votes $v_{2001}, \ldots, v_{2052}$, so we add those $52$ votes to the support of $P_2$ and obtain the final support allocation 
		\begin{align*}
			\supportalloc(P_1) &= 1104,\\
			\supportalloc(P_2) &= 415, \\
			\supportalloc(P_3) &= 355,\\
			\supportalloc(P_4) &= 178,\\
			\supportalloc(P_5) &= 0.
		\end{align*}
	\end{example}

	Note that similar voting systems are already being used in Australia, e.g., for the \textit{House of Representative} and the \textit{Senate}.\footnote{\scriptsize{\url{https://www.aec.gov.au/learn/preferential-voting.htm}}} 
	In those elections, voters rank the candidates or parties from most to least preferred and votes for excluded choices are transferred according to the given ranking until the vote counts.

	\subsection{Single-District Case}

	In Section~\ref{subsec:top-choice-single-district}, we showed that in the classical apportionment setting, 
	$\appmethod$-\textsc{AB} and $\appmethod$-\textsc{DAB} can be solved efficiently in the single-district case for FPTP, each divisor sequence method, and LRM.
	In contrast, we will show that these problems are $\np$-hard in the second-chance mode.
        To this end, we provide reductions from the $\np$-complete \textsc{Hitting-Set} problem~\cite{kar:b:reducibilities}, which has been defined in
        Section~\ref{sec:preliminaries:useful-problems}.

	Instead of just focusing on specific apportionment methods, in the following we generalize our results to a whole class of apportionment methods.	

	\begin{definition}[majority consistency]
		We call an apportionment method \emph{majority-consistent} if no party in $\parties$ with less support than $A$ receives more seats than~$A$, where $A \in \parties$ is a party with the highest support.
	\end{definition}
        
	Undoubtedly, this is a criterion every reasonable apportionment method should satisfy.
	Note that FPTP, all divisor sequence methods, and LRM are obviously majority-consistent.
	
	We now show that the second-chance mode of apportionment voting makes computing an optimal strategic campaign computationally intractable, and thus can prevent attempts of running such campaigns.
	
	\begin{theorem}\label{thm:bribery-2nd-chance}
		For each majority-consistent apportionment method~$\appmethod$, $\appmethod$-\textsc{AB} and $\appmethod$-\textsc{DAB} are $\np$-hard in the second-chance mode.
		If $\appmethod$ is polynomial-time computable, they are $\np$-complete.
	\end{theorem}
        
	\begin{proof}
		Membership of both problems in $\np$ is obvious whenever $\appmethod$ is polynomial-time computable. 
		We show $\np$-hardness of $\appmethod$-\textsc{AB} by a reduction from \textsc{Hitting-Set}.
		Let $(U, S, \budget) = (\{u_1, \dots, u_p\}, \{S_1, \dots, S_q\}, \budget)$ be an instance of \textsc{Hitting-Set} with $q \geq 4$. 
		In polynomial time, we construct an instance of $\appmethod$-\textsc{AB} with parties $\parties = \{P^*,P\} \cup U$, a threshold $\threshold=2q+1$, $\desiredseats = 1$ desired seat, $\seats=1$ available seat, and the votes
		\begin{eqnarray}
			\votes = (  4q+2 \text{ votes} && P^* \succ \cdots, \nonumber\\ %
			4q+\budget+2\text{ votes}  && P \succ \cdots, \label{eq:cb-four}\\
			\text{for each } j \in [q], 2 \text{ votes} && S_j \succ P \succ \cdots,\label{eq:cb-one}\\
			\text{for each } i \in [p], q-\gamma_i \text{ votes} && u_i \succ P^* \succ \cdots,\label{eq:cb-two}\\
			\text{for each } i \in [p], q-\gamma_i \text{ votes} && u_i \succ P \succ \cdots \label{eq:cb-three}),
		\end{eqnarray}
		where $S_j \succ P$ denotes that each element in $S_j$ is preferred to~$P$, but we do not care about the exact order of the elements in~$S_j$.
		Further, $2\gamma_i$ is the number of votes from group~(\ref{eq:cb-one}) in which $u_i$ is in the first position. 
		Thus it is guaranteed that
                \begin{itemize}
                  \item each $u_i \in U$ receives exactly $2\gamma_i + (q-\gamma_i) + (q-\gamma_i) =2q < \threshold$ top choices 
        (which is below the threshold, preventing them from receiving a seat), while 
	          \item $P^*$ has $4q+2 \ge \threshold$
		top choices and the support of $(p+4)q+2-\sum_{i=1}^p \gamma_i$ voters,  
		and
                  \item $P$ has $4q+\budget+2 \ge \threshold$ top choices 
		with the support of $(p+6)q+\budget+2-\sum_{i=1}^p \gamma_i$ voters.
                \end{itemize}
		Note that the voters in groups~(\ref{eq:cb-one}) and~(\ref{eq:cb-three}) use their second chance to vote for~$P$, and those in group~(\ref{eq:cb-two}) use it to vote for~$P^*$.
		It follows that $P$ currently receives exactly $2q+\budget$ more votes than~$P^*$ and thus wins the seat.
                
		We now show that we can make the distinguished party
                $P^*$ win the seat by bribing at most $\budget$ voters if and only if there is a hitting set of size at most~$\budget$.
		
		\textbf{($\boldsymbol{\Leftarrow}$)}\;\;
		Suppose there exists a hitting set $U' \subseteq U$ of size
                $\budget$ (if $|U'| < \budget$, it can be padded to size exactly~$\budget$ by adding arbitrary elements from~$U$).
		For each $u_i\in U'$, we bribe one voter from group~(\ref{eq:cb-four}) to put $u_i$ at their first position.
		These $u_i$ now each receive the $2q+1 = \threshold$ top choices required by the threshold, i.e., they participate in the further apportionment process.		
		Each $u_i$ can receive a support of at most $4q+1$.
		Since the support of $P^*$ is not affected by any bribes, no $u_i$ can win the seat against $P^*$
                under any majority-consistent apportionment method~$\appmethod$.
		Groups~(\ref{eq:cb-two}) and~(\ref{eq:cb-three}) do not change the support difference between $P^*$ and $P$ and thus can be ignored.
		However, since
                all members of $U'$ now reach the threshold and 
                $U'$ is a hitting set, all $2q$ voters in group~(\ref{eq:cb-one}) now vote for a party in $U'$ instead of~$P$, reducing the difference between $P^*$ and $P$ by~$2q$.
		Further, we have bribed $\budget$ voters from group~(\ref{eq:cb-four}) to not vote for~$P$, which reduces the difference between $P^*$ and $P$ by another $\budget$ votes.
		Therefore, $P^*$ and $P$ now have the same support
        of $4q+2$ voters, 
        and since we assume tie-breaking to prefer~$P^*$, party $P^*$ wins the desired seat.
		
		\textbf{($\boldsymbol{\Rightarrow}$)}\;\;
		Suppose the smallest hitting set has size $\budget' > \budget$.
		That is, with only $\budget$ elements of $U$ we can hit at most $q-1$ sets from~$S$. 
		It follows that by bribing $\budget$ voters from group~(\ref{eq:cb-four}) to vote for some $u_i \in U$ instead of $P$, we can only prevent up to $2(q-1)$ voters from group~(\ref{eq:cb-one}) to use their second chance to vote for~$P$.
		Thus we reduce the difference between $P^*$ and $P$ by at most $2(q-1)+\budget$, which is not enough to make $P^*$ win the seat.
		Now consider that we do not use the complete budget $\budget$ on this strategy, i.e., to bribe voters of group~(\ref{eq:cb-four}), but only $\budget'' < \budget$.
		Note that by bringing only $\budget''$ parties from $U$ above the threshold, we can only hit up to $2(q-1-(\budget-\budget''))$ sets from~$S$.
		So the difference between $P^*$ and $P$ is reduced by at most $2(q-1-(\budget-\budget'')) + \budget''$ using this strategy.
		However, we now have a budget of $\budget-\budget''$ left to bribe voters, e.g., from group~(\ref{eq:cb-one}),
                to vote for $P^*$ as their top choice, without bringing any additional parties $u_i$ above the threshold.
		It is easy to see that the support difference between $P^*$ and $P$ will be only reduced by at most $2(\budget-\budget'')$ with this strategy because, in the best case, $P^*$ gains one supporter and $P$ loses one with a single bribery action.
		Thus, with this mixed strategy, we cannot reduce the difference between $P^*$ and $P$ by more than
                \[
                2(q-1-(\budget-\budget'')) + \budget'' + 2(\budget-\budget'') = 2(q-1)+ \budget''.
                \]
		For each $\budget'' \leq \budget$, we have $2(q-1)+ \budget'' < 2q+\budget$. 
		Therefore, if there is no hitting set of size at most $\budget$, we cannot make the  distinguished party $P^*$ win against~$P$.

		The proof for the destructive variant works by swapping the roles of $P^*$ and~$P$ and setting $\desiredseats = 0$.
	\end{proof}

	In the case that $\seats=1$ (i.e., we only have one seat to allocate), asking whether our distinguished party $P^*$ can receive at least $\desiredseats=1$ seat (respectively, $\desiredseats = 0$ in the destructive case) is the same as asking whether $P^*$ can become the strongest (respectively, not the strongest) party.
        Therefore, the same reduction can be used for the winner problems.

	\begin{theorem}\label{thm:winner-bribery-2nd-chance}
		For each majority-consistent apportionment method~$\appmethod$, $\appmethod$-\textsc{AWB} and $\appmethod$-\textsc{DAWB} are $\np$-hard in the second-chance mode.
		If $\appmethod$ is polynomial-time computable, they are $\np$-complete.
	\end{theorem}

	\subsection{Multi-District Case}
	
	Since in the second-chance mode all single-district problems %
	are already $\np$-hard for all majority-consistent apportionment methods (see Theorems~\ref{thm:bribery-2nd-chance} and~\ref{thm:winner-bribery-2nd-chance}), a simple reduction from
        the single-district problems to the corresponding multi-district problems suffices to prove $\np$-hardness for the multi-district case. 
        Specifically, for our reductions in the constructive case (i.e., from $\appmethod$-\textsc{AB} to $\appmethod$-\textsc{MAB}) and in the destructive case (i.e., from $\appmethod$-\textsc{DAB} to $\appmethod$-\textsc{DMAB}), we use the same input as in the single-district problem and set the number of districts to $q = 1$.
        For the winner problems, we reduce from $\appmethod$-\textsc{AWB} to $\appmethod$-\textsc{MAWB} in the constructive case and from $\appmethod$-\textsc{DAWB} to $\appmethod$-\textsc{DMAWB} in the destructive case, using the same simple construction as above, except that here we are not given a desired number of seats.

     \begin{proposition}\label{prop:multibribery_2ndchance}
		For each majority-consistent apportionment method~$\appmethod$, $\appmethod$-\textsc{MAB}, $\appmethod$-\textsc{DMAB}, $\appmethod$-\textsc{MAWB}, and $\appmethod$-\textsc{DMAWB} are $\np$-hard in the second-chance mode.
		If $\appmethod$ is polynomial-time computable, they are $\np$-complete.
	\end{proposition}

\section{Conclusion and Future Work}

\begin{table}[t]
	\caption{
	  Overview of our results on the complexity of
          bribery in apportionment elections.
                The \np-hardness results in the rightmost column are
                \np-completeness results for majority-consistent
                apportionment methods computable in polynomial time.
                All \wonehardness{} results are with respect to the number of districts.
        }
	\label{tab:results_overview}
	\begin{tabularx}{\textwidth}{ll|CCC|CCC}
          \toprule
		&  & \multicolumn{3}{c}{Top-choice mode} & \multicolumn{3}{|c}{Second-chance mode} \\
		\midrule
		&  & FPTP & divisor sequence & LRM & \multicolumn{3}{c}{majority-consistent} \\
		\midrule
		\multirow{2}{*}{AB} & constr. & $\p$ (Prop~\ref{prop:fptp}) & $\p$ (Thm~\ref{thm:bribery-top-choice}) & $\p$ (Thm~\ref{thrm:bribery-top-choice-lrm}) & \multicolumn{3}{c}{\np-hard (Thm~\ref{thm:bribery-2nd-chance})} \\
		&  destr. & $\p$ (Prop~\ref{prop:fptpdestructive}) & $\p$ (Thm~\ref{thm:destr-bribery-top-choice}) & $\p$ (Thm~\ref{thm:destr-bribery-top-choice-lrm}) & \multicolumn{3}{c}{\np-hard (Thm~\ref{thm:bribery-2nd-chance})}\\
		\midrule
		\multirow{2}{*}{AWB} & constr. & $\p$ (Prop~\ref{prop:fptp}) & $\p$ (Thm~\ref{thm:winner-single-bribery}) & $\p$ (Thm~\ref{thm:winner-single-bribery}) & \multicolumn{3}{c}{\np-hard (Thm~\ref{thm:winner-bribery-2nd-chance})} \\
		& destr. & $\p$ (Prop~\ref{prop:fptpdestructive}) & $\p$ (Prop~\ref{prop:destr-winner-brib})& $\p$ (Prop~\ref{prop:destr-winner-brib-lrm}) & \multicolumn{3}{c}{\np-hard (Thm~\ref{thm:winner-bribery-2nd-chance})}\\
		\midrule
		\multirow{2}{*}{MAB} & constr. & $\p$ (Prop~\ref{prop:divisor_lrm_multi_bribery}) & $\p$ (Prop~\ref{prop:divisor_lrm_multi_bribery}) & $\p$ (Prop~\ref{prop:divisor_lrm_multi_bribery}) & \multicolumn{3}{c}{\np-hard (Prop~\ref{prop:multibribery_2ndchance})} \\
		&  destr. & $\p$ (Prop~\ref{prop:destr_divisor_lrm_multi_bribery}) & $\p$ (Prop~\ref{prop:destr_divisor_lrm_multi_bribery}) & $\p$ (Prop~\ref{prop:destr_divisor_lrm_multi_bribery}) &\multicolumn{3}{c}{\np-hard (Prop~\ref{prop:multibribery_2ndchance})}\\
		\midrule
		\multirow{4}{*}{MAWB} & constr. & \np-compl (Thm~\ref{thm:fptp-np-hard-wone-hard-winner}) & \np-compl (Thm~\ref{thm:pAG-np-h}) & \np-compl (Thm~\ref{thm:pAG-np-h}) &\multicolumn{3}{c}{\np-hard (Prop~\ref{prop:multibribery_2ndchance})} \\
		&  & \wone{}-hard (Thm~\ref{thm:fptp-np-hard-wone-hard-winner}) & \wone{}-hard (Thm~\ref{thm:wone-hard-districts}) & \wone{}-hard (Thm~\ref{thm:wone-hard-districts}) & \multicolumn{3}{c}{\ } \\
		& destr. &  $\p$ (Thm~\ref{thm:fptp-dmawb-in-p}) & \np-compl (Thm~\ref{thm:dest-multi-winner}) & \np-compl (Thm~\ref{thm:dest-multi-winner}) & \multicolumn{3}{c}{\np-hard (Prop~\ref{prop:multibribery_2ndchance})} \\
		&  & & \wone{}-hard (Thm~\ref{thm:wone-hard-multi-destructive}) & \wone{}-hard (Thm~\ref{thm:wone-hard-multi-destructive}) & \multicolumn{3}{c}{\ } \\
          \bottomrule
	\end{tabularx}
\end{table}

We have studied strategic campaign management (modeled as bribery problems) for apportionment elections in the single- and multi-district setting and introduced the second-chance mode of voting, where voters for parties below the threshold get a second chance to use their vote.
An overview of our theoretical complexity results for constructive bribery problems is given in Table~\ref{tab:results_overview}.
We also ran extensive experimental analyses studying the effectiveness of optimal campaigns, in particular as opposed to using heuristic bribing strategies and with respect to the influence of the threshold and the influence of the number of districts. 
We found that in most cases, moving votes optimally is (significantly) more effective than simple bribing strategies, our optimal campaigns can exploit and significantly benefit from a threshold, and the fewer districts there are, the more vote moves are required to obtain (respectively, take away) a certain number of seats.

Our algorithms require precise information about vote counts,
i.e., information that is difficult to acquire \emph{prior to} an election.
We consider it an interesting task for future work to study
bribery and strategic campaigns in apportionment elections based on imperfect information (e.g., poll data).
The comparison of simple bribing schemes in the experimental setting of Section~\ref{subsec:experiments-heuristics} can be seen as a first step in this direction.

Further, we focused on popular apportionment schemes, but there are other such methods as well, such as the quota method of Balinski and Young~\citep{bal-you:j:quota-method-of-apportionment,bal-you:j:apportionment-schemes-and-quota-method} or Frege's apportionment method, described in a long-lost writing of Frege's that was rediscovered in Jena, German, in 1998 (see, e.g., the work of Harrenstein et al.~\citep{har-lac-lac:j:mathematical-analysis-of-election-system-by-gottlob-frege}).

Moreover, we propose to study other types of strategic campaigns (e.g., cloning of parties; in a spirit similar to cloning in single-winner and multiwinner elections~\cite{tid:j:independence-of-clones,elk-fal-sli:j:cloning,nev-rot:c:complexity-of-cloning-candidates-in-multiwinner-elections}). 
We have already studied electoral control problems (see the book chapters by Faliszewski and Rothe~\cite{fal-rot:b:handbook-comsoc-control-and-bribery} and Baumeister and Rothe~\cite{bau-rot:b-2nd-edition:economics-and-computation-preference-aggregation-by-voting}), in particular constructive and destructive control by adding or deleting parties or voters.
It turns out that both the top-choice and the second-chance mode are resistant to all four party control problems.
For the proofs, it suffices to adapt the proofs of Bartholdi et al.~\cite{bar-tov-tri:j:control} and Hemaspaandra et al.~\cite{hem-hem-rot:j:destructive-control} showing that plurality voting is resistant to the corresponding electoral control problems.
For voter control in the top-choice mode, Algorithm~\ref{alg:bribery-constructive} can be adapted to show that all four cases of voter control are in $\p$ for divisor sequence methods.
However, this only works when the threshold is fixed, i.e., not given as a percentage of the number $n$ of voters.
Regarding the second-chance mode, we have been able to show that all majority-consistent apportionment methods are resistant to voter control if the threshold is fixed.
 
Another direction for future research is to study the complexity of these problems in restricted domains such as (nearly) single-peaked preferences~\cite{fal-hem-hem:j:single-peaked-nearly,erd-lac-pfa:c:nearly-sp,fal-hem-hem-rot:j:single-peaked-preferences}.
Also, studying the effectiveness of strategic campaigns in the second-chance mode using
integer linear programming or approximation algorithms is an interesting direction for future work.

To make our strategic campaigns even more realistic, we propose to study more sophisticated cost functions such as \emph{distance bribery}~\cite{bau-hog-rey:c:generalized-distance-bribery} where the cost of bribing a voter depends on how much we change the vote.
We conjecture the problem to be harder under the assumption of distance bribery, founded on the observation that Lemma~\ref{lem:moving-voters} no longer holds.
That is, there are cases where it is more effective to move votes within $\parties_{- P^*}$ than to move them to~$P^*$. 
To illustrate this, suppose we have two seats available for three parties, a support allocation given by $\supportalloc(P^*) = 7$, $\supportalloc(P_A) = 4$, and $\supportalloc(P_B) = 2$, and a threshold of $\threshold < 2$.
According to D'Hondt, $P^*$ and $P_A$ each receive one seat.
Suppose the budget is $\budget = 1$ but the cost for changing a vote from $P_A$ to $P^*$ is~$2$, and the cost for changing a vote from $P_B$ to $P^*$ is even higher.
We thus cannot move a single vote to~$P^*$, i.e., we cannot gain any seats for $P^*$ by this strategy.
However, if the cost for moving a vote from $P_A$ to $P_B$ is~$1$, we gain one seat for $P^*$ by moving a vote from $P_A$ to~$P_B$.

While NP-hardness is desirable in the context of strategic campaigns, in the context of, e.g., margin of victory or robustness, the interpretations can be flipped, which can also be studied as future work.
Finally, we suggest studying the extent to which voter satisfaction with the parliament increases when the second-chance mode is used.

\section* {Acknowledgements}
We thank Niclas Boehmer and Martin Bullinger for their insightful comments and ideas that we discussed with them during Seminar~19381 ``Application-Oriented Computational Social Choice'' at Schloss Dagstuhl in September 2019.
We also thank the anonymous IJCAI'20 and SOFSEM'24 reviewers for helpful comments.
This work was supported in part by Austrian Science Fund (FWF) under grant P31890, by Deutsche Forschungsgemeinschaft (DFG) under grants BR~5207/1 and NI~369/15 (project number 284041127) and RO~1202/21-1 and RO~1202/21-2 (project number 438204498), from the funds assigned to AGH by Polish Ministry of Science and Higher Education, and by a Friedrich Wilhelm Bessel Research Award given by Alexander von Humboldt Foundation to Piotr Faliszewski.

\bibliographystyle{abbrvnat}
\bibliography{master}

\begin{thebibliography}{42}
\providecommand{\natexlab}[1]{#1}
\providecommand{\url}[1]{\texttt{#1}}
\expandafter\ifx\csname urlstyle\endcsname\relax
  \providecommand{\doi}[1]{doi: #1}\else
  \providecommand{\doi}{doi: \begingroup \urlstyle{rm}\Url}\fi

\bibitem[Balinski and Young(1975)]{bal-you:j:quota-method-of-apportionment}
M.~Balinski and P.~Young.
\newblock The quota method of apportionment.
\newblock \emph{The American Mathematical Monthly}, 82\penalty0 (7):\penalty0
  701--730, 1975.

\bibitem[Balinski and
  Young(1977)]{bal-you:j:apportionment-schemes-and-quota-method}
M.~Balinski and P.~Young.
\newblock Apportionment schemes and the quota method.
\newblock \emph{The American Mathematical Monthly}, 84\penalty0 (6):\penalty0
  450--455, 1977.

\bibitem[Balinski and Young(1982)]{bal-you:b:polsci:fair-representation}
M.~Balinski and P.~Young.
\newblock \emph{Fair Representation: {Meeting} the Ideal of One Man, One Vote}.
\newblock Yale University Press, New Haven, CT, USA, 1982.

\bibitem[{Bartholdi III} et~al.(1989){Bartholdi III}, Tovey, and
  Trick]{bar-tov-tri:j:manipulating}
J.~{Bartholdi III}, C.~Tovey, and M.~Trick.
\newblock The computational difficulty of manipulating an election.
\newblock \emph{Social Choice and Welfare}, 6\penalty0 (3):\penalty0 227--241,
  1989.

\bibitem[{Bartholdi III} et~al.(1992){Bartholdi III}, Tovey, and
  Trick]{bar-tov-tri:j:control}
J.~{Bartholdi III}, C.~Tovey, and M.~Trick.
\newblock How hard is it to control an election?
\newblock \emph{Mathematical and Computer Modelling}, 16\penalty0
  (8/9):\penalty0 27--40, 1992.

\bibitem[Baumeister and
  Rothe(2024)]{bau-rot:b-2nd-edition:economics-and-computation-preference-aggregation-by-voting}
D.~Baumeister and J.~Rothe.
\newblock Preference aggregation by voting.
\newblock In J.~Rothe, editor, \emph{Economics and Computation. An Introduction
  to Algorithmic Game Theory, Computational Social Choice, and Fair Division},
  Classroom Companion: Economics, chapter~4, pages 233--367. Springer, 2nd
  edition, 2024.

\bibitem[Baumeister et~al.(2012)Baumeister, Faliszewski, Lang, and
  Rothe]{bau-fal-lan-rot:c:campaigns-for-lazy-voters}
D.~Baumeister, P.~Faliszewski, J.~Lang, and J.~Rothe.
\newblock Campaigns for lazy voters: {Truncated} ballots.
\newblock In \emph{Proceedings of the 11th International Conference on
  Autonomous Agents and Multiagent Systems}, pages 577--584. IFAAMAS, June
  2012.

\bibitem[Baumeister et~al.(2019)Baumeister, Hogrebe, and
  Rey]{bau-hog-rey:c:generalized-distance-bribery}
D.~Baumeister, T.~Hogrebe, and L.~Rey.
\newblock Generalized distance bribery.
\newblock In \emph{Proceedings of the 33rd AAAI Conference on Artificial
  Intelligence}, pages 1764--1771. AAAI Press, January/February 2019.

\bibitem[Bredereck et~al.(2020)Bredereck, Faliszewski, Furdyna, Kaczmarczyk,
  and
  Lackner]{bre-fal-fur-kac-lac:c:strategic-campaign-management-in-apportionment-elections}
R.~Bredereck, P.~Faliszewski, M.~Furdyna, A.~Kaczmarczyk, and M.~Lackner.
\newblock Strategic campaign management in apportionment elections.
\newblock In \emph{Proceedings of the 29th International Joint Conference on
  Artificial Intelligence}, pages 103--109. ijcai.org, Jan. 2020.

\bibitem[Bredereck et~al.(2021)Bredereck, Faliszewski, Kaczmarczyk,
  Niedermeier, Skowron, and Talmon]{bre-fal-kac-nie-sko-tal:j:mw-robustness}
R.~Bredereck, P.~Faliszewski, A.~Kaczmarczyk, R.~Niedermeier, P.~Skowron, and
  N.~Talmon.
\newblock Robustness among multiwinner voting rules.
\newblock \emph{Artificial Intelligence}, 290:\penalty0 103403, 2021.

\bibitem[Brill et~al.(2018)Brill, Laslier, and
  Skowron]{bri-las-sko:j:approval-as-apportionment}
M.~Brill, J.~Laslier, and P.~Skowron.
\newblock Multiwinner approval rules as apportionment methods.
\newblock \emph{Journal of Theoretical Politics}, 30:\penalty0 358--382, 2018.

\bibitem[Conitzer and Walsh(2016)]{con-wal:b:handbook-comsoc-manipulation}
V.~Conitzer and T.~Walsh.
\newblock Barriers to manipulation in voting.
\newblock In F.~Brandt, V.~Conitzer, U.~Endriss, J.~Lang, and A.~Procaccia,
  editors, \emph{Handbook of Computational Social Choice}, chapter~6, pages
  127--145. Cambridge University Press, 2016.

\bibitem[Conitzer et~al.(2007)Conitzer, Sandholm, and
  Lang]{con-san-lan:j:when-hard-to-manipulate}
V.~Conitzer, T.~Sandholm, and J.~Lang.
\newblock When are elections with few candidates hard to manipulate?
\newblock \emph{Journal of the ACM}, 54\penalty0 (3):\penalty0 Article~14,
  2007.

\bibitem[Elkind and Faliszewski(2010)]{elkind2010shiftbribery}
E.~Elkind and P.~Faliszewski.
\newblock Approximation algorithms for campaign management.
\newblock In \emph{Proceedings of the 6th International Workshop on Internet \&
  Network Economics}, pages 473--482, 2010.

\bibitem[Elkind et~al.(2011)Elkind, Faliszewski, and
  Slinko]{elk-fal-sli:j:cloning}
E.~Elkind, P.~Faliszewski, and A.~Slinko.
\newblock Cloning in elections: {F}inding the possible winners.
\newblock \emph{Journal of Artificial Intelligence Research}, 42:\penalty0
  529--573, 2011.

\bibitem[Erd\'{e}lyi et~al.(2013)Erd\'{e}lyi, Lackner, and
  Pfandler]{erd-lac-pfa:c:nearly-sp}
G.~Erd\'{e}lyi, M.~Lackner, and A.~Pfandler.
\newblock Computational aspects of nearly single-peaked electorates.
\newblock In \emph{Proceedings of the 27th AAAI Conference on Artificial
  Intelligence}, pages 283--289. AAAI Press, July 2013.

\bibitem[Faliszewski and
  Rothe(2016)]{fal-rot:b:handbook-comsoc-control-and-bribery}
P.~Faliszewski and J.~Rothe.
\newblock Control and bribery in voting.
\newblock In F.~Brandt, V.~Conitzer, U.~Endriss, J.~Lang, and A.~Procaccia,
  editors, \emph{Handbook of Computational Social Choice}, chapter~7, pages
  146--168. Cambridge University Press, 2016.

\bibitem[Faliszewski et~al.(2009{\natexlab{a}})Faliszewski, Hemaspaandra, and
  Hemaspaandra]{fal-hem-hem:j:bribery}
P.~Faliszewski, E.~Hemaspaandra, and L.~Hemaspaandra.
\newblock How hard is bribery in elections?
\newblock \emph{Journal of Artificial Intelligence Research}, 35:\penalty0
  485--532, 2009{\natexlab{a}}.

\bibitem[Faliszewski et~al.(2009{\natexlab{b}})Faliszewski, Hemaspaandra,
  Hemaspaandra, and Rothe]{fal-hem-hem-rot:j:llull-copeland-full-techreport}
P.~Faliszewski, E.~Hemaspaandra, L.~Hemaspaandra, and J.~Rothe.
\newblock Llull and {Copeland} voting computationally resist bribery and
  constructive control.
\newblock \emph{Journal of Artificial Intelligence Research}, 35:\penalty0
  275--341, 2009{\natexlab{b}}.

\bibitem[Faliszewski et~al.(2011)Faliszewski, Hemaspaandra, Hemaspaandra, and
  Rothe]{fal-hem-hem-rot:j:single-peaked-preferences}
P.~Faliszewski, E.~Hemaspaandra, L.~Hemaspaandra, and J.~Rothe.
\newblock The shield that never was: {Societies} with single-peaked preferences
  are more open to manipulation and control.
\newblock \emph{Information and Computation}, 209\penalty0 (2):\penalty0
  89--107, 2011.

\bibitem[Faliszewski et~al.(2014)Faliszewski, Hemaspaandra, and
  Hemaspaandra]{fal-hem-hem:j:single-peaked-nearly}
P.~Faliszewski, E.~Hemaspaandra, and L.~Hemaspaandra.
\newblock The complexity of manipulative attacks in nearly single-peaked
  electorates.
\newblock \emph{Artificial Intelligence}, 207:\penalty0 69--99, 2014.

\bibitem[Faliszewski et~al.(2017)Faliszewski, Skowron, and
  Talmon]{fal-sko-tal:c:bribery-success}
P.~Faliszewski, P.~Skowron, and N.~Talmon.
\newblock Bribery as a measure of candidate success: {Complexity} results for
  approval-based multiwinner rules.
\newblock In \emph{Proceedings of the 16th International Conference on
  Autonomous Agents and Multiagent Systems}, pages 6--14. IFAAMAS, May 2017.

\bibitem[Garey and Johnson(1979)]{gar-joh:b:int}
M.~Garey and D.~Johnson.
\newblock \emph{Computers and Intractability: A Guide to the Theory of
  NP-Completeness}.
\newblock {W. H. Freeman and Company}, 1979.

\bibitem[Garey et~al.(1976)Garey, Johnson, and
  Stockmeyer]{gar-joh-sto:j:simplified-np-completeness}
M.~Garey, D.~Johnson, and L.~Stockmeyer.
\newblock Some simplified {N}{P}-complete graph problems.
\newblock \emph{Theoretical Computer Science}, 1:\penalty0 237--267, 1976.

\bibitem[Gavenciak et~al.(2018)Gavenciak, Knop, and
  Kouteck{\'{y}}]{gav-kno-kou:c:ilp-guidebook}
T.~Gavenciak, D.~Knop, and M.~Kouteck{\'{y}}.
\newblock Integer programming in parameterized complexity: {Three} miniatures.
\newblock In \emph{Proceedings of the 13th International Symposium on
  Parameterized and Exact Computation}, volume 115 of \emph{LIPIcs}, pages
  21:1--21:16. Schloss Dagstuhl~-- Leibniz-Zentrum f\"{u}r Informatik, 2018.

\bibitem[G{\"{u}}ney(2017)]{gue:c:mixed-integer-linear-program-for-election-campaign-optimization}
E.~G{\"{u}}ney.
\newblock A mixed integer linear program for election campaign optimization
  under {D'Hondt} rule.
\newblock In \emph{Operations Research Proceedings 2017: Selected Papers of the
  Annual International Conference of the German Operations Research Society},
  pages 73--79. Springer, September 2017.

\bibitem[Harrenstein et~al.(2022)Harrenstein, Lackner, and
  Lackner]{har-lac-lac:j:mathematical-analysis-of-election-system-by-gottlob-frege}
P.~Harrenstein, M.-L. Lackner, and M.~Lackner.
\newblock A mathematical analysis of an election system proposed by {Gottlob}
  {Frege}.
\newblock \emph{Erkenntnis}, 87:\penalty0 2609--2644, 2022.

\bibitem[Hemaspaandra et~al.(2007)Hemaspaandra, Hemaspaandra, and
  Rothe]{hem-hem-rot:j:destructive-control}
E.~Hemaspaandra, L.~Hemaspaandra, and J.~Rothe.
\newblock Anyone but him: {T}he complexity of precluding an alternative.
\newblock \emph{Artificial Intelligence}, 171\penalty0 (5--6):\penalty0
  255--285, 2007.

\bibitem[Jansen et~al.(2013)Jansen, Kratsch, Marx, and
  Schlotter]{jan-kra-mar-sch:j:bin-packing-revisited}
K.~Jansen, S.~Kratsch, D.~Marx, and I.~Schlotter.
\newblock Bin packing with fixed number of bins revisited.
\newblock \emph{Journal of Computer and System Sciences}, 79\penalty0
  (1):\penalty0 39--49, 2013.

\bibitem[Karp(1972)]{kar:b:reducibilities}
R.~Karp.
\newblock Reducibility among combinatorial problems.
\newblock In R.~Miller and J.~Thatcher, editors, \emph{Complexity of Computer
  Computations}, pages 85--103. Plenum Press, 1972.

\bibitem[Lackner and
  Skowron(2018)]{lac-sko:c:approval-based-multiwinner-rules-and-strategic-voting}
M.~Lackner and P.~Skowron.
\newblock Approval-based multi-winner rules and strategic voting.
\newblock In \emph{Proceedings of the 27th International Joint Conference on
  Artificial Intelligence}, pages 340--436. ijcai.org, July 2018.

\bibitem[Lau{\ss}mann et~al.(2024)Lau{\ss}mann, Rothe, and
  Seeger]{lau-rot-see:c:apportionment-with-thresholds}
C.~Lau{\ss}mann, J.~Rothe, and T.~Seeger.
\newblock Apportionment with thresholds: {Strategic} campaigns are easy in the
  top-choice but hard in the second-chance mode.
\newblock In \emph{Proceedings of the 49th International Conference on Current
  Trends in Theory and Practice of Computer Science}, pages 355--368.
  Springer-Verlag {\it Lecture Notes in Computer Science~\#14519}, Feb. 2024.

\bibitem[{Lenstra, Jr.}(1983)]{len:j:integer-fixed}
H.~{Lenstra, Jr.}
\newblock Integer programming with a fixed number of variables.
\newblock \emph{Mathematics of Operations Research}, 8\penalty0 (4):\penalty0
  538--548, 1983.

\bibitem[Neveling and
  Rothe(2020)]{nev-rot:c:complexity-of-cloning-candidates-in-multiwinner-elections}
M.~Neveling and J.~Rothe.
\newblock The complexity of cloning candidates in multiwinner elections.
\newblock In \emph{Proceedings of the 19th International Conference on
  Autonomous Agents and Multiagent Systems}, pages 922--930. IFAAMAS, May 2020.

\bibitem[Oelbermann and
  Pukelsheim(2020)]{oel-puk:study:european-elections-of-may-2019}
K.-F. Oelbermann and F.~Pukelsheim.
\newblock The {European} {Elections} of {May} 2019: {Electoral} systems and
  outcomes, 2020.
\newblock Study for the European Parliamentary Research Service.

\bibitem[Ostapenko et~al.(2012)Ostapenko, Ostapenko, Belyaeva, and
  Stupnitskaya]{ostapenko2012mathematical}
V.~Ostapenko, O.~Ostapenko, E.~Belyaeva, and Y.~Stupnitskaya.
\newblock Mathematical models of the battle between parties for electorate or
  between companies for markets.
\newblock \emph{Cybernetics and Systems Analysis}, 48\penalty0 (6):\penalty0
  814--822, 2012.

\bibitem[Pellicer and Wegner(2014)]{pel-weg:j:effects-of-legal-thresholds}
M.~Pellicer and E.~Wegner.
\newblock The mechanical and psychological effects of legal thresholds.
\newblock \emph{Electoral Studies}, 33:\penalty0 258--266, 2014.

\bibitem[Peters(2018)]{peters2018proportionality}
D.~Peters.
\newblock Proportionality and strategyproofness in multiwinner elections.
\newblock In \emph{Proceedings of the 17th International Conference on
  Autonomous Agents and Multiagent Systems}, pages 1549--1557. IFAAMAS, May
  2018.

\bibitem[Pukelsheim(2017)]{puk:b-2nd-edition:proportional-representation}
F.~Pukelsheim.
\newblock \emph{Proportional Representation: Apportionment Methods and Their
  Applications}.
\newblock Springer, second edition, 2017.

\bibitem[Tideman(1987)]{tid:j:independence-of-clones}
N.~Tideman.
\newblock Independence of clones as a criterion for voting rules.
\newblock \emph{Social Choice and Welfare}, 4\penalty0 (3):\penalty0 185--206,
  1987.

\bibitem[Xia(2012)]{xia:c:margin-of-victory}
L.~Xia.
\newblock Computing the margin of victory for various voting rules.
\newblock In \emph{Proceedings of the 13th ACM Conference on Electronic
  Commerce}, pages 982--999. ACM Press, June 2012.

\bibitem[Yang(2019)]{yan:c:multiwinner-control}
Y.~Yang.
\newblock Complexity of manipulating and controlling approval-based multiwinner
  voting.
\newblock In \emph{Proceedings of the 28th International Joint Conference on
  Artificial Intelligence}, pages 637--643. ijcai.org, Aug. 2019.

\end{thebibliography}
\end{document}